\documentclass[reqno,12pt,letterpaper]{amsart}
\usepackage{amsmath,amssymb,amsthm,graphicx,mathrsfs,url}
\usepackage[usenames,dvipsnames]{color}
\usepackage{amsxtra}
\usepackage{subcaption}
\usepackage{graphicx}
\usepackage{wasysym}
\usepackage{bm}
\usepackage{mathtools}

\usepackage[colorlinks=true,linkcolor=Red,citecolor=Green]{hyperref}
\def\arXiv#1{\href{http://arxiv.org/abs/#1}{arXiv:#1}}

\def\?[#1]{\textbf{[#1]}\marginpar{\Large{\textbf{??}}}}

\def\smallsection#1{\smallskip\noindent\textbf{#1}.}
\let\epsilon=\varepsilon 

\newcommand{\RR}{{\mathbb R}}
\newcommand{\NN}{{\mathbb N}}
\newcommand{\CC}{{\mathbb C}}
\newcommand{\ZZ}{{\mathbb Z}}
\newcommand{\C}{{\mathbb C}}
\newcommand{\Z}{{\mathbb Z}}

\newcommand\tinyvarhexagon{\vcenter{\hbox{\scalebox{0.7}{$\varhexagon$}}}}

\newtheorem{theo}{Theorem}
\newtheorem{prop}{Proposition}[section]	
\newtheorem{defi}[prop]{Definition}

\newtheorem{lemm}[prop]{Lemma}

\newtheorem{rem}{Remark}
\numberwithin{equation}{section}
\usepackage{float}

\let\Im=\Imag

\let\Re=\Real
\DeclareMathOperator{\sgn}{sgn}

\DeclareMathOperator{\supp}{supp}

\DeclareMathOperator{\WF}{WF}
\DeclareMathOperator{\tr}{tr}

\def\indic{\operatorname{1\hskip-2.75pt\relax l}}

\newcommand\reallywidehat[1]{\arraycolsep=0pt\relax%
\begin{array}{c}
\stretchto{
  \scaleto{
    \scalerel*[\widthof{\ensuremath{#1}}]{\kern-.5pt\bigwedge\kern-.5pt}
    {\rule[-\textheight/2]{1ex}{\textheight}} 
  }{\textheight} %
}{0.5ex}\\           
#1\\                 
\rule{-1ex}{0ex}
\end{array}
}

\title[Density of states and Delocalization.]{Density of states and Delocalization for discrete magnetic random Schr\"odinger operators}

\author{Simon Becker}
\email{simon.becker@damtp.cam.ac.uk}
\address{DAMTP, University of Cambridge, Wilberforce Rd, Cambridge CB3 0WA, UK}

\author{Rui Han}
\email{rhan@lsu.edu}
\address{Department of Mathematics, Louisiana State University, Baton Rouge, LA 70803, US}

\begin{document}

\begin{abstract}
We study discrete magnetic random Schr\"odinger operators on the square and honeycomb lattice. 
For the non-random magnetic operator on the hexagonal lattice with any rational magnetic flux, we show that the middle two dispersion surfaces exhibit Dirac cones. 
We then derive an asymptotic expansion for the density of states on the honeycomb lattice for oscillations of arbitrary rational magnetic flux. 
This allows us, as a corollary, to rigorously study the quantum Hall effect and conclude dynamical delocalization close to the conical point under disorder. 
We obtain similar results for the discrete random Schr\"odinger operator on the $\mathbb Z^2$-lattice with weak magnetic fields, close to the bottom and top of its spectrum. 
\end{abstract}

\maketitle

\section{Introduction and statement of results}
In this article, we study discrete random Schr\"odinger operators, the tight-binding limits of continuous random Schr\"odinger operators, under weak disorder in weak magnetic fields on the $\mathbb Z^2$ lattice $ \Lambda_{{\scriptscriptstyle{ \blacksquare}}}$ and in addition for magnetic fluxes close to rationals on the honeycomb lattice $\Lambda_{\tinyvarhexagon}$:
\begin{equation*}
\begin{split}
(H^h_{{\scriptscriptstyle{ \blacksquare}},\lambda,\omega} u)(\gamma)&:=-\frac{1}{4} \Bigg( e^{i h \gamma_2/2} u(\gamma + \vec{b}_1)+e^{-i h \gamma_2/2} u(\gamma - \vec{b}_1) \\
& \qquad + e^{-i h \gamma_1/2} u(\gamma+\vec{b}_2) + e^{i h \gamma_1/2} u (\gamma-\vec{b}_2) \Bigg) + \lambda V_{\omega}(\gamma)u(\gamma) \\
(H^h_{\tinyvarhexagon,\lambda,\omega} u)(v)&:=-\frac{1}{3} \left(\sum_{\vec{e} \in \mathcal{E}, i(\vec{e})=v } e^{-i  A_{\vec{e}}} u(t(\vec{e})) +  \sum_{\vec{e} \in \mathcal{E}, t(\vec{e})=v } e^{iA_{\vec e}} u(i(\vec{e})) \right) + \lambda V_{\omega}(v) u(v),
\end{split}
\end{equation*}  
where $V_{\omega}$ is an i.i.d. random potential on the respective lattice $\Lambda$.
For precise definitions of these operators, we refer to Section \ref{sec:defH}.

The spectral properties of the discrete magnetic Laplacian (DML) on $\mathbb Z^2$, and of the almost Mathieu operator (AMO), have been extensively studied over the past forty years, see for instance a survey \cite{JMsurvey} and some recent advancements \cite{AYZ17,JL18,JK19}. Significant progress on the location of the spectrum has been made for magnetic Schr\"odinger operators using semiclassical analysis \cite{HS0,HS1,HS2,W94}.
In two preceding articles \cite{BHJ17,BZ}, by the authors, this study was extended to spectral properties and the density of states (DOS) of the magnetic Schr\"odinger operator on the honeycomb lattice -but without disorder. It was shown in \cite[Theorem $1$]{BZ} that the DOS for the magnetic Schr\"odinger operator on the honeycomb quantum graph- close to the conical point- is concentrated at so-called relativistic Landau levels. 

The spectral analysis in \cite{BHJ17} showed that for the DML on the hexagonal lattice, close to the conical point, there is no point spectrum, as the analogy to the magnetic two-dimensional Dirac operator suggests. Instead, the spectrum of the DML on the honeycomb lattice is either absolutely continuous (a.\@c.\@) band spectrum or singular continuous (s.\@c.\@) and a Cantor set of Lebesgue measure zero, depending on the arithmetic properties of the magnetic flux through a single honeycomb.

Next let us introduce our results. We start with the non-random operator on the hexagonal lattice $H^{h}_{\tinyvarhexagon,\lambda=0}$.
The part of the energy spectrum of graphene, modeled here by the discrete operator $H^{h}_{\tinyvarhexagon,\lambda=0}$ that is relevant for most of its remarkable physical properties, is the energy spectrum close to the conical points, the so-called \emph{Dirac points} at energy zero, see Fig.\@ \ref{fig:hex}. 
The existence of Dirac points for the tight-binding graphene model in the absence of magnetic field is known since \cite{W47}.
In the absence of magnetic fields, the operator can be reduced to a $2\times 2$ matrix via Floquet-Bloch theory.
Hence the only two dispersion surfaces can be computed explicitly, whence conical touching of the two surfaces is evident.
It is natural to ask the question if Dirac points still exist for arbitrary rational magnetic flux, where the operator is still periodic.
Indeed, for flux $h=2\pi p/q$, the operator can be reduced to a $2q\times 2q$ matrix. 
The dispersion surfaces thus have to be analyzed implicitly and hence making it much harder to prove conical structures.
Our first result is to prove the existence of Dirac cones at energy zero for the tight-binding model for {\it any} rational magnetic flux.
\begin{theo}\label{thm:Dirac}
For any rational flux $h=2\pi \frac{p}{q} \in 2\pi\, \mathbb{Q}$, the operator
$H^{h}_{\tinyvarhexagon,\lambda=0}$ possesses Dirac points at energy zero.
\end{theo}

Using the conical structures as a starting point, we are able to carry out the semi-classical analysis and obtain the expansion of the density of states (DOS) near the energy zero for flux $2\pi p/q+h$ with $h$ being the semi-classical parameter, see Theorems \ref{theo1}\footnote{Theorem \ref{theo1} actually proves the expansion of DOS for the operator with disorder.} and \ref{theol}. 
This in particular allows us to prove the localization of the spectrum in Landau bands near the zero energy, characterized by the Bohr-Sommerfeld condition, and the existence of spectral gaps between any two consecutive Landau bands. Our framework follows \cite{HS20}, but uses independent arguments for the derivation of the density of states and the presence of spectral gaps in between Landau bands. 
In addition to the study of the discrete magnetic Laplacian on the honeycomb lattice, we also derive the expansion of DOS for the operator $H^h_{{\scriptscriptstyle{ \blacksquare}},\lambda,\omega}$ on the $\Z^2$ lattice with small flux near the top and bottom of the spectrum, which is included in Theorem \ref{theo1}.

By combining the expansion of DOS with the St\v{r}eda formula, we are able to compute the Hall conductivity explicitly in each of the aforementioned spectral gaps for the non-random operators $H^{h}_{\tinyvarhexagon,\lambda=0}$ and $H^h_{{\scriptscriptstyle{ \blacksquare}},\lambda=0}$, thus giving a rigorous derivation of the Quantum Hall effect (QHE). 
We then argue using the index-theoretic formulation that the Hall conductivity is invariant under a random perturbation in the spectral gaps between any two consecutive disorder-broaden Landau bands. The study of the quantum Hall effect of the continuous Laplacian in a homogeneous magnetic field is much simpler, as the (infinitely-degenerate) eigenfunctions are fully explicit and so all computations can be done analytically. In contrast to this, the discrete magnetic Laplacian, does not have point spectrum and closed-form describing it are also not available. This is a major difficulty in the discrete setting, which can be partly overcome by gap-labelling methods techniques as in \cite{AEG}. However, we would like to emphasize that such methods are usually not quantitative in the sense that they do not specify the Hall conductivity at prescribed energies. From our refined study of the density of states with error bounds, we are able to solve this problem and get precise information on the Hall conductivity in the gaps between Landau bands that are quantitative. 
For the cleanness of the presentation, we present below the QHE for small magnetic fields, and refer the readers to Theorem \ref{theol} for $H^h_{\tinyvarhexagon,\lambda,\omega}$ with fluxes close to rationals.

\begin{prop}[QHE under weak disorder; Small magnetic fields]
\label{QHEran}
For sufficiently small magnetic flux $h>0,$ there are spectral gaps between disorder-broadened Landau bands up to some magnetic-dependent disorder parameter $\lambda_0(h)> 0$. 
In the spectral gap between two consecutive disorder-broadened Landau bands $B^h_{\tinyvarhexagon,\lambda,n}$ and $B^h_{\tinyvarhexagon,\lambda,n+1}$, the Hall conductivity $c_H$ with Fermi energy $\mu$ is quantized with its value given below.
\begin{align*}
c_H(H^h_{\tinyvarhexagon,\lambda,\omega},\mu) &=  \frac{2n+1}{2\pi}, \text{ with }-N_{\scriptscriptstyle\tinyvarhexagon}(h,\lambda_0) \le n \le N_{\scriptscriptstyle\tinyvarhexagon}(h,\lambda_0)\\
c_H(H^h_{{\scriptscriptstyle{ \blacksquare}},\lambda,\omega},\mu) &= \frac{n}{2\pi},  \text{ with } 1 \le n \le N_{\scriptscriptstyle\blacksquare}(h,\lambda_0)
\end{align*} 
\end{prop}


Using the jump of Hall conductivity in each disorder-broadened Landau band, we then show that the discrete magnetic random Schr\"odinger operators undergo metal/insulator transitions, using the framework of Germinet-Klein \cite{GK} and Klein-Germinet-Schenker \cite{GKS}. More precisely, we prove the existence of (at least one) mobility edge near the Landau levels. 
Again, we only present the small magnetic fields case here, and refer the readers to Theorem \ref{theol} for $H^h_{\tinyvarhexagon,\lambda,\omega}$ with perturbations of rational fluxes.
\begin{theo}[Dyn. Delocalization; Small fields]
\label{Deloc}
Under the same assumptions as Proposition \ref{QHEran}, there exists in each disorder-broadened Landau band (at least) one energy that belongs to the region of dynamical delocalization. 
\end{theo}

The paper is structured as follows: Section 2 serves as preliminary and background, the study of DOS is presented in Section 3, QHE is studied in Section 4, dynamical delocalization is proved in Section 5, the proof of Theorem \ref{thm:Dirac} is presented in Section 6, and finally the semiclassical analysis together with the proofs of Theorems \ref{theo1} and \ref{theol} are presented in Section 7.

\label{s:intr}

\smallsection{Acknowledgements} 
The first author is supported by the UK Engineering and Physical Sciences Research Council (EPSRC) grant EP/L016516/1 for the University of Cambridge Centre for Doctoral Training. The second author is partially supported by NSF-DMS 2053285.
The Cambridge Centre for Analysis is gratefully acknowledged (S.B). The first author is grateful to Gian-Michele Graf for bringing the St\v{r}eda formula to his attention. Helpful remarks and discussions with Svetlana Jitomirskaya, Hermann Schulz-Baldes and Maciej Zworski are gratefully acknowledged as well.

\smallsection{Notation}
$B_x(r)$ is the ball of radius $r$ centred at $x.$ We write $ f_\alpha =                                                                  
\mathcal O_\alpha ( g )_H $ for $ \| f \|_H \leq C_\alpha g $ and
$ f = \mathcal O ( h^\infty )_H  $ means that for any $ N $ there exists
$ C_N $ such that $ \|f \|_H \leq C_N h^N $. 
We write $\langle x\rangle:=\sqrt{1+|x|^2}$. $\mathcal U(\mathcal H)$ are the unitary operators on a Hilbert space $\mathcal H.$
The symbol class $\mathcal S_{h_0}$, of possibly matrix-valued symbols, is defined as 
\[\mathcal S_{h_0}:= \left\{ a(\bullet,h) \in C^{\infty}(T^*\RR):  \forall \alpha \in \mathbb N_0^2 \ \exists\, C_{\alpha}>0 \ \forall h \in [0,h_0]: \ \vert \partial^{\alpha}a(\bullet,h) \vert \le C_{\alpha} \right\}.\]
We write $a \sim \sum_{j=0}^{\infty} a_jh^j$ to denote an asymptotic expansion of symbols, cf. \cite[$4.4.2$]{ev-zw} where $a_j \in \mathcal S$, with 
\[ \mathcal S:=\left\{ a \in C^{\infty}(T^*\RR); \forall \alpha \in \mathbb N_0^2 \ \exists\, C_{\alpha}>0: \ \vert \partial^{\alpha}a \vert \le C_{\alpha} \right\}\] 
and denote the class of symbols allowing such an expansion by $\mathcal S^{\text{cl}}.$
The standard basis vectors of $\ell^2(\mathbb Z^2)$ are for $\gamma \in \mathbb Z^2$ denoted by $\delta_{\gamma}:=\left(\delta_{\gamma,\gamma'} \right)_{\gamma'}$ and occasionally by $\vec{e}_i$ if the Hilbert space is finite-dimensional. $\mathcal L(X,Y)$ are the bounded linear operators between normed spaces $X,Y$. $\mathbb E$ and $\operatorname{Var}$ denote expectation and variance.
The semiclassical Weyl quantization of a symbol $a \in \mathcal S_h(T^*\mathbb R)$ is for suitable functions $u$ defined as
\[(\operatorname{Op}_h^{\text{w}}(a)u)(x):=(a^{\text{w}}(x,hp_x,h)u)(x):= \frac{1}{2\pi h} \int_{\mathbb R} \int_{\mathbb R} e^{ \frac{i}{h} \langle x-y, \xi \rangle} a\left( \tfrac{x+y}{2},\xi,h \right) u(y) \ dy \ d\xi.  \] 
Here, $p_x:=-i\frac{d}{dx}.$ Conversely, we write $ \sigma\left(\operatorname{Op}_h^w(a)\right):=a$ to denote the Weyl symbol of a $\Psi$DO and $\sigma_0\left(\operatorname{Op}_h^w(a)\right)$ for the principal symbol. Analogously, higher order symbols are denoted by $\sigma_k$, respectively. The semiclassical wavefront set is denoted by $\operatorname{WF}_h$, see \cite[Sec.\@8.4]{ev-zw}.
We also write $\mathbb Z^2_*:=(2\pi \ZZ)^2.$
For a subset $I \subset \mathbb R$ we denote by $\oint_{I}$ a contour integral over a path in the complex plane that encloses $I$ sufficiently close. 

\medskip

The spectrum of an operator $T$ is denoted by $\Sigma(T).$ We sometimes use the convention $\hbar:=\frac{h}{2\pi}$ where $h$ is the \emph{magnetic flux} (thus this notation should not be confused with Planck's constant). The $p$-th Schatten class is denoted by $\mathcal L^p.$ The symplectic form on $\mathbb{R}^2$ is denoted by $\sigma_{\text{symp}}(\gamma,\delta):=\gamma_1\delta_{2}-\delta_1\gamma_2.$ Finally, we use Wirtinger derivatives $D_{ z}:=\frac{1}{2}(\partial_x - i \partial_y)$ and $ D_{\overline{z}}:=\frac{1}{2}(\partial_x + i \partial_y)$ where we recall that $D_zf$ is nothing but the derivative of a holomorphic function $f$. In particular, holomorphic functions satisfy $ D_{\overline{z}}f=0$ by the Cauchy-Riemann equations. $\mathscr S(\mathbb Z^2)$ are the sequences that decay faster than any polynomial power. We also write $\mathscr S(\RR^n)$ or $\mathscr S(\CC^n)$ for the Schwartz functions on $\RR^n$ or $\CC^n.$
We also define for one of the two lattices $\Lambda$ we study in this article, the truncated sets
\begin{equation}
\begin{split}
\label{lattice}
\Lambda_{L}&:= \Big\{ y \in \mathbb{R}^2; y = \gamma_1\vec{b}_1+\gamma_2\vec{b}_2+[y]\text{ for } \gamma \in \left\{-L,...,L \right\}^2 \text{ and } [y] \in W_{\Lambda} \Big\}
\end{split}
\end{equation}
where $\vec{b}_1$ and $\vec{b}_2$ are the basis vectors of the lattice and $W_{\Lambda}$ a fundamental domain.

\section{Lattices and discrete random Schr\"odinger operators}
\begin{figure*}[t!]
    \centering
    \begin{subfigure}[t]{0.5\textwidth}
        \centering
        \includegraphics[height=2.0in]{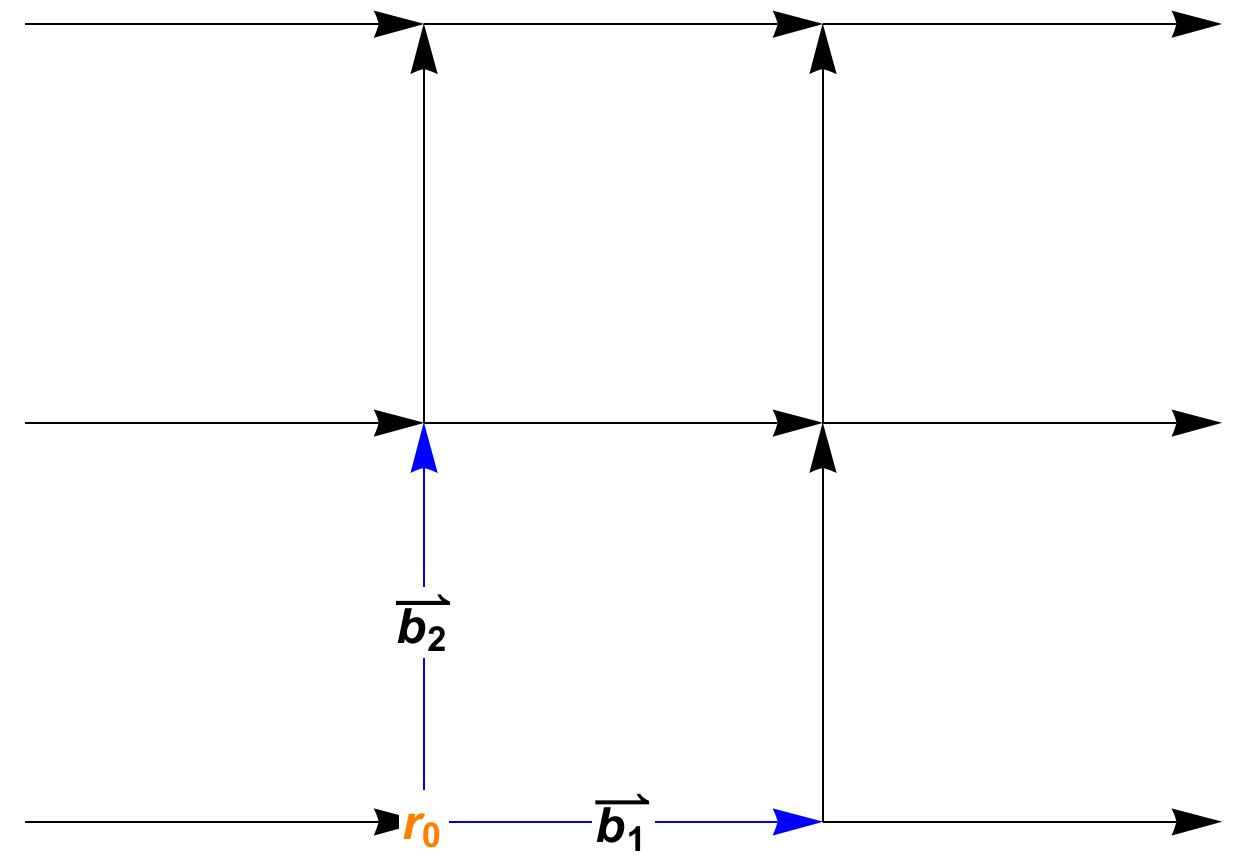}
        \caption{The square lattice $\Lambda_{\blacksquare}.$}
        \label{fig:a}
    \end{subfigure}%
    ~ 
    \begin{subfigure}[t]{0.5\textwidth}
        \centering
        \includegraphics[height=2.0in]{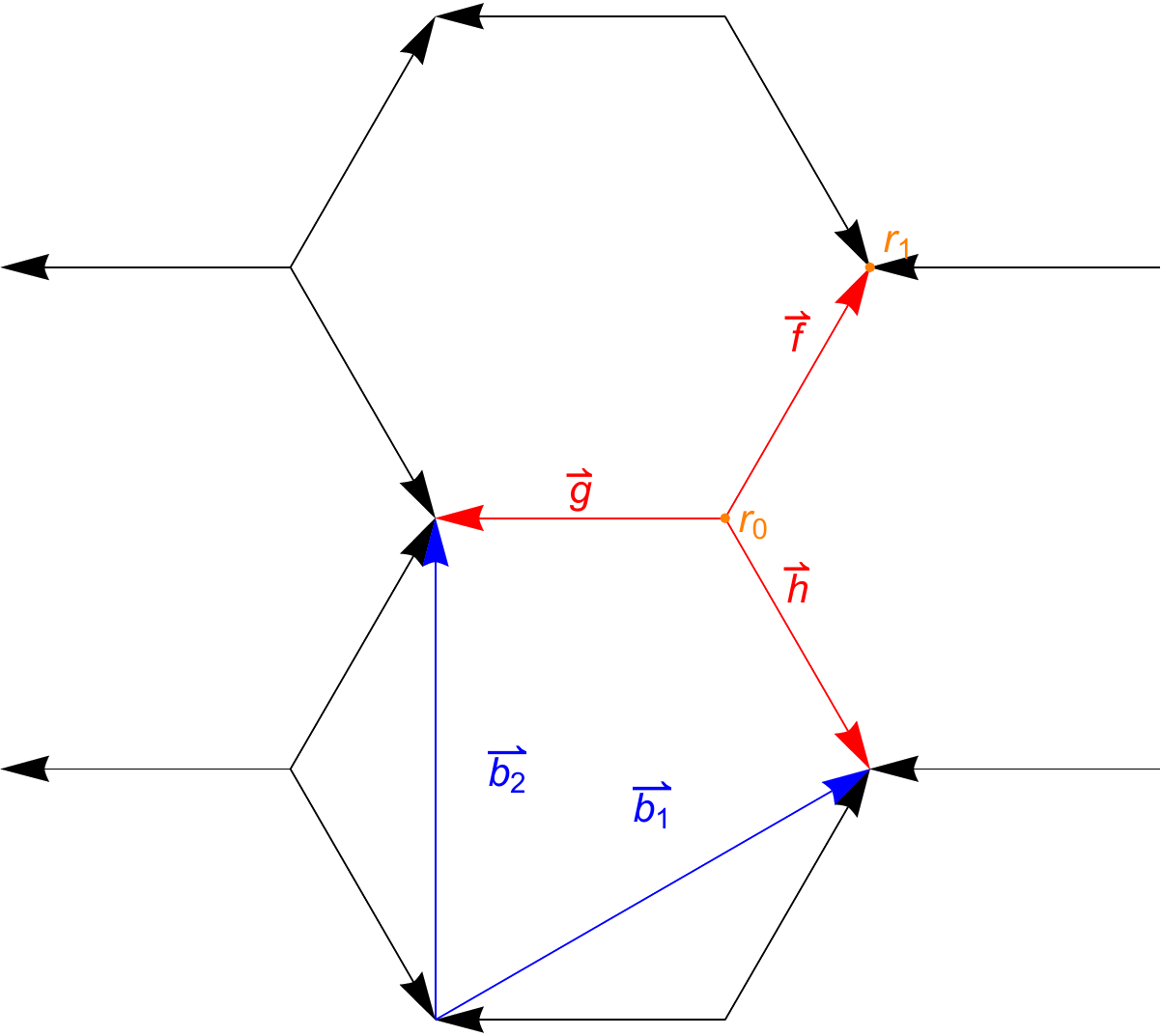}
        \caption{The hexagonal lattice $\Lambda_{\hexagon}.$}
        \label{fig:2}
    \end{subfigure}
    \caption{Fundamental cells of lattices.}
\end{figure*}
\subsection{Geometry of lattices}
\label{sec:lattices}

 \medskip

\smallsection{The $\mathbb Z^2$ lattice $ \blacksquare$, see Fig.\@ \ref{fig:a}}
The square lattice $\Lambda_{{\scriptscriptstyle{ \blacksquare}}}:=\mathbb Z^2$ is spanned by basis vectors $\vec{b}_{\scriptscriptstyle{ \blacksquare},1}:=(1,0), \ \vec{b}_{\scriptscriptstyle{ \blacksquare},2}:=(0,1)$ and its fundamental cell $W_{\Lambda_{{\scriptscriptstyle{ \blacksquare}}}}$ consists of just the vertex $r_0:=(0,0)$.
Although we do not study operators on the associated graph, we also introduce the set of edges $\mathcal E_{{\scriptscriptstyle{ \blacksquare}}}$ on the square graph consisting of the two edges
\begin{equation}
\label{eq:fgh}
\begin{split}
\vec{f}_{\uparrow} & :=\operatorname{conv}\left(\left\{r_0,(1,0)\right\}\right) \ \backslash \ \left\{r_0,(1,0)\right\},  \\ 
\vec{f}_{\rightarrow} & :=\operatorname{conv}\left(\left\{r_0,\left(0,1\right)\right\}\right) \ \backslash \  \left\{r_0,\left(0,1\right)\right\} \\
\end{split}
\end{equation}
and translations thereof by basis vectors $\vec{b}_{\scriptscriptstyle{ \blacksquare},1},\vec{b}_{\scriptscriptstyle{ \blacksquare},2}$, where $\operatorname{conv}$ denotes the convex hull. To orient the graph, we also define a map $i:\mathcal E_{{\scriptscriptstyle{ \blacksquare}}} \rightarrow  \Lambda_{{\scriptscriptstyle{ \blacksquare}}}$ by $i(\vec{f}_{\uparrow}):=i(\vec{f}_{\rightarrow}):=r_0$ and extend it to all edges by translation
\[i(\vec{f}_{\uparrow}+\gamma) =i(\vec{f}_{\rightarrow}+\gamma)=r_0 + \gamma \text{ for } \gamma \in \mathbb Z^2. \] 

 \medskip

 Let us now turn to the hexagonal lattice: \newline
 \smallsection{The hexagonal lattice $\varhexagon$, see Fig.\@ \ref{fig:2}}
The hexagonal lattice $\Lambda_{\tinyvarhexagon}$ is obtained by translating its fundamental cell $W_{\Lambda_{\tinyvarhexagon}}$, consisting of vertices
\begin{equation}
\label{eq:r0r}
r_0:=(0,0) , \ \ \  r_1:=\left(\tfrac{1}{2},\tfrac{\sqrt{3}}{2}\right)
\end{equation}
along the basis vectors of the lattice.
The basis vectors are
\begin{align}
\vec{b}_{\tinyvarhexagon,1}:= \left(\tfrac{3}{2}, \tfrac{\sqrt{3}}{2} \right) \  \text{ and } \ 
\vec{b}_{\tinyvarhexagon,2}:= \left(0,\sqrt{3}\right).
\end{align}

As in the case of the $\mathbb Z^2$ lattice, we also introduce auxiliary edges
\begin{equation}
\begin{split}
\vec{f} & :=\operatorname{conv}\left(\left\{r_0,r_1\right\}\right) \ \backslash \ \left\{r_0,r_1\right\},  \\ 
\vec{g}& :=\operatorname{conv}\left(\left\{r_0,\left(-1,0\right)\right\}\right) \ \backslash \left\{r_0,\left(-1,0\right)\right\} , \\
\vec{h}& :=\operatorname{conv}\left(\left\{r_0,\left(\tfrac{1}{2},-\tfrac{\sqrt{3}}{2}\right)\right\}\right) \ \backslash \ \left\{r_0,\left(\tfrac{1}{2},-\tfrac{\sqrt{3}}{2}\right)\right\},
\end{split}
\end{equation}
and define the set of all edges $\mathcal E_{\tinyvarhexagon}$ as the set of all translates of these three edges along the basis vectors $\vec{b}_{\tinyvarhexagon,1},\vec{b}_{\tinyvarhexagon,2}$ of the hexagonal lattice.

We call translates of $r_0$ by basis vectors $\vec{b}_{\tinyvarhexagon,1},\vec{b}_{\tinyvarhexagon,2}$ \emph{initial vertices} $\Lambda_{\tinyvarhexagon}^i$ whereas translates of $r_1$ will be referred to as \emph{terminal vertices} $\Lambda_{\tinyvarhexagon}^t$. Moreover, we consider maps $i:\mathcal E_{\tinyvarhexagon} \rightarrow \Lambda_{\tinyvarhexagon}$ and $t:\mathcal E_{\tinyvarhexagon} \rightarrow \Lambda_{\tinyvarhexagon}$ that map edges to the respective initial or terminal vertex they contain. 

In the sequel, we will use the isomorphism $\ell^2(\Lambda_{\tinyvarhexagon}) \simeq \ell^2(\ZZ^2; \CC^2)$ as the honeycomb has two basis vectors and two vertices in its fundamental domain. More generally, any lattice with $\Lambda$ spanned by two basis vectors with $n$ vertices in its fundamental domain satisfies $\ell^2(\Lambda) \simeq \ell^2(\ZZ^2;\CC^n).$

\subsection{Discrete random Schr\"odinger operators}\label{sec:defH}
We consider a constant magnetic field. The vector potential $\textbf{A} $ is a one form on $ \RR^2 $ and the magnetic field is given by $ \textbf{B} = d \textbf{A} $. For homogeneous magnetic fields
\begin{equation}
\label{eq:magfield}
\textbf{B} : =B \ dx_1 \wedge dx_2
\end{equation}
we can choose a symmetric gauge for the vector potential $ \mathbf A $ such that
\begin{equation}
\label{eq:gauge}
\textbf{B} = d\textbf{A} , \ \ \ \textbf{A}=\tfrac12 {B} \left(-x_2 \, dx_1 +x_1 \, dx_2 \right).
\end{equation}
The discrete magnetic Laplacians (DMLs) with single-site disorder are then defined as follows: First, we take the scalar potential $A_{\vec{e}} \in C^{\infty}(\vec{e})$ along edges $\vec{e} =e_1 \ dx_1^* + e_2\ dx_2^*$ of the respective graph, where $dx_j (dx_i^*)=\delta_{i,j}$ is defined by evaluating the 1-form on the graph along the vector field generated by the respective edge $\vec{e}$:
\begin{equation}
\begin{split}
\label{eq:edgeA}
A_{\vec{e} }(t)&:= \textbf{A}\left(i(\vec{e} )+t\vec{e}  \right)\left(e_1 \ dx_1^* + e_2 \ dx_2^* \right) =\textbf{A}\left(i(\vec{e} )\right)\left(e_1 \ dx_1^* + e_2\ dx_2^* \right).
\end{split}
\end{equation}
The quantities $A_{\vec{e}}$ on the
square lattice are given by
\begin{equation}
\label{eq:squareA}
A_{\vec{f}_{\uparrow}+\gamma_1 \vec{b}_{\scriptscriptstyle{ \blacksquare},1}+\gamma_2\vec{b}_{\scriptscriptstyle{ \blacksquare},2}}=\tfrac{h_{{\scriptscriptstyle{ \blacksquare}}}}{2} \gamma_1 \text{ and } A_{\vec{f}_{\rightarrow}+\gamma_1\vec{b}_{\scriptscriptstyle{ \blacksquare},1}+\gamma_2\vec{b}_{\scriptscriptstyle{ \blacksquare},2}}=-\tfrac{h_{{\scriptscriptstyle{ \blacksquare}}}}{2} \gamma_2
\end{equation}
and the quantities $A_{\vec{e}}$ on the
hexagonal lattice are explicitly given by
\begin{equation}
\label{Vertauscher}
\begin{split}
A_{\vec{f}+\gamma_1 \vec{b}_{\tinyvarhexagon,1}+\gamma_2\vec{b}_{\tinyvarhexagon,2}} &=  \tfrac{h_{\tinyvarhexagon}}{6}(\gamma_1-\gamma_2) , \
A_{\vec{g}+\gamma_1 \vec{b}_{\tinyvarhexagon,1}+\gamma_2\vec{b}_{\tinyvarhexagon,2}}=  \tfrac{h_{\tinyvarhexagon}}{6}(\gamma_1+2\gamma_2) , \text{ and }\\
A_{\vec{h}+\gamma_1 \vec{b}_{\tinyvarhexagon,1}+\gamma_2\vec{b}_{\tinyvarhexagon,2}} &=  -\tfrac{h_{\tinyvarhexagon}}{6}(2\gamma_1+\gamma_2)
\end{split}
\end{equation}
where the magnetic flux for either lattice is defined as 
\begin{equation}
\label{eq:h}
 h_{{\scriptscriptstyle{ \blacksquare}}} := B \text{ and }h_{\tinyvarhexagon}:=\tfrac{B}{ | \vec{b}_1 \wedge \vec{b}_2 |}= \tfrac{3\sqrt{3}} 2 B. 
\end{equation} 
From this point on, we may suppress the dependence on the lattices in some notations if there is no ambiguity or if the results hold for both lattices.

We now define the discrete magnetic random Schr\"odinger operators:
\begin{defi}[Discrete magnetic Schr\"odinger operators]
We define discrete magnetic random Schr\"odinger operators $H^h_{{\scriptscriptstyle{ \blacksquare}}} \in \mathcal L(\ell^2(\Lambda_{{\scriptscriptstyle{ \blacksquare}}}))$ and $H^h_{\tinyvarhexagon} \in \mathcal L(\ell^2(\Lambda_{\tinyvarhexagon}))$ on the square $\scriptscriptstyle\blacksquare$, using \eqref{eq:squareA}, and hexagonal $\scriptscriptstyle\tinyvarhexagon$ lattice, using \eqref{Vertauscher}, respectively
\begin{equation}
\begin{split}
\label{discreter}
(H^h_{{\scriptscriptstyle{ \blacksquare}},\lambda,\omega} u)(\gamma)&:=\frac{1}{4} \Bigg( e^{i h \gamma_2/2} u(\gamma +\vec{f}_{\rightarrow})+e^{-i h \gamma_2/2} u(\gamma - \vec{f}_{\rightarrow}) \\
& \qquad + e^{-i h \gamma_1/2} u(\gamma+\vec{f}_{\uparrow}) + e^{i h \gamma_1/2} u (\gamma-\vec{f}_{\uparrow}) \Bigg) + \lambda V_{\omega}(\gamma)u(\gamma) \\
(H^h_{\tinyvarhexagon,\lambda,\omega} u)(v)&:=\frac{1}{3} \left(\sum_{\vec{e}  \in \mathcal{E}_{\tinyvarhexagon}, i(\vec{e} )=v } e^{-i  A_{\vec{e} }} u(t(\vec{e} )) +  \sum_{\vec{e}  \in \mathcal{E}_{\tinyvarhexagon}, t(\vec{e} )=v } e^{iA_{\vec{e}}} u(i(\vec{e})) \right) + \lambda V_{\omega}(v) u(v) 
\end{split}
\end{equation} 
where the parameter $\lambda>0$ measures the disorder strength. The random potential satisfies $V_{\omega}(v)=\omega(v)$, where $\{\omega(v)\}_{v\in \Lambda}$ is a family of  i.i.d with common probability distribution $\nu$ of compact support on $\mathbb R$. We write $(\Omega, \mathbb P)$ the underlying probability space, and $\mathbb E$ the expectation. 
\end{defi}
We will write $(\Omega, \mathbb P)$ for the underlying probability space, hence $\Omega=\times_{v\in \Lambda} \mathbb R$, and $\mathbb P=\times_{v\in \Lambda} \nu$. 
We define the shifts operators $\{T_{\delta}^{\Omega}\}_{\delta\in \Z^2}$ on $\Omega$ by 
\begin{align}\label{shift_Omega}
T_{\delta}^{\Omega} \omega(v)=\omega(v-\delta_1\vec{b}_1-\delta_2\vec{b}_2).
\end{align}
The sample space $\Omega$ of the configuration space of impurities $(\Omega,\mathbb P)$ is, without loss of generality, assumed to be compact, cf. \cite[p.\@ 372f.]{C94} for details.

We then write $H^h:=H^h_{\lambda=0,\omega}$ for the non-random DML.

\begin{figure}
\includegraphics[height=6cm]{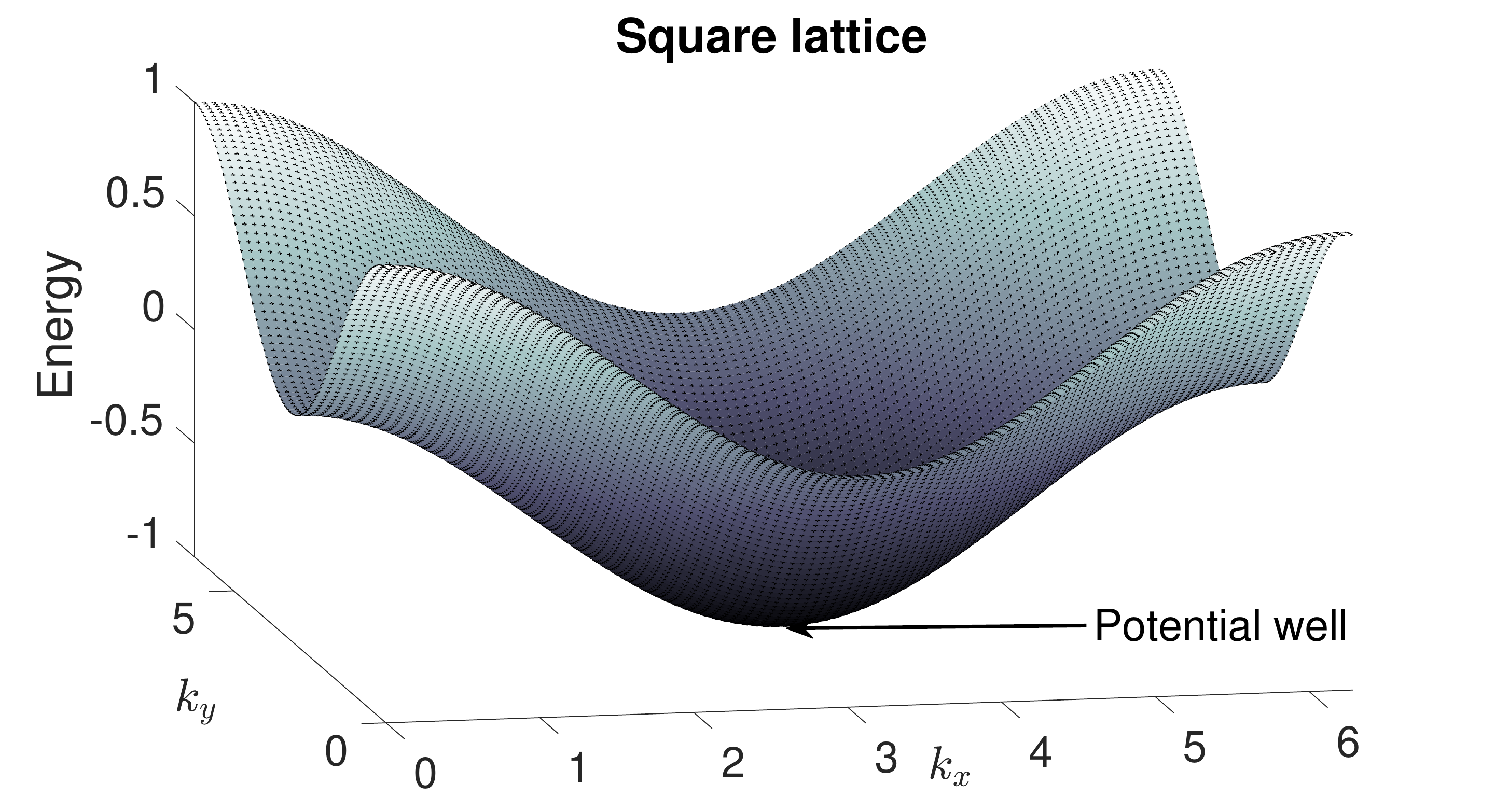} 
\caption{Energy band of the non-magnetic discrete Laplacian on $\Lambda_{{\scriptscriptstyle{ \blacksquare}}}$. The bottom of the spectrum forms a potential well.}
\label{fig:square}
\end{figure}
\subsection{Magnetic translations, regularized traces, and the density of states measure}
We start our analysis by introducing discrete translation operators $T_{\gamma}$ with $\gamma \in \mathbb Z^2$ for $\psi \in \ell^2(\Lambda)$
\begin{equation}
\label{stl}
T_{\gamma}\psi(v):=\psi(v-\gamma_1 \vec{b}_1- \gamma_2 \vec{b}_2).
\end{equation}
The magnetic Schr\"odinger operator $H^h$ does, in general, not commute with standard lattice translations $T_{\gamma},$ but with magnetic translations $T_{\gamma}^h$ instead. These operators and powers of them, do not commute with each other, if $T_{(0,1)}^h$ and $T_{(1,0)}^h$ generate the irrational  ($\hbar \in \mathbb R \backslash \mathbb Q$) rotation algebra.
Magnetic translations $T^{h}_{\gamma} : \ell^2(\Lambda) \rightarrow \ell^2(\Lambda)$ are unitary operators of the form
\begin{equation}
\label{MagTra}
T^{h}_{\gamma} \psi := u^h(\gamma)T _{\gamma}\psi, 
\ \ \ \psi=(\psi_v)_{v \in \Lambda} \in 
\ell^2(\Lambda), \ \vert u^h(\gamma) \vert=1, \ \gamma \in \mathbb{Z}^2
\end{equation}
that satisfy the commutation relation
\begin{equation}
\label{commrelkl}
T_{\gamma}^hT_{\delta}^h = e^{i h \sigma_{\text{symp}}\left( \gamma,\delta \right)  }  T_{\delta}^hT_{\gamma}^h.
\end{equation}

On the square lattice we define magnetic translations as
\begin{equation}
(T^h_{(1,0)} u)(\gamma) = e^{ih/2\gamma_2} u(\gamma-\vec{b}_1) \text{ and }(T^h_{(0,1)} u)(\gamma) = e^{-ih/2\gamma_1} u(\gamma-\vec{b}_2)
\end{equation}
and set then $T^h_{\gamma}:=(T^h_{(1,0)})^{\gamma_1}(T^h_{(0,1)})^{\gamma_2}.$ 

On the hexagonal lattice, the magnetic translations $T^{h}_{\gamma} : \ell^2(\Lambda_{\tinyvarhexagon}) \rightarrow \ell^2(\Lambda_{\tinyvarhexagon})$ are unitary operators of the above form \eqref{MagTra} with prefactors $(u^h(\gamma)_{v})_{v \in \Lambda_{\tinyvarhexagon}}$ defined as follows:
Let $\alpha(\gamma)= \tfrac{h}{6}(\gamma_1-\gamma_2)$, then we can define $u^B(\gamma)_{r_{*}-\delta_1\vec{b}_1-\delta_2\vec{b}_2} = e^{i\frac{h}{2}\sigma_{\text{symp}}(\gamma,\delta)} u^B(\gamma)_{r_{*}}$ with $*\in \{0,1\}$ where $u^B(\gamma)_{r_0}=1$ and $u^B(\gamma)_{r_1}=e^{i\alpha(\gamma)}.$
This way, the magnetic translations on both lattices satisfy
\begin{equation}
\label{eq:commreld}
 T_{\gamma}^hH^{h}_{\lambda,\omega} = H^{h}_{\lambda,T_{\gamma}^{\Omega}\omega}T_{\gamma}^h.
 \end{equation}

\begin{figure}
\includegraphics[height=6cm]{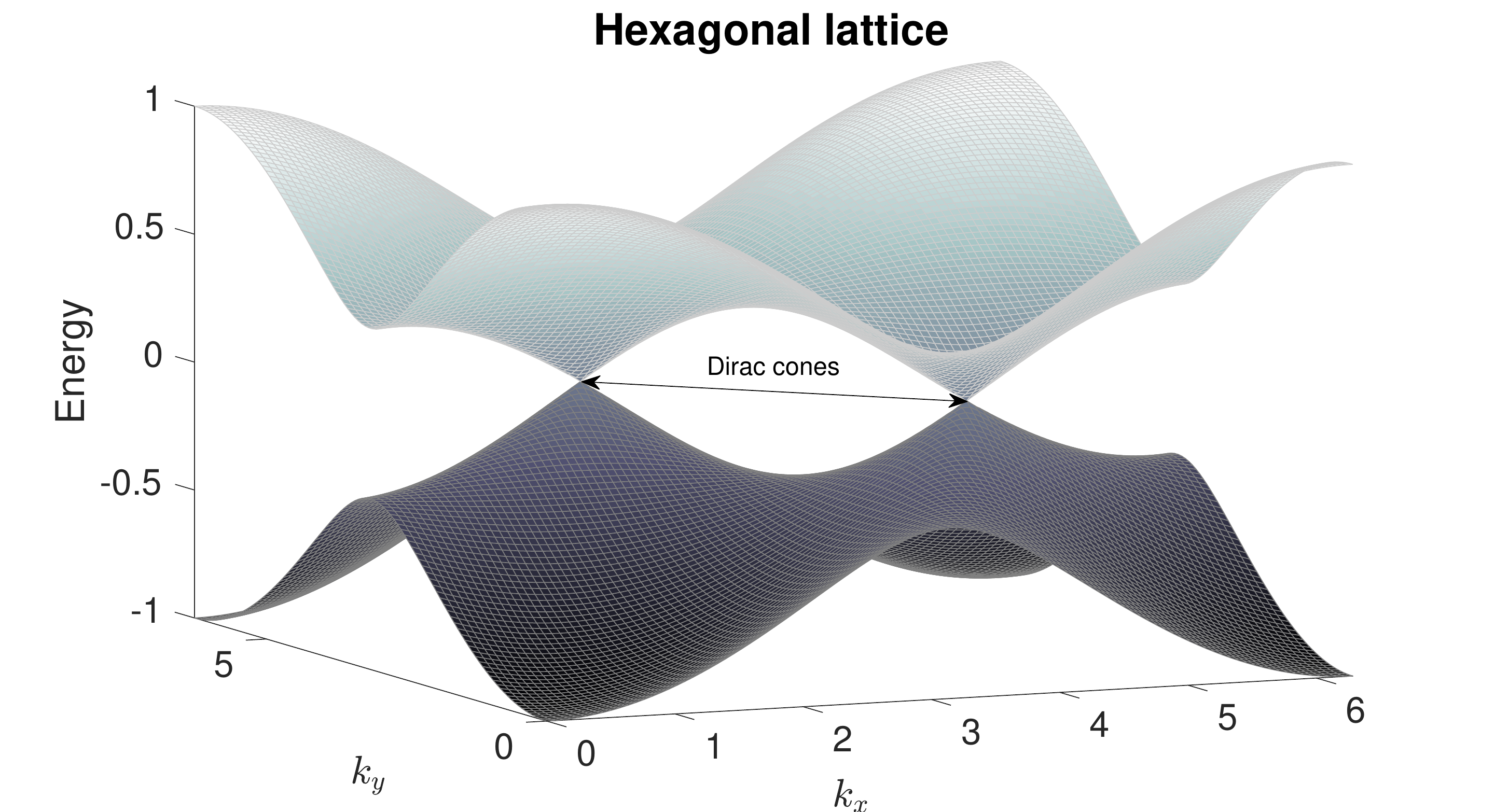} 
\caption{The two energy bands of the non-magnetic discrete Laplacian on $\Lambda_{\varhexagon}$. The Dirac cones are located at zero energy.}
\label{fig:hex}
\end{figure} 

The functional calculus implies that for measurable $f:\RR \rightarrow \RR$
\begin{equation}
\label{eq:equiv}
T_{\gamma}^h f(H^{h}_{\lambda,\omega} )= f(H^{h}_{\lambda,T_{\gamma}^{\Omega}\omega})T_{\gamma}^h
 \end{equation}
such that for the Schwartz kernels $f(H^h_{\lambda,\omega})[x,y] :=\langle \delta_x, f(H^h_{\lambda,\omega}) \delta_y \rangle$ on the diagonal
\begin{equation}
\label{eq:Schwartz}
f(H^h_{\lambda,\omega})[x,x] = f(H^h_{\lambda,T_{\gamma}^{\Omega}\omega})[x-\gamma_1\vec{b}_1-\gamma_2\vec{b}_2, x-\gamma_1\vec{b}_1-\gamma_2\vec{b}_2].
 \end{equation}

To study the density of states (DOS) of the model, we define, for a lattice $ \Gamma \subset \RR^2  $ and operators $ A \in \mathcal L ( \ell^2( \Gamma , \CC^n  ) ) $ 
given by $A (s) (\gamma):= \sum_{\beta \in \Gamma} A[\gamma,\beta]s(\beta)$
with possibly matrix-valued kernel $A[\gamma,\beta]$ \footnote{$(A[\gamma, \beta])_{i,j}=\langle \delta_{\gamma} \vec{e}_i, A (\delta_{\beta} \vec{e}_j\rangle)$, where $\{\delta_{\gamma}\}_{\gamma\in \Gamma}$ is the standard basis of $\ell^2(\Gamma)$ and $\{\vec{e}_m\}_{m=1}^n$ is the standard basis of $\CC^n$.}  $\in \mathbb{C}^{n \times n}$, the regularized trace
\begin{equation}
\label{discretetrace}
\widetilde{\operatorname{tr}}_{\Gamma}(  A ):= \lim_{r \rightarrow \infty} \frac{1}{\left\lvert B_0(r) \right\rvert} \sum_{\gamma \in \Gamma  \cap B_0(r)} \operatorname{tr}_{\mathbb{C}^{n}} A [ \gamma, \gamma ]
\end{equation}
provided the limit exists.

Birkhoff's ergodic theorem implies the a.s.\@ existence of the regularized trace
\begin{equation}
\label{eq:Birkhoff}
 \widetilde{\operatorname{tr}}_{\Lambda}(f(H^h_{\lambda,\omega})) = \mathbb E \left( \tfrac{\sum_{x \in W_{\Lambda}} f(H^h_{\lambda,\omega})[x,x]}{\vert \vec{b}_1 \wedge \vec{b}_2 \vert}\right)  = \tfrac{\mathbb E \operatorname{tr} \indic_{W_{\Lambda}} f(H^{h}_{\lambda,\omega})}{\vert \vec{b}_1 \wedge \vec{b}_2 \vert},
 \end{equation}
 where $\vert \vec{b}_1 \wedge \vec{b}_2 \vert^{-1}$ normalizes the number of vertices per unit volume. 
By Riesz's theorem one can then associate to the regularized trace a Radon measure $\rho_{H^h_{\lambda,\omega}}$, the DOS measure, and by the preceding discussion, this measure is a.s.\@ non-random. Thus $\rho_{H^h_{\lambda,\omega}} =: \rho_{H^h_{\lambda}}$ a.s.\@ and therefore $\int_{\mathbb R} f(x) \ d\rho_{H^h_{\lambda}}(x) =  \widetilde{\operatorname{tr}}_{\Lambda}(f(H^h_{\lambda,\omega}))$ a.s..

\section{The semiclassical expansion of the DOS}
We study the DOS by investigating operators $ f ( H^h_{\lambda,\omega} ) $ using the functional calculus of Helffer--Sj\"ostrand \cite{HS0}. We first recall that any function $f \in C_{\rm{c}}^{\infty}(\mathbb{R})$ can be extended to functions $ \widetilde f \in \mathscr S ( \CC ) $ such that $ \widetilde f|_\RR = f $ and
$D_{\overline{z}} \widetilde f = \mathcal {\mathcal O} ( |\Im z |^\infty ) $. 
Such functions $ \tilde f $ are then called {\em almost analytic} extensions of $ f $. One possible way of defining $ \widetilde f $ is by
\begin{equation}
\label{eq:mather}
\begin{gathered} 
\widetilde f ( x+ i y ) = \frac{1}{ 2 \pi }  \chi ( y ) \psi ( x ) 
\int_{\RR} \chi( y \xi ) \widehat f ( \xi ) e^{ i (x + i y ) \xi } d \xi , 
\\  \chi, \psi \in C^\infty_{\rm{c}}( \RR ) , \ \  \psi| _{ \supp f + 
( - 1, 1 )} = 1, \   \ \chi |_{ (-1, 1) } = 1 ,
\end{gathered}
\end{equation} 
\cite[see Chapter 8]{D-S} for details. A more pedestrian, but also more restrictive, way of defining almost-analytic extensions, for smooth functions $f\in C_c^{\infty}(\mathbb R)$, is for $n \in \mathbb N$ by
\begin{equation}
\begin{split}
\label{eq:taylor}
&\widetilde{f}(x+iy) = \left( \sum_{r=0}^n f^{(r)}(x) \frac{(iy)^r}{r!} \right) \zeta(x+iy)\\
&\zeta(x+iy):=\chi(y/\langle x \rangle),\ \chi \in C^{\infty},\ \chi\vert_{[-1,1]}=1,\ \supp(\chi)\subset[-2,2].
\end{split}
\end{equation}
Differentiating \eqref{eq:taylor}, one finds that $\left\vert D_{\overline{z}}\widetilde{f}(z) \right\rvert = \mathcal O \left(\vert \Im z \vert^{n}\right)$ which follows from
\begin{equation}
\label{eq:diff}
D_{\overline{z}}\widetilde{f}(x+iy) =   \sum_{r=0}^n f^{(r)}(x) \frac{(iy)^r}{r!}  D_{\overline{z}}\zeta(x+iy) + f^{(n+1)}(x)\frac{(iy)^n}{n!}\frac{\zeta(x+iy)}{2}.
\end{equation} 
A similar computation shows that the quasi-analytic extension satisfies 
\begin{equation}
\begin{split}
\label{eq:bd}
\left\vert D_{\overline{z}}\widetilde{f^{(k)}}(z) \right\rvert = \mathcal O \left(\vert \Im z  \vert^{n-k}\right).
\end{split}
\end{equation}

The almost-analytic extension enters then in the Helffer-Sj\"ostrand formula which states that for any self-adjoint operator $ P $, 
\begin{equation}
\label{eq:HeSj}
f(P) = \frac{1}{\pi} \int_{\mathbb{C}} D_{\overline{z}} \widetilde{f}(z) (P -z)^{-1} dm(z)
\end{equation}
where $m$ is the Lebesgue measure on $\mathbb{C}.$ 
For discrete random Schr\"odinger operators \eqref{discreter} this yields by applying the regularized trace
\begin{equation}
\begin{split}
\label{eq:tracesII}
\widetilde{\operatorname{tr}}_{\Lambda}(f(H^h_{\lambda,\omega})) &= \frac{1}{\pi}  \int_{\mathbb C} D_{\overline{z}} \widetilde{f}(z) \widetilde{\operatorname{tr}}_{\Lambda} \left((H^h_{\lambda,\omega}-z)^{-1}\right)  \ dm(z). 
\end{split}
\end{equation}

\subsection{Magnetic matrices}
\begin{defi}[Magnetic matrices]
\label{def:MM}
\label{magmat}
Let $f_{\omega}(\gamma)\in C_c(\Omega \times \mathbb Z^2; \CC^{n \times n})$ at first, where $\omega\in \Omega$ and $\gamma\in \Z^2$. We define {\em magnetic matrices} as discrete operators
as
\begin{equation}
\label{Magneticmatrix}
A^{h} (f_{\omega}) \in 
\mathcal L\left(\ell^2(\mathbb{Z}^2;\mathbb{C}^{n \times n})\right) ,
\ \ \ 
A^{h}(f_{\omega}) := \left(e^{-i \frac{h}{2} \sigma_{\text{symp}}(\gamma,\delta)} f_{T_{\gamma}^{\Omega} \omega}(\gamma-\delta)  \right)_{\gamma,\delta \in \mathbb{Z}^2}.
\end{equation}
These matrices act on $\ell^2(\mathbb{Z}^2;\mathbb{C}^{ n })$ by matrix-like multiplication
\begin{equation}
(A^{h}(f_{\omega})u)_{\gamma} = \sum_{\delta \in \mathbb{Z}^2} 
\left(A^{h}(f_{\omega})\right)_{\gamma,\delta} u_{\delta}.
\end{equation}
\end{defi}
For yet another set of discrete magnetic translation operators $\tau^h_{\gamma}$ on the $\mathbb{Z}^2$-lattice\begin{equation}
\label{eq:deftau}
\tau^h_{\delta}(f_{\omega})(\gamma):= e^{-i\frac{h}{2} \sigma_{\text{symp}}(\gamma,\delta )} f_{T_{\gamma}^{\Omega}\omega}(\gamma-\delta), 
\end{equation}
we find, in analogy to \eqref{eq:commreld}, that magnetic matrices are covariant with respect to discrete magnetic translations \eqref{eq:deftau}
\begin{equation}
\label{eq:equiv2}
A^h(f_{T_{\gamma}^{\Omega}\omega}) \tau^h_{\gamma} = \tau^h_{\gamma}A^h(f_{\omega}).
\end{equation}
Moreover, translations \eqref{eq:deftau} satisfy the Weyl commutation relations
\begin{equation}
\label{commrel}
\tau^h_{\gamma}{\tau}^h_{\delta} = e^{i h \sigma_{\text{symp}}(\gamma,\delta)}{\tau}^h_{\delta}\tau^h_{\gamma}.
\end{equation}
For $f,g \in C_c(\Omega \times \mathbb Z^2; \mathbb C^{n \times n})$ we introduce the product 
\begin{equation}
\begin{split}
\label{eq:hash}
 (f\#_h g)_{\omega}(\gamma) &:=\sum_{z \in \mathbb Z^2} f_{\omega}(\gamma-z)g_{T_{\gamma-z}^{\Omega}\omega}(z)e^{-i \frac{h}{2} \sigma_{\text{symp}}(\gamma,z)} \\
&= \sum_{z \in \mathbb Z^2} f_{\omega}(z) g_{T_{z}^{\Omega}\omega}(\gamma-z)e^{-i\frac{h}{2} \sigma_{\text{symp}}(\gamma,z)}.
\end{split}
\end{equation}
This product is reconcilable with multiplication of magnetic matrices
\begin{equation}
\begin{split}
A^h(f\#_h g)_{\omega}u(\xi) 
&=A^h(f_{\omega})(A^h(g_{\omega})(u))(\xi).
\end{split}
\end{equation}
Moreover, defining the involution
\begin{equation}
\label{eq:involution} 
f_{\omega}^*(\gamma):=\overline{f_{T^{\Omega}_{-\gamma}\omega}(-\gamma)}
\end{equation}
we see that the adjoint of a magnetic matrix is again given by a magnetic matrix 
\begin{equation}
\begin{split}
\langle A^h(f_{\omega})(g),h \rangle 
= \langle g, A^h(f_{\omega}^*)(h) \rangle.
\end{split}
\end{equation}
\begin{rem}
\label{C*remark}
The preceding computations show that magnetic matrices are the \linebreak $*$-representation of a $C^*$-algebra $\mathcal C_h$ which is the closure of functions $f \in C_c(\Omega \times \mathbb Z^2; \mathbb C^{n \times n})$ with composition $\eqref{eq:hash}$ and involution $\eqref{eq:involution} $ under the norm $\Vert f \Vert_{\mathcal C_h}:=  \sup_{\omega \in \Omega} \left\lVert A^h(f) \right\rVert.$
This defines a continuous field (as a function of $h$) of $C^*$-algebra $\mathcal C_h$, cf. \@\cite[Sec.\@ F]{BES},\cite{ST}. 
\end{rem}
To connect operators $H^h_{\lambda,\omega}$ with magnetic matrices, we define symbols
\begin{equation}
\label{1}
\begin{gathered}
a_{{\scriptscriptstyle{ \blacksquare}}}(1,0)=a_{{\scriptscriptstyle{ \blacksquare}}}(0,1)=a_{{\scriptscriptstyle{ \blacksquare}}}(-1,0)=a_{{\scriptscriptstyle{ \blacksquare}}}(0,-1)=\tfrac{1}{4}\text{, and for the hexagonal lattice} \\
a_{\tinyvarhexagon}(0,0):=\tfrac{1}{3} \left(\begin{matrix} 0 && 1 \\ 1 && 0 \end{matrix}\right), \ \ 
a_{\tinyvarhexagon}(1,0):=a_{\tinyvarhexagon}(0,1):=\tfrac{1}{3} \left(\begin{matrix} 0 && 1 \\ 0 && 0 \end{matrix}\right), \\ 
a_{\tinyvarhexagon}(-1,0):=a_{\tinyvarhexagon}(0,-1):=\tfrac{1}{3} \left(\begin{matrix} 0 && 0 \\ 1 && 0 \end{matrix}\right),
\end{gathered}
\end{equation}
and $a(\eta)=0$ otherwise. The random symbols are then defined as $a_{\lambda,\omega,{\scriptscriptstyle{ \blacksquare}}}(\gamma)=a_{{\scriptscriptstyle{ \blacksquare}}}(\gamma)+\lambda \delta_0(\gamma) V_{\omega}(0)$ or 
$a_{\lambda.\omega,\tinyvarhexagon}(\gamma)=a_{\tinyvarhexagon}(\gamma)+\lambda \delta_0(\gamma) \operatorname{diag}(V_{\omega}(r_0),V_{\omega}(r_1)).$
\begin{lemm}
\label{lem:unitmul}
There exist unitary multiplication operators $U_{\scriptscriptstyle{ \blacksquare}}:\ell^2(\ZZ^2;\CC) \rightarrow \ell^2(\ZZ^2;\CC)$ and $U_{\tinyvarhexagon}:\ell^2(\ZZ^2;\CC^2) \rightarrow \ell^2(\Lambda_{\tinyvarhexagon};\CC)$ such that 
\begin{equation}
\label{eq:conjugation}
H^h_{\lambda,\omega,{\scriptscriptstyle{ \blacksquare}}} = U_{\scriptscriptstyle{ \blacksquare}}A^h(a_{\lambda,\omega,{\scriptscriptstyle{ \blacksquare}}})U_{\scriptscriptstyle{ \blacksquare}}^* \text{ and }H^h_{\lambda,\omega,\tinyvarhexagon} = U_{\tinyvarhexagon}A^h(a_{\lambda,\omega,\tinyvarhexagon})U_{\tinyvarhexagon}^*.
\end{equation}
In particular, since operators $U$ are multiplication operators, we find
\begin{equation}
\label{eq:traces}
\widetilde{\operatorname{tr}}_{\Lambda}\left((H^h_{\lambda,\omega}-z)^{-1} \right)= \vert \vec{b}_1 \wedge \vec{b}_2 \vert^{-1} \widetilde{\operatorname{tr}}_{\mathbb Z^2}\left((A^h(a_{\lambda,\omega})-z)^{-1} \right).
\end{equation}
\end{lemm}
\begin{proof}
The first equivalence on the $\ZZ^2$ lattice in \eqref{eq:conjugation} is obtained by first passing from the symmetric to the Landau gauge and then conjugating this operator by $\label{mul2}Wu(\gamma):=e^{-i\frac{h}{2} \gamma_1\gamma_2}u(\gamma).$
For the hexagonal lattice, the transformation is slightly more involved. We start by defining two unitary maps: The first one is $U_1z:=\left(\zeta_{v} z ( v) \right)_{v \in \mathcal{V}(\Lambda_{\tiny\varhexagon})}$ with recursively defined factors 
\begin{equation}
\begin{split}
\zeta_{r_0} :=1 &, \ \ 
\zeta_{\gamma_1\vec{b_1}+\gamma_2\vec{b_2}+r_1}:=e^{i A_{\gamma_1\vec{b_1}+\gamma_2\vec{b_2}+\vec{f}}} \zeta_{\gamma_1\vec{b_1}+\gamma_2\vec{b_2}+r_0}  \\
\zeta_{(\gamma_1+1)\vec{b_1}+\gamma_2 \vec{b_2}+r_0}&:=e^{i \left(-A_{(\gamma_1+1)\vec{b_1}+\gamma_2\vec{b_2} + \vec{g}}+A_{\gamma_1\vec{b_1}+\gamma_2\vec{b_2} + \vec{f}} \right)}\zeta_{\gamma_1\vec{b_1}+\gamma_2\vec{b_2}+r_0} \text{ and }\\
\zeta_{\gamma_1\vec{b_1}+(\gamma_2+1)\vec{b_2}+r_0}&:=e^{i \left(-A_{\gamma_1\vec{b_1}+(\gamma_2+1)\vec{b_2} + \vec{h}}-h\gamma_1+A_{\gamma_1\vec{b_1}+\gamma_2\vec{b_2} + \vec{f}} \right)}\zeta_{\gamma_1\vec{b_1}+\gamma_2\vec{b_2}+r_0} 
\end{split}
\end{equation}
and $U_2:\ell^2(\mathcal{V}(\Lambda_{\tiny\varhexagon})) \rightarrow \ell^2(\mathbb{Z}^2, \mathbb{C}^2)$, $U_2 (z)\left(\gamma\right):= \left( \begin{matrix} z(r_0 + \gamma_1 \vec{b}_1+\gamma_2 \vec{b}_2) \ ,z(r_1+\gamma_1 \vec{b}_1+\gamma_2 \vec{b}_2 )\end{matrix} \right)^T.$
The unitary transform is then $A^h(a_{\lambda,\omega,\tiny\varhexagon}) = (U_1 U_2^*W^*)^*H^h_{\lambda,\omega,\tiny\varhexagon}(U_1 U_2^*W^*)$, see also \cite[Lemma $3.3$, $3.5$]{BZ}. 
\end{proof}
\subsection{Reduction of DOS}
We now continue with the derivation of the DOS. For this, we consider a $\Psi$DO representation of (non-random) magnetic matrices.
To start, we observe the following expansion of the regularized trace of the resolvent of the random operators in terms of the deterministic one. Recall that we write $H^h:=H^h_{\lambda=0,\omega}$ for the non-random DML.
\begin{lemm}
The resolvent of the discrete random Schr\"odinger operator $H^h_{\lambda,\omega}$ satisfies 
\begin{equation}
\begin{split}
\label{eq:cond}
 \widetilde{\operatorname{tr}}_{\Lambda}\left(\left(H^h_{\lambda,\omega}-z\right)^{-1}\right)
&= \sum_{k=0}^2 \tfrac{(-\lambda \mathbb E(V)D_z)^k }{ k!} \widetilde{\operatorname{tr}}_{\Lambda}\left(\left(H^h-z\right)^{-1}\right) \\
 &\quad + \tfrac{\lambda^2}{2}\operatorname{Var}(V) D_{z}  \sum_{r \in W_{\Lambda}}\left(  \operatorname{tr} \left( \indic_{\left\{ r \right\}}\left(H^h-z\right)^{-1}\right)  \right)^2 \\
 &\quad + \mathcal O\left(\lambda^3 \left\Vert(H^h-z)^{-1} \right\Vert^3\left\Vert (H^h_{\lambda,\omega}-z)^{-1} \right\Vert\right).
\end{split}
\end{equation}
\end{lemm}
\begin{proof}
The resolvent identity then yields a second-order approximation in the disorder parameter $\lambda$  
\begin{equation}
\label{eq:powerseries1}
\begin{split} 
\left(H^h_{\lambda,\omega}-z\right)^{-1}
&=\left(H^h-z\right)^{-1}-\lambda  \left(H^h-z\right)^{-1} V_\omega \left(H^h-z\right)^{-1} \\
&\quad+ \lambda^2 \left(H^h-z\right)^{-1} V_{\omega} \left(H^h-z\right)^{-1}V_{\omega} \left(H^h-z\right)^{-1} \\
&\quad + \mathcal O\left(\lambda^3 \left\Vert\left(H^h-z\right)^{-1} \right\Vert^3 \left\Vert\left(H^h_{\lambda,\omega}-z\right)^{-1} \right\Vert\right).
\end{split}
\end{equation}
We study second-order approximations in $\lambda$ since this is the leading-order level at which the stochastic nature of the perturbation enters. Taking regularized traces in \eqref{eq:powerseries1} yields
\begin{equation}
\begin{split}
\label{eq:new}
 \widetilde{\operatorname{tr}}_{\Lambda}\left(\left(H^h_{\lambda,\omega}-z\right)^{-1}\right)
&=(1-\lambda \mathbb E(V) D_z) \widetilde{\operatorname{tr}}_{\Lambda}\left(\left(H^h-z\right)^{-1}\right) \\
&\quad + \lambda^2\widetilde{\operatorname{tr}}_{\Lambda} \left((H^h-z)^{-1} V_{\omega} (H^h-z)^{-1}V_{\omega} (H^h-z)^{-1} \right)\\
&\quad + \mathcal O\left(\lambda^3 \left\Vert(H^h-z)^{-1} \right\Vert^3\left\Vert (H^h_{\lambda,\omega}-z)^{-1} \right\Vert\right).
\end{split}
\end{equation}
Interchanging derivatives and regularized traces is easily justified by \eqref{eq:Birkhoff}.
Equation \eqref{eq:new} can be rewritten, by separating (independent) potentials on different vertices from the squares of potentials such that
\begin{equation}
\begin{split}
&\widetilde{\operatorname{tr}}_{\Lambda}\left(\left(H^h-z\right)^{-1} V_{\omega} \left(H^h-z\right)^{-1} V_{\omega} \left(H^h-z\right)^{-1} \right)\\
&= | \vec{b}_1 \wedge \vec{b}_2 |^{-1} \ \mathbb E \operatorname{tr} \left( \indic_{W_{\Lambda}}(H^h-z)^{-1} V_{\omega} (H^h-z)^{-1}V_{\omega} (H^h-z)^{-1} \right) \\
&=| \vec{b}_1 \wedge \vec{b}_2 |^{-1}\ \mathbb E(V)^2 \operatorname{tr} \left( \indic_{W_{\Lambda}} (H^h-z)^{-3} \right) \\
& \quad +  | \vec{b}_1 \wedge \vec{b}_2 |^{-1} \ \operatorname{Var}(V) \sum_{r \in W_{\Lambda}} \operatorname{tr} \left( \indic_{\left\{ r \right\}} (H^h-z)^{-2} \right)   \operatorname{tr} \left( \indic_{\left\{ r \right\}} (H^h-z)^{-1} \right).\label{eq:inverse}
\end{split}
\end{equation}
Here, we used since $(H^h-z)^{-1}[\gamma,\gamma] = (H^h-z)^{-1}[T_{\nu}\gamma,T_{\nu}\gamma]$, cf. \eqref{MagTra} and \eqref{eq:Schwartz}
\begin{equation}
\begin{split}
&\sum_{x_1,x_2 \in W_{\Lambda},\gamma \in \ZZ^2} (H^h-z)^{-1}[x_1,T_{-\gamma}x_2 ](H^h-z)^{-1}[T_{-\gamma}x_2,T_{-\gamma}x_2] (H^h-z)^{-1}[T_{-\gamma}x_2,x_1] \\
&=\sum_{x_1,x_2 \in W_{\Lambda},\gamma \in \ZZ^2}  (H^h-z)^{-1}[T_{\gamma}x_1,x_2] (H^h-z)^{-1}[x_2,x_2] (H^h-z)^{-1}[x_2,T_{\gamma}x_1]  \\
&=\sum_{r \in W_{\Lambda},v \in \Lambda}  (H^h-z)^{-1}[r,r] (H^h-z)^{-1}[r,v] (H^h-z)^{-1}[v,r] \\
&= \sum_{r \in W_{\Lambda}} \operatorname{tr} \left( \indic_{\left\{ r \right\}} (H^h-z)^{-2} \right)   \operatorname{tr} \left( \indic_{\left\{ r \right\}} (H^h-z)^{-1} \right).
\end{split}
\end{equation}
Inserting this into \eqref{eq:new} yields \eqref{eq:inverse}.
\end{proof}
We now continue expressing the regularized traces of discrete Schr\"odinger operators in terms of pseudodifferential operators.
For vectors $\vec{e}_1:=(1,0)$ and $\vec{e}_2:=(0,1)$, the identity \eqref{commrel} reduces to
\begin{equation}
\label{eq:commrel2}
 \tau^{-h}_{\vec{e}_1}\tau^{-h}_{\vec{e}_2}= e^{-i h}\tau^{-h}_{\vec{e}_2} \tau^{-h}_{\vec{e}_1}.
\end{equation}
This is a version of the canonical commutation relation.
In semiclassical Weyl quantization, the same commutation relation is satisfied by
\begin{equation}
\operatorname{Op}_{h}^{\rm{w}}\left(e^{i x}\right) \operatorname{Op}_{h}^{\rm{w}} \left(e^{i\xi}\right)=e^{-i h}\operatorname{Op}_{h}^{\rm{w}} \left(e^{i \xi}\right)\operatorname{Op}_{h}^{\rm{w}}\left(e^{i x}\right).
\end{equation}
Rather than analyzing directly the discrete operators $H^h:=H^h_{\lambda=0,\omega}$ or 
$A^h(a):=A^h(a_{\lambda=0,\omega})$, we use a pseudodifferential representation that we obtain from the following $*$-homomorphism $\Theta : \mathscr S ( \ZZ^2; \mathbb C^{n \times n} ) \to \mathcal L\left(L^2(\mathbb{R};\CC^{n \times n})\right)$:
\begin{equation}
\begin{split}
 \label{star}
\Theta ( f) :=\operatorname{Op}^{\rm{w}} _{h}(\widehat{f}(x, \xi))= \sum_{\gamma \in \mathbb{Z}^2} f(\gamma)\operatorname{Op}^{\rm{w}} _{h}\left((x,\xi)\mapsto e^{i\langle \gamma, (x,\xi) \rangle}\right) \nonumber \\
\text{such that }\Theta ( f\#_hg)  = \Theta ( f ) \circ \Theta ( g).  
\end{split}
\end{equation}
See \cite[Sec.$6$]{HS1} for details of this construction.
Here, $ \mathscr S ( \ZZ^2; \mathbb C^{n \times n} )$ are the $\mathbb C^{n \times n} $-valued functions that decay faster than any polynomial power on $\ZZ^2$.
We now define a regularized trace $\widetilde{\operatorname{tr}}$ for $\Psi$DOs with periodic symbol such that $\widetilde{\operatorname{tr}}_{\mathbb Z^2}(A^h(f)) = \widetilde{\operatorname{tr}}(\operatorname{Op}_h^{w}(\widehat{f}))$:
\begin{defi}
\label{pdotrace}
Let $\widehat{f} \in C^\infty ( \RR^2; \CC^{n \times n} ) $ be $\mathbb{Z}_{*}^2$ periodic. Then we define the regularized trace
\begin{equation}
\widetilde{\operatorname{tr}}(\operatorname{Op}_h^{w}(\widehat{f})) := 
 \int_{\mathbb{T}^2_*} \operatorname{tr}_{\mathbb{C}^{n }} \widehat{f}(x,\xi)\,  \frac{dx \ d\xi}{ \vert \mathbb T_*^2 \vert }.
\end{equation}
\end{defi}

We can express the resolvent of the Hamiltonian in \eqref{eq:tracesII}, by the $C^*$-homomorphism $\Theta$ and the trace identity, in terms of $\operatorname{\Psi DOs}$
\begin{equation}
\begin{split}
\label{eq:effha}
Q^{\rm{w}}_{{\scriptscriptstyle{ \blacksquare}}}( x , h p_x )  &:= \tfrac{1}2 \left( \cos(x)+ \cos(hp_x) \right)\text{ and } \\
Q^{\rm{w}}_{\tinyvarhexagon}( x , h p_x )  &:= \tfrac{1}3 \left(\begin{matrix} 0& 1 + e^{ i x } + e^{ i h p_x } \\
1 + e^{ - i x } + e^{ - i h p_x } & 0 \end{matrix}\right),
\end{split}
\end{equation}
which are the semiclassical Weyl-quantizations of
\begin{align}
\begin{split}
\label{eq:symbolQ}
&Q_{{\scriptscriptstyle{ \blacksquare}}}(x,\xi) :=\widehat{a}_{\scriptscriptstyle{ \blacksquare}}(x,\xi)= \tfrac{\cos(x)+\cos(\xi)}2\\  \text{ and } &Q_{\tinyvarhexagon}(x,\xi):=\widehat{a}_{\tinyvarhexagon}(x,\xi)=\left( \begin{matrix} 0 && \tfrac{1+e^{ix}+ e^{i\xi}}{3} \\ \tfrac{1+e^{-ix}+e^{-i\xi}}{3} &&  0  \end{matrix} \right). 
\end{split}
\end{align}
In particular, the $C^*$-homomorphism $\Theta$ implies
\begin{equation}
\label{eq:pseudo}
 \widetilde{\operatorname{tr}}_{\mathbb Z^2} \left((A^h(a)-z)^{-1}\right) = \widetilde{\operatorname{tr}} \left((Q^{\rm{w}}( x , h p_x )-z)^{-1}\right).
\end{equation}
The trace on the right hand side is well-defined, as $(Q^{\rm{w}}( x , h p_x )-z)^{-1}$ is again a $\Psi$DO with periodic symbol in $\mathcal S$ by the semiclassical Beal's lemma \cite[Theorem $8.3$]{ev-zw}, \cite[Prop.\@5.1]{HS0}.
To conclude, we can express the DOS of $H^h_{\lambda,\omega}$ in terms of pseudodifferential operators \eqref{eq:effha} as follows:
\begin{prop}
\label{fP2Q}
Let $f \in C^5_{\rm{c}}  ( \RR) $ and $ \widetilde f $ be an almost analytic extension \eqref{eq:taylor}, then for $n=1$, in case of the square, and $n=2$, in case of the hexagonal lattice,
\begin{equation}
\label{tracediff}
\begin{split} 
&\widetilde{\operatorname{tr}}_{\Lambda}(f(H^h_{\lambda,\omega})) =  \sum_{k=0}^2 \tfrac{\lambda^k \mathbb E(V)^k }{ \pi| \vec{b}_1 \wedge \vec{b}_2 | k!}  \int_{\mathbb{C}} D_{\overline{z}}\widetilde{f^{(k)}}(z)   \widetilde{\operatorname{tr}} \left(\,  (Q^{\rm{w}} ( x, h p_x ) - z)^{-1}\right) \ d m (z) \\
& \quad - \tfrac{ \operatorname{Var}(V) \lambda^2}{2\pi| \vec{b}_1 \wedge \vec{b}_2 |} \sum_{i=1}^n  \int_{\mathbb{C}}D_{\overline{z}}\widetilde{f'}(z)  \widetilde{\operatorname{tr}} \left((Q^{\rm{w}} ( x, h p_x )-z)^{-1}_{ii}\right) ^2    \ d m (z) + \mathcal O(\Vert f^{(5)} \Vert_{L^{\infty}}\lambda^3).
\end{split} 
\end{equation}
\end{prop}
\begin{proof}
 By inserting \eqref{eq:cond} into the Helffer-Sj\"ostrand formula \eqref{eq:tracesII}, we find
\begin{equation}
\begin{split}
&\widetilde{\operatorname{tr}}_{\Lambda}(f(H^h_{\lambda,\omega})) = \tfrac{1}{\pi | \vec{b}_1 \wedge \vec{b}_2 | } \int_{\mathbb C} D_{\overline{z}}\widetilde{f}(z)\Bigg(\sum_{k=0}^2 \tfrac{(-\lambda \mathbb E(V)D_z)^k }{ k!} \widetilde{\operatorname{tr}}_{\Lambda}\left(\left(H^h-z\right)^{-1}\right) \\
 &+ \tfrac{\lambda^2\operatorname{Var}(V) }{2}D_{z}  \sum_{r \in W_{\Lambda}}\left(  \operatorname{tr} \left(\indic_{\{r\}}\left(H^h-z\right)^{-1}\right)  \right)^2   + \mathcal O\left(\lambda^3 \left\lvert \Im (z) \right\rvert^{-4} \right)\Bigg) \ dm(z). 
 \end{split}
 \end{equation} 
Using $D_{\overline{z}}\widetilde{f} = \mathcal O\left(\vert \Im (z)\vert^{4}\right)$, as in \eqref{eq:diff} for the almost-analytic extension, we can compensate the $\left\lvert \Im (z) \right\rvert^{-4}$ singularity. 
To express the right-hand side in terms of $\Psi$DOs, rather than $H^h$, we use \eqref{eq:traces} and \eqref{eq:pseudo} which upon integration by parts yields \eqref{tracediff}.
\end{proof}

Our main result on the DOS for small magnetic fields is stated in the following Theorem:

\begin{theo}[Semiclassical expansion of DOS]
\label{theo1}
For small magnetic fields $h>0$ and small disorder $\lambda$ the \emph{DOS} satisfies: \\
\emph{Square lattice ($ \blacksquare$):} Let $I$ be an interval $I \subset [-1,-1+\delta)$ or $I \subset (1-\delta,1]$ for some $\delta>0$ sufficiently small\footnote{This interval is located at the bottom/top of the spectrum in Figure \ref{fig:square}.} and $ f \in C^{5}_{\rm{c}} ( I) $, then for functions $g_{{\scriptscriptstyle{ \blacksquare}},n}$ (independent of $\lambda$), defined in \eqref{eq:g} ,
\begin{equation}
\begin{split}
\label{eq:tracef}
&\widetilde \tr_{\Lambda} (f ( H^h_{{\scriptscriptstyle{ \blacksquare}},\lambda,\omega} )) =\tfrac{h}{2\pi} \sum_{ n \in \NN } 
f( z_{n} ( h  ) +\lambda \mathbb E(V)) \\
&\ \ \ \ \ \ \ \ \ \ \ \ \ \ \  \ \ \ \  \ - \tfrac{h  \operatorname{Var}(V) \lambda^2}{4\pi }  \sum_{ n \in \NN } 
\left(\tfrac{f'' ( z_n ( h  ) )}{2\pi} +f' ( z_n ( h  ) )g_{{\scriptscriptstyle{ \blacksquare}},n}(z_n(h),h) \right) \\
&\ \ \ \ \ \ \ \ \ \ \ \ \ \ \  \ \ \ \  \ + \mathcal O ( \left\lVert f \right\rVert_{C^5}(\lambda^3+h^{\infty})) \text{ a.s.} , \ 
\end{split}
\end{equation}
with Landau levels $z_{n} ( h ) = \kappa ( n h , h )-1$ defined, for $n \in \mathbb N$, by a Bohr-Sommerfeld condition
\begin{equation}
\label{eq:g2F}  
\begin{gathered} F_{{\scriptscriptstyle{ \blacksquare}}} ( \kappa ( \zeta, h ) , h ) = \zeta + \mathcal O 
(h^\infty)  , \ \ 
F_{{\scriptscriptstyle{ \blacksquare}}} ( s , h ) \sim   \sum_{j=0}^\infty h^j F_{j,{\scriptscriptstyle{ \blacksquare}}} ( s ) , \ \ F_{j,{\scriptscriptstyle{ \blacksquare}}} \in C^\infty ( \RR ) , \\ 
F_{0,{\scriptscriptstyle{ \blacksquare}}} ( s  ) =  \frac{1}{2\pi} \int_{ \gamma_s } \xi \ dx , \ \ 
\gamma_s = \left\{ ( x, \xi) \in \mathbb{T}^2_*: 2-\cos(x)-\cos(\xi) = 2s \right\}, \ F_{1,{\scriptscriptstyle{ \blacksquare}}} ( s  ) = \frac{1}{2},
\end{gathered} 
\end{equation}
where $ \gamma_s $ is oriented clockwise in the $ ( x, \xi ) $ plane. 

\emph{Hexagonal lattice ($ \varhexagon$):} Let $I$ be an interval $I \subset (-\delta,\delta)$ for some $\delta>0$ sufficiently small\footnote{This interval encloses energies around the Dirac points in Figure \ref{fig:hex}.} and $ f \in C^{5}_{\rm{c}} ( I), $ then for functions $g_{\tinyvarhexagon,n}$, defined in \eqref{eq:g} ,
\begin{equation}
\begin{split}
\label{eq:tracef2}
&\widetilde \tr_{\Lambda} (f (  H^h_{\tinyvarhexagon,\lambda,\omega})) = \tfrac{ h}{\pi \vert \vec{b}_1 \wedge \vec{b}_2 \vert}  \sum_{ n \in \ZZ } 
f( z_{n} ( h  ) +\lambda \mathbb E(V))\\
&\ \ \ \ \ \ \ \ \ \ \ \ \ \  \ \ \ \ \ - \tfrac{h  \operatorname{Var}(V) \lambda^2}{2\pi| \vec{b}_1 \wedge \vec{b}_2 |}  \sum_{ n \in \ZZ } 
\left(\tfrac{f'' ( z_n ( h  ) )}{2\pi} + f' ( z_n ( h  ) )g_{\tinyvarhexagon,n}(z_n(h),h) \right) \\
&\ \ \ \ \ \ \ \ \ \ \ \ \ \  \ \ \ \ \ + \mathcal O ( \left\lVert f \right\rVert_{C^5}(\lambda^3+h^{\infty}))\text{ a.s.},
\end{split}
\end{equation}
with Landau levels $z_{n} ( h ) = \kappa ( n h , h )$ satisfying $ \kappa ( - \zeta, h ) = - \kappa ( \zeta , h )  $, defined, for $n \in \mathbb Z$, by a Bohr-Sommerfeld condition
\begin{equation}
\label{eq:g2F2}  
\begin{gathered} 
F_{\tinyvarhexagon} ( \kappa ( \zeta, h )^2 , h ) = |\zeta| + \mathcal O 
(h^\infty)  , \ \ 
F_{\tinyvarhexagon} ( s , h ) \sim  F_{0,\tinyvarhexagon} ( s ) + \sum_{j=2}^\infty h^j F_{j,\tinyvarhexagon} ( s ) , \ \ F_{j,\tinyvarhexagon} \in C^\infty ( \RR ) , \\ 
F_{0,\tinyvarhexagon} ( s  ) = \frac{1}{4 \pi} \int_{ \gamma_s } \xi \ dx , \ \ 
\gamma_s = \left\{ ( x, \xi) \in \mathbb{T}^2_*: | 1 + e^{ix } + e^{i\xi} |^2 = 9s \right\}, \ \ F_{j,\tinyvarhexagon} ( 0 ) = 0 , 
\end{gathered} 
\end{equation}
where $ \gamma_s $ is oriented clockwise in the $ ( x, \xi ) $ plane. 
\end{theo}

\begin{figure}
\includegraphics[height=5 cm]{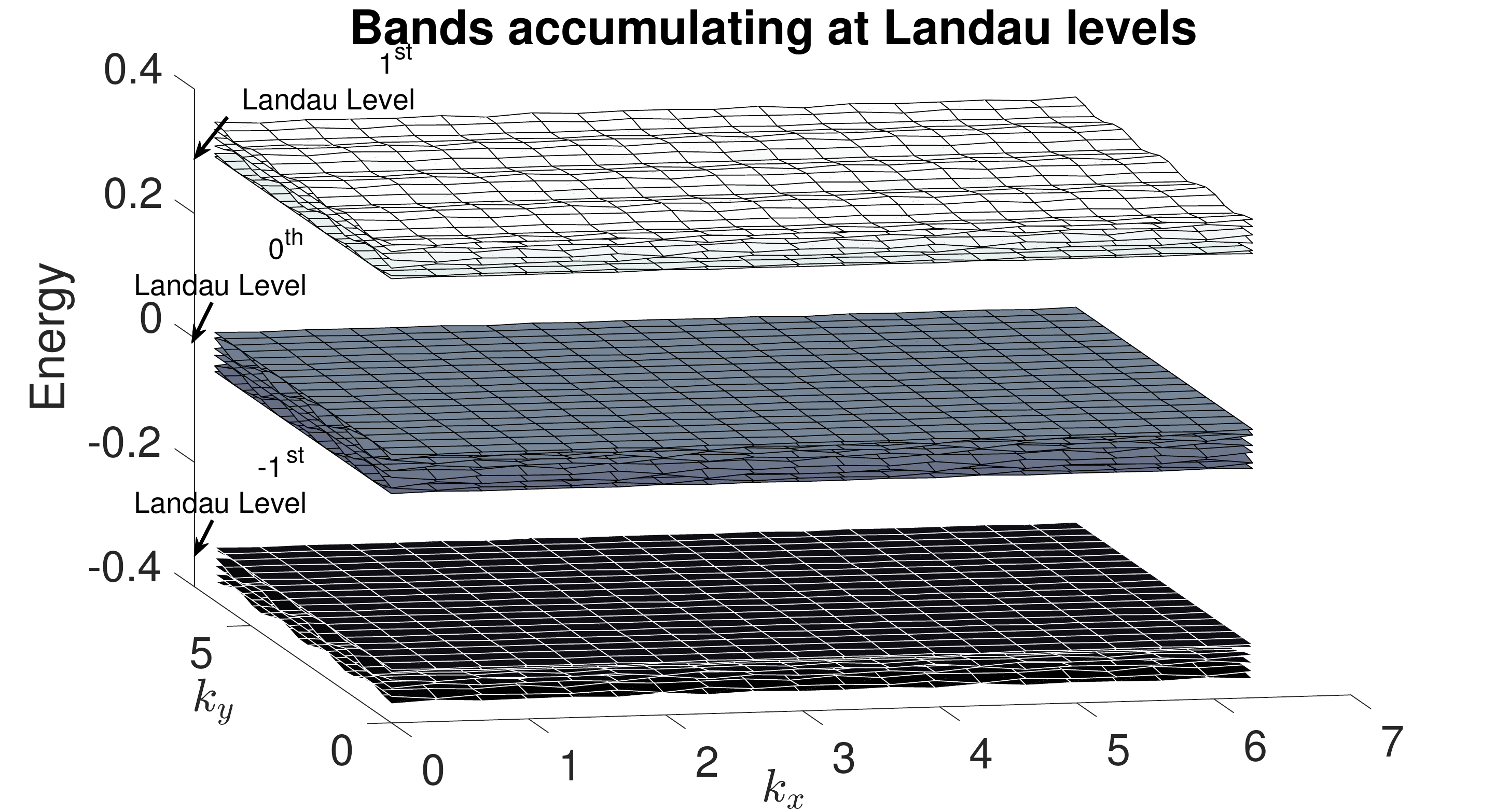} 
\caption{Energy bands for magnetic flux $h=2\pi \frac{4}{30}$ on $\Lambda_{\varhexagon}$ close to the zero energy level. Bands concentrate around certain energies which are precisely the Landau levels defined in Theorem \ref{theo1}.}
\label{fig:HC}
 \end{figure} 

The proof of Theorem \ref{theo1} is given at the end of this article in Section \ref{sec:proofs}.

\begin{rem}
The different prefactor $h/2\pi$ for the square lattice compared with $h/\pi$ for the hexagonal lattice is due to the two-fold degeneracy of quasimodes on the hexagonal lattice (two Dirac cones and therefore two potential wells), cf.\@ Fig. \ref{fig:hex}.

In particular, for functions $f$ whose first and second derivative vanishes at the Landau levels, the randomness only causes a shift of the Landau levels by $\lambda \mathbb E(V).$ This can be thought of as a \emph{semiclassical universality result} for the integrated density of states, if one takes $f$ to be (a smooth approximation of) an indicator function.
By this we mean that the leading-order contribution in the semiclassical parameter $\lambda>0$ in the second line of \eqref{eq:tracef2} vanishes.
\end{rem}

We start by showing that for small enough magnetic fields \emph{without disorder} there exist spectral gaps between the \emph{Landau levels} stated in Theorem \ref{theo1}, cf. Figure \ref{fig:HC}. The presence of spectral gaps is crucial for the study of the quantum Hall effect, as the Hall conductivity remains unchanged as long as the Fermi energy stays inside a spectral gap. 

From the Bohr-Sommerfeld condition stated in Theorem \ref{theo1} in the absence of disorder, i.e.\@ $\lambda \equiv 0,$ we obtain to leading-order approximative Landau levels $z^{(1)}(h)$
\begin{equation}
\begin{split}
\label{BSC}
F_{0,{\scriptscriptstyle{ \blacksquare}}}\vert_{I}\left(z^{(1)}_{{\scriptscriptstyle{ \blacksquare}},n} ( h )\right)=nh,  \text{ and  }
F_{0,\tinyvarhexagon}\vert_{I} \left(z^{(1)}_{\tinyvarhexagon,n} ( h )\right) =|n|h,
\end{split}
\end{equation}
where $F_0$ is the respective normalized phase space area in the Brillouin zone as stated in \eqref{eq:g2F} and \eqref{eq:g2F2}, and $I$ is the respective region of interest, i.e.\@ the respective interval defined in Theorem \ref{theo1}. 
While approximate Landau levels $z^{(1)}_{{\scriptscriptstyle{ \blacksquare}},n} ( h )$ for the square lattice are uniquely defined by the first of the equations in \eqref{BSC}, there are two solutions for the hexagonal lattice (because of the upper and lower cone, see Figure \ref{fig:hex}):
Let us recall from Theorem \ref{theo1} that the asymptotic expansion yields
\begin{equation}
\begin{split}
 F_{{\scriptscriptstyle{ \blacksquare}}} (  z_{{\scriptscriptstyle{ \blacksquare}},n} ( h ) ,h) &= F_{0,{\scriptscriptstyle{ \blacksquare}}} (  z_{{\scriptscriptstyle{ \blacksquare}},n} ( h ) ) + 
\mathcal O ( h^2 z_{{\scriptscriptstyle{ \blacksquare}},n} ( h ) ,h) =nh + \mathcal O ( h^\infty ) ,\\
 F_{\tinyvarhexagon} ( z_{\tinyvarhexagon,n}  ( h ) ^2 ) &= F_{0,\tinyvarhexagon} ( z_{\tinyvarhexagon,n}  ( h ) ^2 ) + 
\mathcal O ( h^2  z_{\tinyvarhexagon,n}  ( h )^2 ) = | n | h + \mathcal O ( h^\infty ) ,\\
\end{split}
\end{equation}
which gives for the leading-order approximations \eqref{BSC} of Landau levels
\begin{equation}
\begin{split}
 z_{{\scriptscriptstyle{ \blacksquare}},n}   ( h )    &=  z_{{\scriptscriptstyle{ \blacksquare}},n}^{(1)}  ( h) 
+ \mathcal O (n h^3) + \mathcal O (h^\infty ) \\
 z_{\tinyvarhexagon,n}  ( h )  ^2  &= z_{\tinyvarhexagon,n}^{(1)}  ( h)^2 
+ \mathcal O ( |n| h^3) + \mathcal O (h^\infty ).
\end{split}
\end{equation}
Hence, by Taylor expansion, Landau levels are to leading order given by
\begin{equation}
\begin{split}
\label{eq:zone}
z_{{\scriptscriptstyle{ \blacksquare}},n}  ( h ) &= z_{{\scriptscriptstyle{ \blacksquare}},n}^{(1) } ( h ) + \mathcal O \left(  nh^{ 3} \right)  \text{ and } z_{\tinyvarhexagon,0} ( h ) = 0+\mathcal O(h^{\infty}) \\
z_{\tinyvarhexagon,n} ( h ) &= z_{\tinyvarhexagon,n}^{(1)} ( h ) + \mathcal O \left( { |n|^{\frac12}} {h^{ \frac52} } \right),     n \neq 0. 
\end{split}
\end{equation}

To make these expressions more concrete, we approximate the cross-section for the square lattice by using that
\[\frac{\cos(x)+\cos(\xi)}{2}+1 = \frac{(x-\pi)^2+(\xi-\pi)^2}{4} + \mathcal O(x^3+\xi^3). \]
Thus, $F_{0,{\scriptscriptstyle{ \blacksquare}}} (s)=2s + \mathcal O(s^2)$ which yields for the Landau levels
\[z_{{\scriptscriptstyle{ \blacksquare}},n}^{(1)}   ( h )  = \frac{(n-\tfrac{1}{2})h}{2} + \mathcal O(n^2h^2), \quad n \in \mathbb N.  \] 
For the hexagonal lattice, we use that 
$ {| 1 + e^{ i x} + e^{ i \xi } |^2}/9 $ vanishes at $( x , \xi ) \in  \mathbb Z^2_* \pm \left( \tfrac{2 \pi } {3} , - \tfrac{ 2 \pi } 3 \right),$
that is, at the Dirac points, see Figure \ref{fig:hex}. 

In small neighbourhoods of $ \pm ( \tfrac{2 \pi } {3} , - \tfrac{ 2 \pi } 3 )  $ we can make a symplectic (and thus area-preserving) change of variables
\[  y = a ( x + \xi ) , \ \ \eta = b \left( \xi - x \pm \tfrac{4 \pi } 3 \right) , \ \ 
2 ab = 1, \]
and find that
\begin{equation}
\label{eq:norfo}
\begin{split}  1 + e^{ i x } + e^{ i \xi } & = 
c ( \eta \mp i y )  + \mathcal O ( y^2 + \eta^2 ) , \\
1 + e^{ -i x } + e^{ - i \xi } & = c ( \eta \pm i y )  + \mathcal O ( y^2 + \eta^2 ), 
\end{split}
\end{equation}
where $ c = 3^{\frac14} 2^{-\frac12} $ by choosing
$ a = \pm 2^{-\frac12} 3^{-\frac14} $ and $ b = \pm 2^{-\frac12} 3^{\frac14} 
$. 
We thus conclude that for a Fermi velocity $v_F:= \sqrt{2}c/3 = 3^{-3/4}$
\[ z_{\tinyvarhexagon,n}^{(1)}(h) =v_F \sgn(n) \sqrt{\vert n \vert h} + \mathcal O(\vert n\vert h), \quad n \in \mathbb Z. \]

\begin{prop}[Spectral gaps between Landau levels]
\label{prop:SG}
For small $h>0$, the intersection of the region of interest $I$, in Theorem \ref{theo1}, with the spectrum of $H^h:=H^h_{\lambda=0,\omega}$, $\Sigma(H^h) \cap I,$ is contained in disjoint intervals defined by constants $C_{\scriptscriptstyle{ \blacksquare},n},C_{\tinyvarhexagon,n}>0$
\begin{equation}
\begin{split}
\label{eq:Bands}
B_{{\scriptscriptstyle{\blacksquare}},n}(h)&:=[z_{{\scriptscriptstyle{ \blacksquare}},n}^{(1)}(h)-C_{{\scriptscriptstyle{ \blacksquare}},n}h^{3},z_{{\scriptscriptstyle{ \blacksquare}},n}^{(1) }(h)+C_{{\scriptscriptstyle{ \blacksquare}},n} h^{3}], \ n  \in [1,...,N_{{\scriptscriptstyle{ \blacksquare}}}(h)] \\
B_{\tinyvarhexagon,n}(h)&:=[z_{\tinyvarhexagon,n}^{(1) }(h)-C_{\tinyvarhexagon,n}h^{\frac{5}{2}},z_{\tinyvarhexagon,n}^{(1) }(h)+C_{\tinyvarhexagon,n} h^{\frac{5}{2}}], \ n  \in [-N_{\tinyvarhexagon}(h),...,N_{\tinyvarhexagon}(h)].
\end{split}
\end{equation}
Moreover, numbers $N(h)$ have the property that $\lim_{h \downarrow 0} N(h)=\infty.$ 
\end{prop}
\begin{proof}
Since the density of states measure is supported exactly where spectrum is, we conclude that the contribution to the DOS from the Landau levels, i.e.\@ the first term on the right hand side of \eqref{eq:tracef} and \eqref{eq:tracef2} is contained in closed \emph{Landau bands}
\begin{equation}
\begin{split}
\label{BLandau}
B_{{\scriptscriptstyle{ \blacksquare}},n}(h)&:=\left[z_{{\scriptscriptstyle{ \blacksquare}},n}^{(1)}(h)-C_{{\scriptscriptstyle{ \blacksquare}},n}h^{3},z_{{\scriptscriptstyle{ \blacksquare}},n}^{(1) }(h)+C_{{\scriptscriptstyle{ \blacksquare}},n} h^{3}\right], \ n \in \mathbb N\\
B_{\tinyvarhexagon,n}(h)&:=\left[z_{\tinyvarhexagon,n}^{(1) }(h)-C_{\tinyvarhexagon,n}h^{\frac{5}{2}},z_{\tinyvarhexagon,n}^{(1) }(h)+C_{\tinyvarhexagon,n}h^{\frac{5}{2}}\right], \ n \in \mathbb Z.
\end{split}
\end{equation}
It remains to exclude spectrum of $\mathcal O(h^{\infty})$-size, see the error bounds in \eqref{eq:tracef} and \eqref{eq:tracef2}, outside intervals $B_n$, possibly after modifying constants $C_n.$ This can be shown, using semiclassical techniques as in \cite[Prop.\@$5.2$]{BZ}. To be precise, the Proposition in \cite{BZ} states that there exists an operator $Q^{\operatorname{w}}_0(x,hp_x)$ whose point spectrum for the hexagonal lattice around zero coincides with the Landau levels, such that if for $z \in \operatorname{nbhd}(0)$, and some fixed $N_0,$
\[ d(z,\Sigma(Q_0^{\operatorname{w}}(x,hp_x)))>h^{N_0} \]
then the operator $Q_{\tinyvarhexagon}^{\operatorname{w}}(x,hp_x)$, that is isospectral to $H^h_{\tinyvarhexagon}$, cf. \cite{BGHKS}[Theo. $6.2$], is also invertible for such $z$. Hence, $H^h_{\tinyvarhexagon}$ does not possess \emph{any spectrum} between the Landau bands. The same argument applies to the square lattice in a neighbourhood of $\pm 1.$
\end{proof}

The preceding Proposition implies that under small disorder, the closed Landau bands in the region of interest will broaden but are still non-overlapping since the decomposition $H^h_{\lambda,\omega} = H^h + \lambda V_{\omega}$ implies
\begin{equation}
\label{eq:TI}
 \Sigma(H^h_{\lambda,\omega}) \subset \left\{z \in \mathbb R;  d(z, \Sigma(H^h)) \le \lambda \left\lVert V \right\rVert_{\infty} \right\}.
 \end{equation}
It follows from Proposition \ref{prop:SG} and \eqref{eq:TI} that for sufficiently weak magnetic fields $h>0$ and small disorder $\lambda \in (0,\lambda_0(h))$ there exist for $H^h_{\lambda,\omega}$ finitely many \emph{(disorder-broadened)} disjoint intervals $\mathcal B_{n,\lambda}(h) \supset B_n(h) $ with $n \in \left\{1,..,N_{\scriptscriptstyle\blacksquare,\lambda}(h) \right\}$, for the square lattice, or with $n\in \left\{-N_{\scriptscriptstyle\tinyvarhexagon,\lambda}(h) ,..,N_{\scriptscriptstyle\tinyvarhexagon,\lambda}(h) \right\}$, in case of the hexagonal lattice, such that 
\begin{equation}
\label{eq:DOLL}
\Sigma(H^h_{ \lambda,\omega}) \subset \mathcal \cup_n \mathcal B_{n,\lambda}(h) \ \ \text{ for all }  \lambda \in (0,\lambda_0(h)),
\end{equation}
where the union of $n$ is taken over the respective sets.

\medskip
Moreover, we assume without loss of generality that the disorder-broadened Landau bands are nested, i.e.\@ for $  \nu  \le \lambda $ we have $\mathcal B_{n,\nu}(h)  \subset \mathcal B_{n,\lambda}(h).$

\section{Quantum Hall effect}
\subsection{The QHE without disorder}
We start by studying the Quantum Hall effect in the absence of disorder using the DOS stated in Theorem \ref{theo1} (we assume $\hbar \in \mathbb R \backslash \mathbb Q$ in the following paragraph). We take St\v{r}eda's formula \cite{S82} as the definition of the Hall conductivity:
\begin{defi}[St\v{r}eda formula]
For (possibly random) Schr\"odinger operators $H^h_{\lambda,\omega}$ with Fermi energy $\mu$ inside a gap $d(\mu,\Sigma(H^h_{\lambda,\omega}))>0$ a.s.\@ we define the Hall conductivity by the St\v{r}eda formula
\begin{equation}
\label{eq:express1}
 c_H(H^h_{\lambda,\omega},\mu) := \vert \vec{b}_1 \wedge \vec{b}_2 \vert D_h \widetilde{\operatorname{tr}}_{\Lambda} \left(\indic_{(-\infty, \mu]}(H^h_{\lambda,\omega}) \right). 
 \end{equation}
 \end{defi}
 
The DOS is differentiable, since by \eqref{eq:Birkhoff} the right-hand side of
\[ \widetilde{\operatorname{tr}}_{\Lambda}(\indic_{I}(H^h_{\lambda,\omega })) = \frac{\mathbb E \operatorname{tr} \indic_{W_{\Lambda}} \indic_{I}(H^h_{\lambda,\omega})}{ \vert \vec{b}_1 \wedge \vec{b}_2 \vert} \]
is differentiable.
This follows from holomorphic functional calculus
\[  \indic_{I}(H^h_{\lambda,\omega}) = (2\pi i)^{-1} \oint_{I} (z-H^h_{\lambda,\omega})^{-1} \ dz, \]
as $H^h_{\lambda,\omega}$ depends analytically on $h$, i.e. $h \mapsto \indic_{I}(H^h_{\lambda,\omega})$ is differentiable as long as $\partial I$ is in a spectral gap. Thus, $h \mapsto \widetilde{\operatorname{tr}}_{\Lambda}(\indic_{I}(H^h_{\lambda,\omega }))$ is differentiable as well.

On $\ell^2(\mathbb Z^2)$ we define the rotation algebra $\mathcal A_{\hbar}$ as the operator norm closure
\begin{equation}
\label{eq:irratalg}
\mathcal A_{\hbar} := \overline{\left\{ T \in \mathcal L( \ell^2(\mathbb Z^2; \CC^n)); \exists k \in \mathbb N, \ c_{\gamma} \in \mathbb C: T=\sum_{\vert \gamma \vert \le k} c_{\gamma} \tau^h_{\gamma} \right\}}^{\Vert \bullet \Vert}. 
 \end{equation}
Magnetic matrices introduced in Definition \ref{def:MM} form a $*$-representation of the irrational rotation algebra.We then focus on the subalgebra $\mathcal A_{\hbar}^{\infty} \subset \mathcal  A_{\hbar} $ of magnetic matrices with rapidly decaying symbols, i.e. with coefficients  in \eqref{eq:irratalg} that satisfy $(c_{\gamma}) \in \mathscr S(\mathbb Z^2; \mathbb C)$. The set $\mathcal A_{\hbar}^{\infty} $ is still a locally convex algebra equipped with standard seminorms inducing decay faster than any polynomial power $\vert (c_{\gamma}) \vert_i:= \sup_{\gamma \in \mathbb Z^2} \left\vert (1+\vert \gamma\vert)^{i} c_{\gamma} \right\vert_{\mathbb C^{n \times n}}.$ Moreover, the inverse of a magnetic matrix $A^h(a) \in \mathcal A_{-\hbar}^{\infty}$ is again a magnetic matrix \cite[Prop.\@ 5.1]{HS0}, i.e.\@ we have for $z \notin \Sigma(A^h(a))$ that $(A^h(a)-z)^{-1} \in \mathcal A_{-\hbar}^{\infty}$, again.\footnote{Equation \eqref{eq:commrel2} shows that magnetic matrices satisfy the canonical commutation relation with $-h$ rather than $h$.}

The smooth subalgebra $\mathcal A_{\hbar}^{\infty}$ is stable under holomorphic functional calculus \cite[Ch.3 App.\@C]{C94} which implies that Fermi projections of $A^h(a)$, are again elements of $\mathcal A_{-\hbar}^{\infty}$, as long as $\mu \notin \Sigma(A^h(a))$
\[ \indic_{(-\infty, \mu]}(A^h(a)) = (2\pi i)^{-1} \oint_{\Sigma(A^h(a))} (z-A^h(a))^{-1}   \ dz \in \mathcal A_{-\hbar}^{\infty}. \]
The irrational rotation algebra $\mathcal A_{\hbar}^{\infty}$ possesses a unique normalized trace\footnote{since the weak closure of $A_{\hbar}$ is a (hyperfinite) type $\Pi_1$ factor.} \cite[Prop.\@ $2.3$,$2.4$]{Sh94} which therefore agrees with the trace $\widetilde{\operatorname{tr}}$ we use in this article.
The $K_0$ group of the irrational rotation algebra is given by $K_0(\mathcal A_{\hbar}) = \mathbb Z+\hbar \ \mathbb Z $ \cite{PV1,PV2}. Moreover, there exists a distinguished projection \cite{R81}, the so-called \emph{Powers-Rieffel} projection $\operatorname{P}_R$, which together with the identity generate the $K_0$ group. 
The inclusion of $K_0$ groups of the dense subalgebra $\mathcal A_{\hbar}^{\infty}$ into the one of $\mathcal A_{\hbar}$ is an isomorphism \cite[App.\@ $3$, Prop.\@ 2a]{C85} which implies that the above results remain true for $\mathcal A_{\hbar}^{\infty}$ as well.
 
This implies that for any projection $P \in \mathcal A^{\infty}_{\hbar}$
\begin{equation}
\label{eq:Rieffel}
\widetilde{\operatorname{tr}}_{\ZZ^2}( P ) = \gamma_1\widetilde{\operatorname{tr}}_{\ZZ^2}( \operatorname{id} ) + \gamma_2 \widetilde{\operatorname{tr}}_{\ZZ^2}( \operatorname{P}_R ) = \gamma_1+\gamma_2 \hbar.
\end{equation}
In the language of noncommutative geometry our trace $\tau_0:=\widetilde{\operatorname{tr}}_{\ZZ^2}$ is called the \emph{$0$-cocycle}. For the quantum Hall effect the \emph{$2$-cocycle} $\tau_2$ with $a_0,a_1,a_2 \in \mathcal A^{\infty}_{\hbar}$ is of particular importance
\begin{equation}
\label{eq:cocycle}
\tau_2(a_0,a_1,a_2) := \tau_0(a_0(\delta_1(a_1)\delta_2(a_2)-\delta_2(a_1) \delta_1(a_2))) 
\end{equation}
with derivations
\begin{equation}
\label{eq:derv}
\delta_1(\tau_{\gamma}^h):= i \gamma_1 \tau_{\gamma}^h \text{ and } \delta_2(\tau_{\gamma}^h):=i \gamma_2  \tau_{\gamma}^h.  
\end{equation}
In particular, we write $\Theta(a_0):=\tau_2(a_0,a_0,a_0)$ and will revisit $\Theta$ in the Kubo-Chern formula for the Hall conductance.
It follows then from \cite[Cor.\@ 16 in Ch.\@ III Sec.\@ 3]{C94} (see also \cite[p. 359]{C94}) that for any $a_0 \in K_0(\mathcal A_{\hbar}^{\infty})$ one has 
\begin{equation}
\label{eq:gamma2}
\Theta(a_0)= 2\pi i\gamma_2
\end{equation}
where $\gamma_2 \in \mathbb Z$ coincides with the eponymous integer in \eqref{eq:Rieffel}.

The semiclassical description of the DOS in Theorem \ref{theo1} implies together with the results from the previous paragraph, the following Proposition\footnote{We gauge the Hall conductivity for the hexagonal lattice in such a way that a full band has Hall conductivity zero.}:
\begin{prop}[Quantum Hall effect]
\label{QHEprop}
Let $h>0$ be small enough and consider zero disorder, i.e.\@ $\lambda=0$. The Hall conductivity is then in the spectral gaps between closed Landau bands \eqref{eq:Bands} for the discrete Schr\"odinger operators $H^h$ given by
\begin{equation}
\begin{split}
c_H(H^h(a_{{\scriptscriptstyle{ \blacksquare}}}),\mu) &= \tfrac{n}{2\pi}, \ \mu\text{ between }B_{{\scriptscriptstyle{ \blacksquare}},n}\ \& \ B_{{\scriptscriptstyle{ \blacksquare}},n+1}\text{ with }n \in \left\{1,..,N_{\scriptscriptstyle\blacksquare}(h) \right\} \text{ and }  \\
c_H(H^h(a_{\tinyvarhexagon}),\mu) &= \begin{cases} & \tfrac{2n+1}{2\pi}, \ \mu\text{ between }B_{\tinyvarhexagon,n}\ \& \ B_{\tinyvarhexagon,n+1}\text{ with }0 \le n \le N_{\scriptscriptstyle\tinyvarhexagon}(h) \\
&\frac{2n-1}{2\pi}, \ \mu\text{ between }B_{\tinyvarhexagon,n-1}\ \& \ B_{\tinyvarhexagon,n}\text{ with }0 \ge n \ge -N_{\scriptscriptstyle\tinyvarhexagon}(h).
\end{cases}
\end{split}
\end{equation}
\end{prop}
\begin{proof}
We just have to find the integer-valued coefficients in \eqref{eq:Rieffel} which we can obtain from the semiclassical expressions for the DOS in Theorem \ref{theo1}. Since Theorem \ref{theo1} does not allow us immediately to study spectral projections $\indic_{I}(H^h_{\lambda,\omega})$ we use smooth cut-off functions $\widetilde{\indic}_I(H^h_{\lambda,\omega})$ that coincide with the indicator function in the Landau bands and decay to zero in the spectral gaps (the DOS is supported on the spectrum, only). 
Theorem \ref{theo1} implies that for Fermi energies $\mu$ between Landau bands
\begin{equation}
\begin{split}
&\widetilde \tr_{\Lambda} (\indic_{(-\infty,\mu]} ( H^h_{{\scriptscriptstyle{ \blacksquare}}} )) =\tfrac{h}{2\pi} \sum_{ n \in \NN } 
\indic_{(-\infty,\mu]}( z_{n} ( h  ))+ \mathcal O(h^{\infty})  \\
&\widetilde \tr_{\Lambda} (\indic_{[0,\mu]} ( H^h_{\tinyvarhexagon} )) =\tfrac{h}{\pi \vert \vec{b}_1\wedge \vec{b}_2 \vert} \sum_{ n \in \ZZ } 
\indic_{[0,\mu]}( z_{n} ( h  ))+ \mathcal O(h^{\infty}). 
\end{split}
\end{equation}

Since the Hall conductivity is constant in spectral gaps and continuous in the magnetic field, the $\mathcal O(h^{\infty})$ error term in Theorem \ref{theo1} does not contribute to \eqref{eq:Rieffel}. We therefore find in \eqref{eq:Rieffel} that $\gamma_1=0$ and
\begin{equation}
\begin{split}
\gamma_{2,{\scriptscriptstyle{ \blacksquare}}} &= n, \quad \mu\text{ between }B_{{\scriptscriptstyle{ \blacksquare}},n}\ \& \ B_{{\scriptscriptstyle{ \blacksquare}},n+1}\text{ with }n \in \left\{1,..,N_{\scriptscriptstyle\blacksquare}(h) \right\}  \\
\gamma_{2,\tinyvarhexagon} &= \begin{cases}  &2n+1, \quad \mu\text{ between }B_{\tinyvarhexagon,n}\ \& \ B_{\tinyvarhexagon,n+1}\text{ with }0 \le n \le N_{\scriptscriptstyle\tinyvarhexagon}(h) \\
& 2n-1, \quad \mu\text{ between }B_{\tinyvarhexagon,n-1}\ \& \ B_{\tinyvarhexagon,n}\text{ with }0 \ge n \ge -N_{\scriptscriptstyle\tinyvarhexagon}(h).
\end{cases}
\end{split}
\end{equation}

\end{proof}

Let us recall how the Hall conductivity relates to the geometric framework of condensed matter physics \cite{B84}, see also \cite{S83}, following the construction in \cite[p.\@237+238]{C94}:
We study the algebra $\Omega^*:=\mathcal A^{\infty}_{\hbar} \otimes \wedge^* \mathbb C^2$. Using derivations \eqref{eq:derv}, we can define the differentials
\begin{equation}
\begin{split}
&d(a \otimes \alpha)  := \delta_1 (a) \vec{e}_1 \wedge \alpha + \delta_2(a)  \vec{e}_2 \wedge \alpha  \\
&d\left( a_1 \otimes  \vec{e}_1 +  a_2 \otimes  \vec{e}_2 \right) = (\delta_1(a_2)-\delta_2(a_1)) \otimes  \vec{e}_1 \wedge  \vec{e}_2.
\end{split}
\end{equation}
For forms of top degree there is the trace $\int:\Omega^{*2} \rightarrow \mathbb C$ given by $\int a \otimes ( \vec{e}_1 \wedge  \vec{e}_2) =a_{00}.$
Let $p \in\mathcal A^{\infty}_{\hbar}$ be a projection with module $M^{\infty}:=p\mathcal A^{\infty}_{\hbar}.$
For $m \in M^{\infty}$ and $a \in \mathcal A^{\infty}_{\hbar}$ we define connections \emph{(Berry connections)} $\nabla_i: M^{\infty} \rightarrow M^{\infty}$
\[ \nabla_i(\xi a)= \nabla_i(\xi)a + \xi \ \delta_i(a) :=p \ \delta_i(\xi) \ a +\xi \ \delta_i(a), \ i \in \left\{1,2 \right\}.\]
The curvature tensor \emph{(Berry curvature)}, is then defined as $R:= [\nabla_1,\nabla_2]\otimes ( \vec{e}_1 \wedge  \vec{e}_2).$

The first Chern number \emph{(Berry phase)} is an invariant of the module, independent of the connection, defined by $\operatorname{Ch}(p):= (2\pi i)^{-1}\int R = (2\pi i)^{-1}\Theta(p).$ 

With this vocabulary at hand, we now come to an equivalent second definition of the Hall conductivity:

\begin{defi}[Kubo-Chern formula]
\label{KC}
Let $\mu$ be an energy in an a.s.\@ spectral gap of $A^h(a_{\lambda,\omega})$ with associated spectral projection $P_A:=\indic_{(-\infty,\mu]}(A^h(a_{\lambda,\omega}))$, then the conductivity tensor $(\sigma_{jk})_{jk} \in \mathbb C^{2 \times 2}$ satisfies
\[ \sigma_{jk} := -i  \ \widetilde{\operatorname{tr}}_{\ZZ^2}\left( P_A[[P_A,x_{j}],[P_A,x_{k}]]   \right) =-i \mathbb E \left[ \Theta(P_A) \right]. \]
\end{defi}

The following Proposition states that the definitions of the Hall conductivity by the Kubo-Chern and St\v{r}eda formula yield the same result and are the same for all equivalent versions of the (random) DML:
\begin{prop}
\label{prop:cK}
Let $I$ be an interval such that $\partial I$ is in an a.s.\@ spectral gap of $A^h(a_{\lambda,\omega})$ and let $P_A:=\indic_{I}(A^h(a_{\lambda,\omega}))$, then the St\v{r}eda formula agrees with the off-diagonal conductivity in the Kubo-Chern formula
\[ D_h \widetilde{\operatorname{tr}}_{\ZZ^2}(P_A)   =  -i  \ \widetilde{\operatorname{tr}}_{\ZZ^2}\left( P_A[[P_A,x_1],[P_A,x_2]]   \right) = -i \Theta(P_A) .\]
Moreover, let $P_{H^h_{\lambda,\omega}}(I):=\indic_{I}(H^h_{\lambda,\omega})$ be the Fermi projection of $H^h_{\lambda,\omega}$, the Kubo-Chern formulas of projections coincide for $X_i(\gamma_1\vec{b}_1+\gamma_2 \vec{b}_2+r_j):=\gamma_i$
\begin{equation}
\begin{split}
 \widetilde{\operatorname{tr}}_{\Lambda} \left(  P_{H}[[P_{H},X_1],[P_{H},X_2]]  \right)
  &=| \vec{b}_1 \wedge \vec{b}_2 |^{-1} \ \widetilde{\operatorname{tr}}_{\ZZ^2} \left(  P_{A}[[P_{A},x_1],[P_{A},x_2]]  \right).
  \end{split}
  \end{equation}

\end{prop}
\begin{proof}
The first part of the Proposition, follows from the noncommutative framework and a direct computation can be found in \cite[Theorem $7$]{ST}.\footnote{The different sign compared with \cite[(51)]{ST} is due to a different sign convention that we use for magnetic matrices.} The second part follows as $UH^h_{\lambda,\omega}=A^h(a_{\lambda,\omega})U$ for a unitary multiplication operator $U$, by Lemma \ref{lem:unitmul},
\begin{equation}
\begin{split}
&| \vec{b}_1 \wedge \vec{b}_2 |\widetilde{\operatorname{tr}}_{\Lambda} \left(  P_H[[P_H,X_1],[P_H,X_2]] \right) \\
&=\mathbb E \  \operatorname{tr} \left( \left\langle U^*\delta_0, P_H[[P_H,X_1],[P_H,X_2]]  U^* \delta_0 \right\rangle\right) \\ 
&=\mathbb E \ \operatorname{tr}_{\CC^n} \left( \left\langle \delta_0, P_A[[P_A,x_1],[P_A,x_2]] \delta_0 \right\rangle\right) \\ 
&=\widetilde{\operatorname{tr}}_{\ZZ^2} \left(   P_A[[P_A,x_1],[P_A,x_2]]   \right). 
\end{split}
\end{equation}
\end{proof}
Finally, we shall use a third way of expressing the Hall conductivity using the relative index of projections. This representation is due to Avron, Seiler, and Simon \cite{ASS}. The version used here can be found in \cite[Ch.\@$14.5$]{AW}.
\begin{defi}[Index-theoretic formulation]
Let $P_{\lambda,\omega}$ be an orthogonal projection on $\ell^2(\mathbb Z^2)$ satisfying the covariance relation $\tau_{\gamma}^h P_{\lambda,T_{\gamma}\omega} = P_{\lambda,\omega} \tau_{\gamma}^h$
 with translations \eqref{eq:deftau} such that
\begin{equation}
\label{eq:series}
\sum_{x \in \mathbb Z^2} \vert x \vert \left(\mathbb E \vert P_{\lambda,\omega}[0,x]\vert^3\right)^{1/3} < \infty. 
\end{equation}
Using unitary operators $(U_a\psi)(x):=e^{-i \theta_a(x)} \psi(x)$ with $\theta_a(x):=\operatorname{arg}(x-a) \in (-\pi,\pi],$\footnote{Here we use the obvious identification of $\mathbb R^2$ with $\mathbb C$.} 
the off-diagonal component of the conductivity tensor $\sigma_{1,2}$ is given by the almost sure and $a \in \mathbb T_2^*$ independent value of the relative index
\[ 2\pi \sigma_{1,2}=\operatorname{ind}(P_{\lambda,\omega},U_aP_{\lambda,\omega}U_a^*) = \mathbb E\operatorname{tr}(P_{\lambda,\omega}-U_a P_{\lambda,\omega} U_a^*)^3 \]
and coincides, if $P_{\lambda,\omega}$ is a spectral projection satisfying the conditions of Proposition \ref{prop:cK}, with the value given by the Kubo-Chern formula in Definition \ref{KC}. 
\end{defi}
\begin{rem}
The index theoretic formulation implies that the Hall conductivity is integer-valued \emph{(up to the prefactor $(2\pi)^{-1}$)} under disorder, too. This follows of course also from the Kubo-Chern formula using the approach presented in \cite{BES}. 
\end{rem}
The index theoretic formulation of the Hall conductivity implies that the Hall conductivity is invariant, see Proposition \ref{QHEprop}, under mild disorder in the spectral gaps between closed disorder-broadened Landau bands:
\begin{proof}[Proof of Proposition \ref{QHEran}]
Consider a Fermi level $\mu$ between disorder-broadened Landau bands $\mathcal B_{n,\lambda}$ and $\mathcal B_{n+1,\lambda},$ i.e.\@ $\mu$ is in a spectral gap of $A^h(a_{\lambda,\omega}).$
We need to show that for Fermi projections $P_{\lambda,\omega}:=\indic_{(-\infty,\mu]}(A^h(a_{\lambda,\omega}))$ and $\lambda$ sufficiently close to zero, we have almost sure equality
\begin{equation}
\label{eq:Plambda}
 \operatorname{ind}(P_{\lambda,\omega},U_aP_{\lambda,\omega}U_a^*) = \operatorname{ind}(P_{0,\omega},U_aP_{0,\omega}U_a^*).
 \end{equation}
By the resolvent identity and holomorphic functional calculus we find for the difference
\[P_{\lambda,\omega}-P_{0,\omega} = \frac{\lambda}{2\pi i}  \oint_{(-\infty,\mu]} (A^h(a)-z)^{-1} V  (A^h(a_{\lambda,\omega})-z)^{-1} \ dz \]
which implies that $\lim_{\lambda \downarrow 0} P_{\lambda,\omega}x = P_{0,\omega}x$ by dominated convergence, which can be argued using the usual Combes-Thomas estimate for the pointwise bound.

Let $T_{\lambda,\omega} =P_{\lambda,\omega}-U_a P_{\lambda,\omega}U_a^*$ be the difference operator, we then find
\begin{equation}
\begin{split}
&\left\lvert \operatorname{ind}(P_{\lambda,\omega},U_aP_{\lambda,\omega}U_a^*) - \operatorname{ind}(P_{0,\omega},U_aP_{0,\omega}U_a^*) \right\rvert = \left\vert \operatorname{tr}(T_{\lambda,\omega}^3)-\operatorname{tr}(T_{0,\omega}^3)\right\vert \\
&\le  \left\vert \sum_{\vert \gamma \vert \le n} \operatorname{tr}_{\mathbb C^n}\left\langle \delta_{\gamma}, (T_{\lambda,\omega}^3-T_{0,\omega}^3) \delta_{\gamma} \right\rangle \right\vert + \left\vert \sum_{\vert \gamma \vert >n} \operatorname{tr}_{\mathbb C^n}\left\langle \delta_{\gamma}, (T_{\lambda,\omega}^3-T_{0,\omega}^3) \delta_{\gamma} \right\rangle \right\vert. 
\end{split}
\end{equation}
It suffices to argue that for $\lambda$ small, the difference of indices is less than one almost surely to show \eqref{eq:Plambda}. 
The first term on the right hand side is continuous in $\lambda$ by strong convergence and can therefore (for any fixed threshold $n$) be made arbitrarily small by taking $\lambda$ small enough.
Thus, by H\"older's inequality we find for the second term
\begin{equation}
\begin{split}
\sup_{\lambda \in (0,\lambda_0)} \left\vert \sum_{\vert \gamma \vert >n} \left\langle \delta_{\gamma}, T_{\lambda,\omega}^3 \delta_{\gamma} \right\rangle \right\vert \le \left\lVert T_{\lambda,\omega} \right\rVert_{\mathcal L^3}^2 \left\lVert T_{\lambda,\omega} \delta_{\vert \gamma\vert > n} \right\rVert_{\mathcal L^3}.
\end{split}
\end{equation}
We can then use the elementary identity 
\[\left\vert e^{-i \theta_{\alpha}(x)}-e^{-i\theta_{\alpha}(x+y)} \right\vert = \left\vert e^{-i \theta_{\alpha}(x)}-e^{-i\theta_{\alpha-y}(x)} \right\vert \le \operatorname{min} \left\{ 2, \frac{\vert y \vert}{\sqrt{\vert x-\alpha \vert \vert x+y -\alpha \vert}} \right\},\] see \cite[(14.24)]{AW}, to estimate \cite[Lemma $14.3$ and (14.27)]{AW} 
\begin{equation}
\begin{split}
\label{eq:bulk}
\mathbb E\left\lVert T_{\lambda,\omega} \delta_{\vert \gamma\vert > n} \right\rVert_{\mathcal L^3} 
&\lesssim \sum_{y \in \mathbb Z^2} \mathbb E\left( \sum_{\vert x \vert >n} \left\lvert T_{\lambda,\omega}[x+y,x] \right\vert^3 \right)^{1/3}  \\
&\lesssim \sum_{y \in \mathbb Z^2} \left(\sum_{\vert x \vert >n} \mathbb E \left\lvert P_{\lambda,\omega}[x+y,x]\right\vert^3 \left\lvert e^{-i \theta_{\alpha}(x+y)}-e^{-i\theta_{\alpha}(x)} \right\vert^3 \right)^{1/3} \\
&\lesssim \sum_{y \in \mathbb Z^2} \left( \mathbb E \left\lvert P_{\lambda,\omega}[y,0]\right\vert^3 \right)^{1/3}  \left(\sum_{\vert x \vert >n} \left\lvert e^{-i \theta_{\alpha}(x+y)}-e^{-i\theta_{\alpha}(x)} \right\vert^3 \right)^{1/3}<\infty.
\end{split}
\end{equation}
The standard Combes-Thomas estimate implies that \eqref{eq:series} is uniformly bounded for $\lambda \in (0,\lambda_0)$. This implies that the summand in \eqref{eq:bulk} is uniformly bounded and by the dominated convergence theorem, this expression goes to zero as $n \rightarrow \infty.$ 
\end{proof}

\section{The metal/insulator transition}
\subsection{Measures of transport}
\label{sec:mit}
For our discussion of metal/insulator transitions, we first recall the definition of transport coefficients stated in \cite{GK3}. Even though the results in that article are stated for non-magnetic Schr\"odinger operators, the results still apply to(discrete) magnetic Schr\"odinger operators as the authors state in the beginning of Section $4$ in \cite{GKS}.
Dynamical properties are studied using weighted norms
\[ M^h_{\lambda,\omega}(p,\zeta,t) = \left\lVert \langle x \rangle^{p/2} e^{-itH^h_{\lambda, \omega}} \zeta(H^h_{\lambda, \omega}) \delta_0 \right\rVert_{\mathcal L^2}^2  \]
where $\zeta \in C_{c,+}^{\infty}(\mathbb R)$ localizes to a fixed energy window. In particular, we say that at energies $E$, $H^h_{\lambda,\omega}$ exhibits Hilbert-Schmidt localization if there is an open interval $I\ni E$ such that for all $\zeta \in C_{c,+}^{\infty}(I)$ and all $p > 0$
\[ \mathbb E \left[ \sup_{t \in \mathbb R} M^h_{\lambda,\omega}(p, \zeta,t) \right]< \infty. \]
The union of all such energies comprises the set $\Sigma^{h,\text{loc}}_{\lambda}$. We also define expected time-C\'esaro averages
\[ M^h_{\lambda}(p,\zeta,T) = \frac{1}{T} \int_0^{\infty} \mathbb E \left(M^h_{\lambda,\omega}(p,\zeta,t)\right) e^{-t/T} \ dt. \]
The \emph{(lower) transport exponent} is defined by 
\[ \beta^h_{\lambda}(p,\zeta) = \liminf_{T \rightarrow \infty} \frac{\log_{+} M^h_{\lambda}(p,\zeta,T)}{p \log(T)} , \text{ for } p > 0, \zeta \in C_{c,+}^{\infty}(\mathbb R)\]
and from this one defines the \emph{$p$-th local transport exponent} 
\[\beta^h_{\lambda}(p,E) =\inf_{ I \ni E} \sup_{\zeta \in C_{c,+}^{\infty}(I)} \beta^h_{\lambda}(p,\zeta) \in [0,1]. \]
The local lower transport exponent is then defined as $\beta^h_{\lambda}(E):= \sup_{p>0}\beta^h_{\lambda}(p,E).$
The exponent $\beta^h_{\lambda}(E)$ is a measure of transport at energy $E.$
This coefficient allows us to define two complementary regions, the (relatively open) \emph{region of dynamical localization} or \emph{insulator region}
\begin{equation}
\label{eq:DL}
\Sigma_{\lambda}^{h,\text{DL}} = \left\{ E \in \mathbb R; \beta^h_{ \lambda}(E)=0 \right\}
\end{equation}
that coincides with $\Sigma^{h,\text{loc}}_{\lambda}$ \cite[Theorem 2.8]{GK3}, and the (relatively closed) \emph{region of dynamical delocalization} or \emph{metallic transport region} 
\begin{equation}    
\label{eq:DD}
\Sigma_{\lambda}^{h,\text{DD}} = \left\{ E \in \mathbb R; \beta^h_{ \lambda}(E)>0 \right\}.
\end{equation}
An energy $E$ at which the transport coefficient $\beta^h_{\lambda}$ jumps from zero to a non-zero value is called a \emph{mobility edge}.
\begin{rem}
 \cite[Theorem 2.10]{GK3} implies that in two dimensions, the random Schr\"odinger operator $H_{\lambda,\omega}^h$ has the property that for all $E \in \mathbb R$ for which the transport exponent is positive $\beta^h_{\lambda}(E)>0$, the coefficient satisfies already $\beta^h_{\lambda}(E)>1/4.$ 
\end{rem}
Fix $\varepsilon>0$ and let $K$ be the multiplication operator by $\langle x \rangle^{1+\varepsilon}$. The random measure of $H^h_{\lambda,\omega}$ is defined for Borel sets $B \subset \mathbb R$ by $\mu_{\lambda,\omega}(B):= \left\lVert K^{-1} \indic_{B}(H^h_{\lambda,\omega}) \right\rVert_{\mathcal L^2}^2$, is supported on the spectrum of $H^h_{\lambda,\omega}$, such that $\mu_{\lambda,\omega}(B)< \infty$ if $B\subset \Sigma(H^h_{\lambda,\omega})$ is bounded.

Whenever the multiscale analysis in \cite{GK6}, which applies to magnetic Schr\"odinger operators, as explained in the beginning of their Section $2$, applies to energies in the region of dynamical localization, this has a strong implication on the eigenfunctions that the authors call \emph{summable uniform decay of eigenfunction correlations} ({SUDEC}), see \cite[Cor.\@ $3$]{GK6}, which we recall in the following Definition:
\begin{defi}[\textbf{SUDEC}]
For a bounded interval $I$ with $\overline{I} \subset  \Sigma_{\lambda}^{h,\text{DL}}(H^h_{\lambda,\omega})$, we say that $H^h_{\lambda,\omega}$ exhibits \emph{{SUDEC}} in $I$ if the spectrum of $H^h_{\lambda,\omega}$ is a.s. pure point and for each eigenvalue $E_{n, \omega,\lambda} \in I$ there is an ONB $(\phi_{n,j,\lambda,\omega})_{j \in \left\{1,..., \nu_{n,\lambda,\omega} \right\}}$ of the finite-dimensional eigenspace $\operatorname{ker}\left(H^h_{\lambda,\omega}-E_{n, \omega,\lambda}\right)$ such that for any $\xi \in (0,1)$ there is $C_{I,\lambda,\omega,\xi}>0$ such that  
\begin{equation}
\left\lVert \phi_{n,i,\lambda,\omega}(x) \right\rVert \left\lVert \phi_{n,j,\lambda,\omega}(y) \right\rVert \le C_{I,\xi,\omega,\lambda} \sqrt{\alpha_{n,i,\lambda,\omega}} \sqrt{\alpha_{n,j,\lambda,\omega}} \langle x \rangle^{1+\varepsilon}\langle y \rangle^{1+\varepsilon} e^{-\vert x-y\vert^{\xi}}.
\end{equation}
Moreover, $\sum_{n \in \mathbb N, j \in \left\{1,2,...,\nu_{n,\lambda,\omega} \right\}} \alpha_{n,j, \lambda,\omega} = \mu_{\lambda,\omega}(I). $
\end{defi}
It follows from standard arguments that the operator $H^h_{\lambda,\omega}$, and equivalently $A^h(a_{\lambda,\omega})$ satisfy SUDEC in the regime of dynamical localization.
\subsection{Dynamical delocalization}
\label{sec:DLoc}
We now turn to the proof of Theorem \ref{Deloc} showing that between disjoint disorder-broadened Landau bands there exists a mobility edge.

We study covariant projections that satisfy the following condition: 
\begin{defi}[\textbf{P}]
\label{def:P}
A covariant projection on $\ell^2(\mathbb Z^2; \mathbb C^{n})$ is said to satisfy condition $\emph{{(P)}}$ if for constants $\xi \in (0,1)$, $k>0$, and $K_P< \infty$ the following bound holds
\[\Vert P[0,x]  \Vert=\Vert  \langle \delta_0, P \delta_x \rangle \Vert \le K_P \langle x \rangle^k e^{-\vert x \vert^{\xi}}. \] 
\end{defi}
Clearly, for covariant eigenprojections $P_{\lambda,\omega}:=\indic_{E_{n,\omega,\lambda}}(A^h(a_{\lambda,\omega}))$ on a single energy, ({SUDEC}) implies $({P})$ with $k=1+\varepsilon$ and $
K_P:= C_{I,\xi,\omega,\lambda} \sum_{i=1}^{\nu_{n,\lambda,\omega} }\alpha_{n,i, \lambda,\omega} .$

The index formulation of the Hall conductivity implies immediately by the cyclicity of the trace that if $P$ is a covariant finite-rank projection satisfying \eqref{eq:series} then
\begin{equation}
 \operatorname{ind}(P_{\lambda,\omega},U_aP_{\lambda,\omega}U_a^*) = \operatorname{tr} \left(P_{\lambda,\omega}-U_aP_{\lambda,\omega}U_a^*\right)=0.
 \end{equation}
Moreover, for two orthogonal covariant projections satisfying sufficient decay properties one finds that \cite[Sec.\@E Lem.\@$12$]{BES} for $\Theta$ as in Definition \ref{KC}
\begin{equation}
\label{eq:additivity}
 \Theta(P+Q) = \Theta(P)+\Theta(Q). 
 \end{equation}
\begin{lemm}
Let $P$ be a covariant projection satisfying condition \emph{{(P)}}. Then the quantity $\Theta(P)$ is finite and is bounded for any $\xi \in (0,1)$ by a finite constant $C_{\xi, \kappa}>0$
\[ \Vert \mathbb E \langle \delta_0, P[[P,x_1],[P,x_2]] \delta_0 \rangle \Vert \le K_P C_{\xi, \kappa}.\]
\end{lemm}
\begin{proof}
Condition ${(P)}$ implies the following bound
\begin{equation}
\begin{split}
\Vert \mathbb E \langle \delta_0, &P[[P,x_1],[P,x_2]] \delta_0 \rangle_{\ell^2}  \Vert_{\CC^n}  = \Vert \langle \mathbb E \langle[[x_1,P], P]\delta_0, [x_2,P] \delta_0 \rangle_{\ell^2} \Vert_{\CC^n} \\
 &\le   \sqrt{ \mathbb E \Vert [[x_1,P], P] \delta_0 \Vert_{\ell^2}^2} \sqrt{\mathbb E \Vert x_2 P \delta_0 \Vert_{\ell^2}^2} \lesssim   \sqrt{ \mathbb E \Vert x_1 P \delta_0 \Vert_{\ell^2}^2} \sqrt{\mathbb E \Vert x_2 P \delta_0 \Vert_{\ell^2}^2} \\ 
  &\lesssim    \mathbb E \Vert x_1 P \delta_0 \Vert_{\ell^2}^2 +  \mathbb E \Vert x_2 P \delta_0 \Vert_{\ell^2}^2  \lesssim \sum_{x \in \mathbb Z^2} \left\lVert x \right\rVert_{\CC^n}^2 \mathbb E \Vert \langle \delta_0, P \delta_x \rangle_{\ell^2} \Vert^2 \\
 & \lesssim K_P^2 \sum_{x \in \mathbb Z^2} \Vert x \Vert_{\CC^n}^{2(1+k)} e^{-2\Vert x \Vert^{\xi}} \lesssim K_P^2 C_{\xi,\kappa}^2.
 \end{split}
\end{equation}
\end{proof}

We can now finish the proof of Theorem \ref{Deloc}:
\begin{proof}[Proof of Theorem \ref{Deloc}]
Let us assume that $H^h_{\lambda,\omega}$ would have only spectrum belonging to the region of dynamical localization. For an interval $I=[\lambda_1,\lambda_2]$ where $\lambda_1$ is in one spectral gap between disorder-broadened Landau bands and $\lambda_2$ in another such gap, it follows for $\mathcal E_{\lambda,\omega}$ the set of eigenvalues of $H^h_{\lambda,\omega}$ in $I$ and $\mathcal E_{\lambda,\omega} = \bigcup_{m \in \mathbb N} \mathcal M_m $ with $\mathcal M_m$ a subset of $\mathcal E_{\lambda,\omega}$ of cardinality $\operatorname{min}\left\{m,\operatorname{dim}\left( \operatorname{ran}(\indic_{I}(H^h_{\lambda,\omega})\right) \right\}$ 
\begin{equation}
\begin{split}
\Theta(\indic_I(A^h(a_{\lambda,\omega}))) = \underbrace{\sum_{E_{n,\lambda,\omega} \in \mathcal M_m} \Theta (\indic_{E_{n,\lambda,\omega}}(A^h(a_{\lambda,\omega}))) }_{=0}+ \Theta(\indic_{ \mathcal E_{\lambda,\omega} \backslash \mathcal M_m}(A^h(a_{\lambda,\omega})))
\end{split}
\end{equation}
which vanishes by letting $m \rightarrow \infty$ due to ({SUDEC}) and Definition \ref{def:P}.
Hence, the Hall conductivity must not jump for operators $H^h_{\lambda,\omega}$ which contradicts the findings of Proposition \ref{QHEran}.
\end{proof}
\begin{rem}
To prove delocalization, the type of disorder was in so far irrelevant, as we only assumed the disorder to be small. Other discrete models to which this argument applies are discussed in \cite[Remark $3.13$]{GK}.
\end{rem}

\section{Honeycomb structures with flux close to a rational}
Hitherto, we studied the case of small magnetic flux $h>0$ on both the square and hexagonal lattice. We will now continue by studying small magnetic perturbations of rational magnetic fluxes $2\pi p/q$ for the hexagonal lattice, see \cite{HS0} for a similar analysis in case of Harper's model.

We start by showing the existence of Dirac cones for rational flux $\phi=2\pi p/q$ for $H_{\tinyvarhexagon}^{\phi}$ at energy level $0.$ In the sequel, we write $\phi$ for the magnetic flux and use the variable $h$ to denote small perturbations thereof.
\subsection{Dirac points}
\label{sec:Diracpoints}
For magnetic flux $\phi=2\pi p/q$, $H^{\phi}_{\tinyvarhexagon}$ is a periodic operator. 
Let $ k=(k_1,k_2)\in \mathbb{T}_2^*$, and let $H^{\phi}_{\tinyvarhexagon}( k)$ be the operator $H^{\phi}_{\tinyvarhexagon}$ on $\ell^2(\Lambda)$ subject to the pseudo-periodic condition: 
$$z(\gamma+q \vec{b}_l,  r_j)=e^{ik_l} z(\gamma,  r_j),\ \ j,l=1,2$$ where $\{\vec{b}_1, \vec{b}_2\}$ is the basis vector of $\Lambda$ and $\{r_0,r_1\}$ are the vertices in the fundamental domain $W_{\Lambda}$.

We say that an energy $E$ corresponding to some quasi-momentum $\tilde{ k}$ in the dispersion surface of $H^{\phi}_{\tinyvarhexagon}$ is a \emph{Dirac point}, if in a neighbourhood of such quasi-momentum, for some positive $c>0$, there are two distinct branches of eigenvalues $F_{\pm}(H_{\tinyvarhexagon}^{\phi}({ k}))$ such that
\begin{equation}
\label{def:Diracpoints}
\begin{split}
&F_{\pm}(H_{\tinyvarhexagon}^{\phi}(\tilde{k}))=E \ \ \text{ and }\\
&F_{\pm} (H^{\phi}_{\tinyvarhexagon}( k))- E= \pm c \vert { k}-\tilde{ k} \vert + \mathcal O( \vert { k-\tilde{k}} \vert^2). 
\end{split}
\end{equation}
Next we will present the proof of Theorem \ref{thm:Dirac}.

\begin{figure}[H]
 \includegraphics[height=6cm, width=13cm]{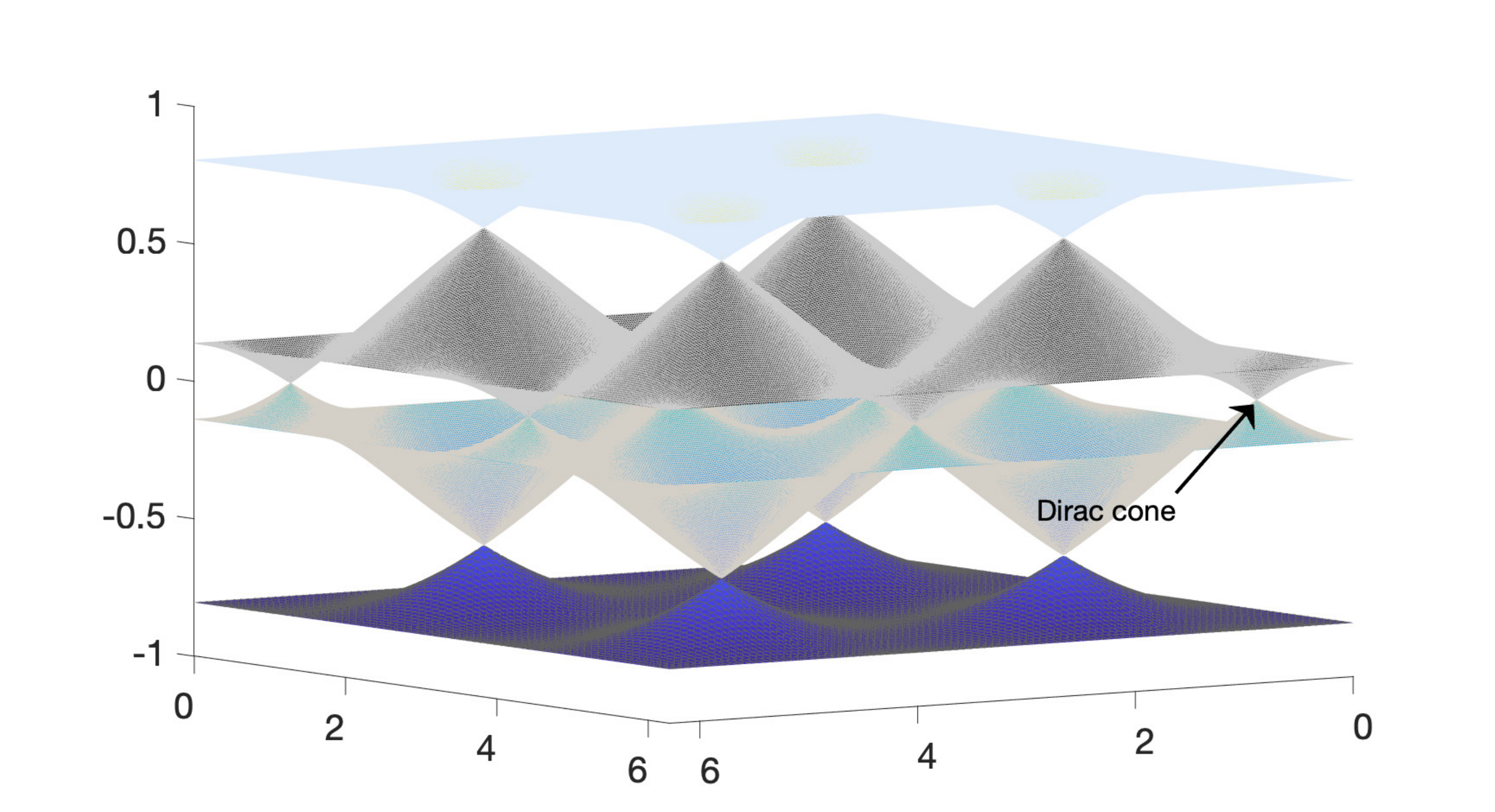}
 \caption{Dispersion surface of $H_{\varhexagon}^{\phi}$. The Dirac cones at energy level zero persist for magnetic flux $\phi=\pi$.}
 \label{fig:touching}
\end{figure}

\subsection*{Proof of Theorem \ref{thm:Dirac}}
The proof is built on some results of \cite{HKL16}. Recall $H^{\phi}_{\tinyvarhexagon}$ is a tight-binding Schr\"odinger operator with flux ${\phi}$ on the hexagonal lattice, 
acting on $\ell^2(\Z^2, \C^2)$.

The Floquet matrix of $H^{\phi}_{\tinyvarhexagon}(k)$ is 
\begin{align}\label{def:MG}
M_{\tinyvarhexagon}( k)=\frac{1}{3}\left(\begin{matrix} 0\ \ \ \ \ \ \ & I_q+e^{ik_1}J_{p,q}+e^{ik_2}K_q\\ I_q+e^{-ik_1} J_{p,q}^*+e^{-ik_2}K_q^*   &0\end{matrix}\right)=:\left(\begin{matrix} 0\ \ &\mathcal{A}\\ \mathcal{A}^* &0\end{matrix}\right),
\end{align}
where $J_{p,q}$ and $K_q$ are $q\times q$ matrices, which are defined as 
\begin{align}\label{def:Jpq}
J_{p,q}=\mathrm{diag}\left(\{e^{i(j-1)\phi}\}_{j=1}^q\right),
\end{align}
and
\begin{align}\label{def:Kq}
(K_q)_{jk}=
\begin{cases}
1\ \ \text{if}\ \ k\equiv j+1 (\mathrm{mod}\ q)\\
0\ \ \text{otherwise.}
\end{cases}
\end{align}
The solutions of the characteristic equation $\det(M_{\tinyvarhexagon}( k)-\lambda)=0$ are the Floquet eigenvalues of $H^{\phi}_{\tinyvarhexagon}( k)$, which we label in increasing order:
$$F_1( k)\leq F_2( k)\leq \cdots \leq F_{2q}( k).$$ 
Take $B_j:=\cup_{ k\in \mathbb{T}_2^*}F_j( k)$, $1\leq j\leq 2q$, to be the $j$-th spectral band of $H^{\phi}_{\tinyvarhexagon}$.
The following was shown in \cite{HKL16}.
\begin{prop}\label{prop:Bk_MG}
We have
\begin{itemize}
\item $\{B_j\}_{j=1}^{2q}$ are non-overlapping. 
\item $B_{q}\cap B_{q+1}=\{0\}$.
\end{itemize}
\end{prop}
The set $S_j:=\{( k, F_j( k)): k\in \mathbb{T}_2^*\}$ is called the $j$-th dispersion surface.

Taking the square of $M_{\tinyvarhexagon}( k)$, we arrive at 
\begin{align}\label{eq:MG^2}
M_{\tinyvarhexagon}^2( k)=\left(\begin{matrix} \mathcal{A}\mathcal{A}^*\ &0\\ 0 & \mathcal{A}^*\mathcal{A}\end{matrix}\right)
=\frac{1}{9}\left(\begin{matrix} 3I_q+M_T( k)\ &0\\ 0 & 3I_q+\widehat{M}_T( k)\end{matrix}\right),
\end{align}
where 
\begin{align}\label{def:MT}
\widehat{M}_T( k)=e^{ik_1}J_{p,q}+e^{-ik_1}J_{p,q}^*+&e^{ik_2}K_q+e^{-ik_2}K_q^* \notag\\
+&e^{i(k_1-k_2)}K_q^* J_{p,q}+e^{-i(k_1-k_2)}J_{p,q}^*K_q,
\end{align}
and for $M_T( k)$ one just exchanges $J_{p,q}$ and $K_q$.
Furthermore, $M_T( k)$ and $\widehat{M}_T( k)$ have the same non-zero eigenvalues.
Let us denote the eigenvalues of $M_T( k)$ by $\{E_j( k)\}_{j=1}^p$, where each $E_j$ is an analytic function in $ k$, note that we do not arrange them in increasing order here.
Clearly we have
\begin{align}\label{eq:det=prod}
\det(M_T( k)-\lambda)=\prod_{j=1}^q (E_j( k)-\lambda).
\end{align}
By \eqref{eq:MG^2}, $M_T(k)+3I_q$ is positive semidefinite, hence $E_j( k)\geq -3$ for $1\leq j\leq q$, and the following holds:
\begin{align}\label{eq:F=E}
\{F_m( k)\}_{m=q+1}^{2q}=\left\lbrace  \frac{1}{3} \sqrt{E_j( k)+3}\right\rbrace_{j=1}^q\ \ \text{and}\ \ 
\{F_m( k)\}_{m=1}^{q}=\left\lbrace - \frac{1}{3}\sqrt{E_j( k)+3}\right\rbrace_{j=1}^q.
\end{align}
By Proposition \ref{prop:Bk_MG}, one concludes that $-3\in \cup_{j=1}^q \cup_{ k\in \mathbb{T}_2^*} E_j( k)$.
Without loss of generality, let 
\begin{align}\label{def:tilde_bm_theta}
E_1(\tilde{ k})=-3.
\end{align}
Since the bands are non-overlapping, $E_1(\tilde{ k})$ must be a single eigenvalue, hence for $2\leq j\leq q$, we have $E_j(\tilde{ k})>-3.$
Now, since $-3$ is the minimal value of $E_1$, we have
\begin{align}\label{eq:partial_E1}
\frac{\partial E_1}{\partial k_m} (\tilde{ k})=0 \ \ \text{ for } m=1,2.
\end{align} 

The following Chambers formula was derived in \cite{HKL16}, see similar formulas in \cite{AEG}.
\begin{prop}
We have
\begin{align}\label{eq:Chambers_MT}
\det(M_T( k)-\lambda)=f_{p,q}(\lambda)+2(-1)^{q+1} (\cos qk_1+\cos qk_2+ (-1)^{q+1}\cos q(k_1-k_2)),
\end{align}
where $f_{p,q}(\lambda)$ is a polynomial in $\lambda$ (independent of $ k$) with leading coefficient $(-1)^q$.
\end{prop}
Clearly, this proposition yields that 
\begin{equation}
\begin{split}
&\det(M_T(k_1,k_2)-\lambda)=\det(M_T(k_1+\frac{2\pi}{q},k_2)-\lambda)=\det(M_T(k_1,k_2+\frac{2\pi}{q})-\lambda),\ \ \text{ and } \nonumber\\
&\det(M_T(k_1,k_2)-\lambda)=\det(M_T(-k_1,-k_2)-\lambda).
\end{split}
\end{equation}
Hence, we can restrict our attention to 
$$(k_1,k_2)\in \left[0,\frac{\pi}{q}\right)\times \left[-\frac{\pi}{q},\frac{\pi}{q}\right).$$

In the following, we denote 
\begin{align}\label{def:g_q}
2(-1)^q (\cos qk_1+\cos qk_2+ (-1)^{q+1}\cos q(k_1-k_2)):=g_q( k)
\end{align}
for simplicity.
A direct consequence of Chambers' formula \eqref{eq:Chambers_MT} is that
\begin{align}
\cup_{ k\in \mathbb{T}_2^*} \Sigma(M_T( k))=\{\lambda: \min_{ k\in \mathbb{T}_2^*} g_q( k)\leq f_{p,q}(\lambda)\leq \max_{ k\in \mathbb{T}_2^*} g_q( k)\}.
\end{align}
Use the fact that the energy $-3$ is the bottom of the spectrum $\cup_{ k\in \mathbb{T}_2^*} \Sigma(M_T( k))$, we have 
\begin{align}\label{eq:f-3=max}
f_{p,q}(-3)=\max_{ k\in \mathbb{T}_2^*} g_q( k).
\end{align}
Simple computations show that
\begin{align}\label{eq:max=3/2}
\max_{ k\in \mathbb{T}_2^*} g_q( k)=3.
\end{align}
Furthermore, for even $q$, the maximum is attained at 
\begin{align}\label{eq:even_max}
q k\in \{(\pi/3, -\pi/3), (-\pi/3,\pi/3)\}+2\pi \Z^2,
\end{align}
and for odd $q$, the maximum is attained at
\begin{align}\label{eq:odd_max}
q k\in \{(2\pi/3, -2\pi/3), (-2\pi/3, 2\pi/3)\}+2\pi \Z^2.
\end{align}

Plugging $ k=\tilde{ k}$ and $\lambda=-3$ into \eqref{eq:Chambers_MT}, using \eqref{eq:det=prod} and the fact that $E_1(\tilde{ k})=-3$, we have
\begin{align}
0=\prod_{j=1}^q (E_j(\tilde{ k})+3)=\det(M_T(\tilde{ k})+3)=f_{p,q}(-3)-g_q(\tilde{ k}).
\end{align}
Hence we have
\begin{equation}\label{eq:tildek}
\tilde{k}=\left(\frac{\pi}{3q},-\frac{\pi}{3q}\right) \text{ for even } q, \text{ and } \tilde{k}=\left(\frac{2\pi}{3q},-\frac{2\pi}{3q}\right) 
\text{ for odd } q.
\end{equation}

Differentiating \eqref{eq:det=prod} w.r.t. $k_j$, $j=1,2$, and taking \eqref{eq:Chambers_MT} into account, we have
\begin{align}\label{eq:MT_partial_theta}
\begin{cases}
2q (-1)^{q+1} (-\sin qk_1+(-1)^q \sin q(k_1-k_2))=
\sum_{m=1}^q \frac{\partial E_m}{\partial k_1}( k) \prod_{\substack{j=1\\ j\neq m}}^q (E_j( k)-\lambda)  \\
\\
2q (-1)^{q+1} (-\sin qk_2-(-1)^q \sin q(k_1-k_2))=
\sum_{m=1}^q \frac{\partial E_m}{\partial k_2}( k) \prod_{\substack{j=1\\ j\neq m}}^q (E_j( k)-\lambda)
\end{cases}
\end{align}
Differentiating \eqref{eq:MT_partial_theta} again w.r.t. $k_j$, $j=1,2$, we have
\begin{align}\label{eq:MT_Hessian_theta}
\begin{cases}
2q^2 (-1)^{q+1} (-\cos qk_1+(-1)^q \cos q(k_1-k_2))=\sum_{\substack{m, \ell=1\\ m\neq \ell}}^q 
\frac{\partial E_m}{\partial k_1}( k) \frac{\partial E_{\ell}}{\partial k_1}( k)
\prod_{\substack{j=1\\ j\neq m,\ell}}^q 
(E_j( k)-\lambda)\\
\qquad\qquad\qquad\qquad\qquad\qquad\qquad\qquad+\sum_{m=1}^q \frac{\partial^2 E_m}{\partial k_1^2}( k) \prod_{\substack{j=1\\ j\neq m}}^q (E_j( k)-\lambda) \\
\\
2q^2 \cos q(k_1-k_2)=\sum_{\substack{m, \ell=1\\ m\neq \ell}}^q 
\frac{\partial E_m}{\partial k_1}( k) \frac{\partial E_{\ell}}{\partial k_2}( k)
\prod_{\substack{j=1\\ j\neq m,\ell}}^q 
(E_j( k)-\lambda)\\
\qquad\qquad\qquad\qquad\qquad\qquad\qquad\qquad+\sum_{m=1}^q \frac{\partial^2 E_m}{\partial k_1 \partial k_2}( k) \prod_{\substack{j=1\\ j\neq m}}^q (E_j( k)-\lambda)\\
\\
2q^2 (-1)^{q+1} (-\cos qk_2+(-1)^q \cos q(k_1-k_2))=\sum_{\substack{m, \ell=1\\ m\neq \ell}}^q 
\frac{\partial E_m}{\partial k_2}( k) \frac{\partial E_{\ell}}{\partial k_2}( k)
\prod_{\substack{j=1\\ j\neq m,\ell}}^q 
(E_j( k)-\lambda)\\
\qquad\qquad\qquad\qquad\qquad\qquad\qquad\qquad+\sum_{m=1}^q \frac{\partial^2 E_m}{\partial k_2^2}( k) \prod_{\substack{j=1\\ j\neq m}}^q (E_j( k)-\lambda)
\end{cases}
\end{align}
We plug in $ k=\tilde{ k}$ and $\lambda=-3$. Using \eqref{def:tilde_bm_theta} and \eqref{eq:partial_E1}, we have
\begin{align}\label{eq:MT_Hessian_theta_at_-3}
\begin{cases}
2q^2 (-1)^{q+1} (-\cos q\tilde{k}_1+(-1)^q \cos q(\tilde{k}_1-\tilde{k}_2))=\frac{\partial^2 E_1}{\partial k_1^2}(\tilde{ k}) \prod_{j=2}^q (E_j(\tilde{ k})+3) \\
\\
2q^2 \cos q(\tilde{k}_1-\tilde{k}_2)=\frac{\partial^2 E_1}{\partial k_1 \partial k_2}(\tilde{ k}) \prod_{j=2}^q 
(E_j(\tilde{ k})+3)\\
\\
2q^2 (-1)^{q+1} (-\cos q\tilde{k}_2+(-1)^q \cos q(\tilde{k}_1-\tilde{k}_2))=\frac{\partial^2 E_1}{\partial k_2^2}(\tilde{ k}) \prod_{j=2}^q (E_j(\tilde{ k})+3)
\end{cases}
\end{align}
Hence the Hessian matrix
\begin{equation}
\begin{split}
&D^2_{k_1,k_2} E_1(\tilde{ k}) \\
&= \frac{2 q^2 (-1)^q}{\prod_{j=2}^q (E_j(\tilde{ k})+3)} \left(
\begin{matrix}
\cos q\tilde{k}_1-(-1)^q \cos q(\tilde{k}_1-\tilde{k}_2)\ \ &(-1)^q\cos q(\tilde{k}_1-\tilde{k}_2)\\
(-1)^q\cos q(\tilde{k}_1-\tilde{k}_2)   &\cos q\tilde{k}_2-(-1)^q \cos q(\tilde{k}_1-\tilde{k}_2)
\end{matrix}
\right)
 \end{split}
 \end{equation}
Plugging in the values of $\tilde{k}$, see \eqref{eq:tildek}, we see that the Hessian matrix for either case is the same: 
\begin{equation} 
\label{eq:Hessian}
D^2_{k_1,k_2}E_1(\tilde{ k})=\frac{2q^2}{\prod_{j=2}^q (E_j(\tilde{ k})+3)}
\left(
\begin{matrix}
1\ \ \ &-\frac{1}{2}\\
-\frac{1}{2}   &1
\end{matrix}
\right),
\end{equation}
which is a positive definite matrix. By doing symplectic change of variables
\begin{equation}
\begin{split}
 &y(k)= a \left(k_1+k_2\right), \ \eta(k) = b\left(k_2-k_1 + \frac{4\pi}{3q} \right) \text{ if q is odd, and} \\
 &y(k)= a \left(k_1+k_2\right), \ \eta(k) = b\left(k_2-k_1 + \frac{2\pi}{3q} \right) \text{ if q is even, where}\\
 &a = 2^{-1/2}3^{-1/4} \text{ and } \ b= 2^{-1/2}3^{1/4},
 \end{split}
\end{equation}
clearly $\tilde{y}:=y(\tilde{k})=0$ and $\tilde{\eta}:=\eta(\tilde{k})=0$. Let $\widetilde{E}_1(y,\eta):=E_1(k_1,k_2)$.
One then checks that using \eqref{eq:Hessian}
\begin{equation}
\begin{split}
D^2_{y,\eta}\, \widetilde{E}_1(0,0)&= \left( \frac{\partial (k_1,k_2)}{\partial (y,\eta)}(0,0) \right)^T  D^2_{k_1,k_2}\, E_1(\tilde{ k}) \left( \frac{\partial (k_1,k_2)}{\partial (y,\eta)}(0,0) \right) \\
&=\frac{\sqrt{3} q^2}{\prod_{j=2}^q (E_j(\tilde{ k})+3)}\left(\begin{matrix}1 & 0 \\ 0 &1 \end{matrix} \right) \\
&\text{ with } \left( \frac{\partial (k_1,k_2)}{\partial (y,\eta)}(0,0) \right) = \left(\begin{matrix}  2^{-1/2}3^{1/4} & 2^{-1/2}3^{-1/4} \\  2^{-1/2} 3^{1/4} &- 2^{1/2} 3^{-1/4} \end{matrix} \right).
\end{split}
\end{equation}
Thus, we have in new coordinates close to each well
\begin{equation}
\begin{split}
\widetilde{E}_1(y,\eta) = -3 +\frac{\sqrt{3} q^2}{2\prod_{j=2}^q (E_j(\tilde{ k})+3)} \left(y^2+\eta^2\right)  + \mathcal O( \Vert(y,\eta) \Vert^3).
\end{split}
\end{equation}
This yields for the hexagonal lattice using \eqref{eq:F=E} the Dirac cones
\begin{equation}
\begin{split}
\label{eq:well}
F_{q+1}(\tilde{ k}) = \frac{q}{3^{3/4}} \frac{1}{\sqrt{2\prod_{j=2}^q (E_j(\tilde{ k})+3)}} \Vert(y,\eta) \Vert  + \mathcal O( \Vert(y,\eta) \Vert^{2}).
\end{split}
\end{equation}
\qed

\subsection{Semiclassical analysis close to any rational}
In this subsection, we use variables $(x,\xi)$ instead of $k = (k_1,k_2)$ to emphasize the underlying phase space structure. This will generalise magnetic matrices in Def.\ref{Magneticmatrix} and their connection to pseudodifferential operators as in Def.\ref{pdotrace}. For the study of magnetic fluxes $\phi= 2\pi \tfrac{p}{q}+h$ with $\operatorname{gcd}(p,q)=1,$ we use that \cite[Sec.\@1]{HS20} there is a $C^*$-homomorphism mapping scalar-valued $\Psi$DOs with $\Z_{*}^2$-periodic Weyl symbol 
\[ \operatorname{Op}_{\phi}^{\text{w}}(\widehat{a_{\tinyvarhexagon}})= \sum_{\gamma \in \Z^2} a_{\tinyvarhexagon}(\gamma) \operatorname{Op}_{\phi}^{\text{w}} \left((x,\xi) \mapsto e^{i \langle (x,\xi),\gamma \rangle} \right). \]
to matrix-valued $\Psi$DOs\, $\operatorname{Op}_h^{\text{w}}(\widehat{\Phi(a_{\tinyvarhexagon})})$ on $L^2(\mathbb R,\mathbb C^2 \otimes \mathbb C^q)$ with symbols that are the Fourier transform of 
\[ \Phi(a_{\tinyvarhexagon}) = \left(e^{-i\gamma_1\gamma_2 h/2} a_{\tinyvarhexagon}(\gamma) \otimes \left[ \left(J_{p,q}\right)^{\gamma_1} \left(K_q^*\right)^{\gamma_2} \right] \right)_{\gamma \in \mathbb Z^2}\]
with $J_{p,q}$ and $K_q$ as in \eqref{def:Jpq} and \eqref{def:Kq}.
Note that $\gamma_1\gamma_2=0$ for any $a_{\tinyvarhexagon}(\gamma)\neq 0$, hence
\[ \Phi(a_{\tinyvarhexagon}) = \left(a_{\tinyvarhexagon}(\gamma) \otimes \left[ \left(J_{p,q}\right)^{\gamma_1} \left(K_q^*\right)^{\gamma_2} \right] \right)_{\gamma \in \mathbb Z^2}\]
In particular, the $C^*$-homomorphism preserves regularized traces, up to constants,
\begin{equation}
\label{eq:traceseq}
\widetilde{\operatorname{tr}}\left(\operatorname{Op}_{\phi}^{\text{w}}\left(\widehat{a_{\tinyvarhexagon}}\right) \right) = \int_{\mathbb{T}_{*}^2} \operatorname{tr}_{\mathbb{C}^2}\left(\widehat{a_{\tinyvarhexagon}}(x,\xi) \right) \frac{dx \ d\xi}{\vert \mathbb T_{*}^2 \vert}=a_{\tinyvarhexagon}(0)= q^{-1} \  \widetilde{\operatorname{tr}}\left(\operatorname{Op}_h^{\text{w}}\left(\widehat{\Phi(a_{\tinyvarhexagon})}\right)\right)
\end{equation}
and, as follows by combining \cite[Theo. $2.1$]{KL14} with \cite[1.2]{HS20}, also spectra 
\begin{equation}
\label{eq:spectra}
\Sigma(H^{\phi})=\Sigma(\operatorname{Op}_{\phi}^{\text{w}}(\widehat{a_{\tinyvarhexagon}})) = \Sigma\left(\operatorname{Op}_h^{\text{w}}\left(\widehat{\Phi(a_{\tinyvarhexagon})}\right)\right).
\end{equation}


Recall that $M_{\tinyvarhexagon}=\widehat{\Phi(a_{\tinyvarhexagon})}$, see \eqref{def:MG}.
We conclude by \eqref{1},\eqref{eq:traces},\eqref{eq:pseudo}, and \eqref{eq:traceseq} that for $M_{\tinyvarhexagon}^{\rm{w}}(x,hp_x)=\operatorname{Op}_h^{\text{w}}M_{\tinyvarhexagon}$,
\[ \widetilde{\operatorname{tr}}_{\Lambda_{\tinyvarhexagon}}\left((H^{\phi}-z)^{-1}\right) =  \frac{\widetilde{\operatorname{tr}}\left((M_{\tinyvarhexagon}^{\rm{w}}(x,hp_x)-z)^{-1}\right)}{q \vert \vec{b}_1 \wedge \vec{b}_2 \vert }.\]

We are concerned with the analysis of this operator close to the Dirac energy $E=0.$ To analyze the spectrum of $M_{\tinyvarhexagon}^{\rm{w}}(x,hp_x)$ close to energies $E=0$, we want to focus on the two bands touching at $E=0$, first.

The obstruction to do so, is that for rational flux $2\pi \frac{p}{q}$ the two bands touching at $E=0$ may not be isolated from the rest of the spectrum, cf. Fig. \ref{fig:touching}. At first glance, this creates an obstruction to \emph{block-diagonalize} the operator $\operatorname{Op}_h^{\text{w}}M_{\tinyvarhexagon}$ at zero energy to leading order. A way to overcome this issue is explained in the following remark: 

\begin{rem}[Isolating bands touching at Dirac energies]
\label{rem:gap}
We recall that $M_{\tinyvarhexagon}$ vanishes only at points $z_0:=(x_0,\xi_0)$ as defined in \eqref{eq:even_max} or \eqref{eq:odd_max}, respectively. To analyze the operator $\operatorname{Op}_h^{\text{w}}M_{\tinyvarhexagon}$ in a neighbourhood of zero energy, it suffices therefore to consider an auxiliary operator with symbol
\begin{equation}
\widetilde M_{\tinyvarhexagon}(z) :=\chi(z) M_{\tinyvarhexagon}(z) + (1-\chi(z))M_{\tinyvarhexagon}\left(2\varepsilon \frac{(z-z_0)}{\Vert z-z_0 \Vert}\right)
\end{equation}
where $\chi \in C^{\infty}(\RR^2)$ and $\chi(z)=1$ in a neighbourhood of $z_0$ and $0$ outside. The parameter $\varepsilon$ is chosen small enough such that the two eigenvalues of $M_{\tinyvarhexagon}\left(2\varepsilon \frac{(z-z_0)}{\Vert z-z_0 \Vert}\right)$ that belong to the two bands which touch at the Dirac energies are distinct from all remaining eigenvalues of $M_{\tinyvarhexagon}\left(2\varepsilon \frac{(z-z_0)}{\Vert z-z_0 \Vert}\right).$ Such a parameter $\varepsilon>0$ exists since the remaining bands of $M_{\tinyvarhexagon}$ are possible touching the two bands that make up the Dirac cones, but they are not intersecting, cf. Fig. \ref{fig:touching}.

This way, $\operatorname{Op}_h^{\text{w}}(\widetilde M_{\tinyvarhexagon})$ and $\operatorname{Op}_h^{\text{w}}( M_{\tinyvarhexagon})$ coincide microlocally, i.e for any $\chi \in C_c^{\infty}(\operatorname{nbhd}(z_0))$ we have 
\[ \left\Vert \operatorname{Op}_h^{\text{w}}\chi \left(\operatorname{Op}_h^{\text{w}}(\widetilde M_{\tinyvarhexagon})-\operatorname{Op}_h^{\text{w}}( M_{\tinyvarhexagon}) \right)\operatorname{Op}_h^{\text{w}}\chi \right\Vert=\mathcal O(h^{\infty}), \]
see e.g. \cite[Theo. $4.25$]{ev-zw}. For our subsequent analysis, we may therefore just assume without loss of generality that the two touching bands of $M_{\tinyvarhexagon}$ at zero energy are gapped from the rest of its spectrum.
\label{rem:isolate}
\end{rem}

To analyze $\operatorname{Op}_h^{\text{w}}M_{\tinyvarhexagon}$, we recall a few properties about the matrix-valued symbol $M_{\tinyvarhexagon}$ first. Clearly, $\cup_{(x,\xi)  \in \mathbb T_*^2} M_{\tinyvarhexagon}(x,\xi) $ has band spectrum $B_{\ell}=[\gamma_{\ell},\delta_{\ell}]$, $1\leq \ell\leq 2q$, and we denote associated energy eigenvalues by $\mu_{\ell}(x,\xi)$. The $q$-th and $q+1$-st band always touch at the Dirac point, i.e.\@ $\delta_q = \gamma_{q+1}=0$ by Theorem \ref{thm:Dirac}. The phase space coordinates at which the $q$-th and $q+1$-st band touch are denoted by $z_{j}:=(x_{j},\xi_{j}) \in \mathbb T_*^2$, where $j \in \left\{1,..,2q^2 \right\}$, i.e. $\mu_q(z_{j})=\mu_{q+1}(z_{j})=0.$ There are by \eqref{eq:even_max} and \eqref{eq:odd_max} precisely $2q^2$ such points in a single fundamental domain $\mathbb T_*^2.$ For the analysis close to individual conical points, we fix a sufficiently small $\varepsilon>0$ and consider energies $E \in I_{\varepsilon} = (-\varepsilon, \varepsilon)$. We define for such energies the phase space level set $ \Sigma_{j}(E):=\mu_{\ell} \vert_{\text{nbhd}(z_{j})}^{-1}(E) \subset \mathbb T_*^2 $ for $\ell\in\{q, q+1\}$ here, close to a \emph{single} potential well centred at $z_{j}$ and the phase space area $V_{j,\varepsilon} := \bigcup_{E \in I_{\varepsilon}} \Sigma_{j}(E)$ of all energies in the interval $I_{\varepsilon}$. 

Remark \ref{rem:gap} allows us to make two simplifying coordinate changes near the conical points which we discuss now:

 There exists a unitary operator $U$ such that\footnote{We assume here by a simple change of coordinates that the Dirac point is located at $(x,\xi)=0$} \cite[Prop.\@$3.1.1$ \& Cor.\@$3.1.2$]{HS20}
\begin{equation}
\begin{split}
\label{eq:Diracform}
&U^* \operatorname{Op}_h^{\text{w}}M_{\tinyvarhexagon} U =\operatorname{diag}(\underbrace{\operatorname{Op}_h^{\text{w}}M_{\text{D},\tinyvarhexagon}}_{ \in \CC^{2 \times 2}}, \underbrace{\operatorname{Op}_h^{\text{w}}M_{\text{R},\tinyvarhexagon}}_{ \in \CC^{(2q-2) \times  (2q-2)}}) \\
&\text{ where }\operatorname{Op}_h^{\text{w}}M_{\text{D},\tinyvarhexagon}=  \left(\begin{matrix} 0 & \operatorname{Op}_h^{\text{w}}b \\  \operatorname{Op}_h^{\text{w}}b^*& 0 \end{matrix} \right)+\mathcal O(h). 
\end{split}
\end{equation}
The subscript $D$ stands for \emph{Dirac} and $R$ for \emph{rest}, and the symbol $b$ satisfies $b(x,\xi) =\tfrac{v_F}{2}( \xi+ ix) + \mathcal O(\Vert (x,\xi) \Vert^2)$ where the Fermi velocity $v_F$ satisfies by \eqref{eq:F=E} and \eqref{eq:well}
\begin{equation}
\label{eq:Fermivel}
v_F= \frac{q}{3^{3/4}} \frac{1}{3^{q-1}\prod_{j=q+2}^{2q} (F_j(\tilde{ k}))}. 
\end{equation}

For the pseudodifferential operator $\operatorname{Op}_h^{\text{w}}M_{\tinyvarhexagon} = \left( \begin{matrix} \textbf{0} &\operatorname{Op}_h^{\text{w}}\mathcal A \\ \operatorname{Op}_h^{\text{w}}\mathcal A^* & \textbf{0} \end{matrix} \right)$, with $\mathcal A $ as in \eqref{def:MG}, we obtain by squaring the operator
\begin{equation}
\begin{split}
\label{eq:square}
\left(\operatorname{Op}_h^{\text{w}}M_{\tinyvarhexagon}\right)^2 = \left(\begin{matrix}\operatorname{Op}_h^{\text{w}}\mathcal A\operatorname{Op}_h^{\text{w}}\mathcal A^* & \textbf{0} \\
\textbf{0} &\operatorname{Op}_h^{\text{w}}\mathcal A^*\operatorname{Op}_h^{\text{w}}\mathcal A \end{matrix} \right). 
\end{split}
\end{equation}
By supersymmetry it follows that away from $0$ both operators $\operatorname{Op}_h^{\text{w}}\mathcal A\operatorname{Op}_h^{\text{w}}\mathcal A^*$ and $\operatorname{Op}_h^{\text{w}}\mathcal A^*\operatorname{Op}_h^{\text{w}}\mathcal A$ have the same spectrum.
The principal symbols are
\begin{equation}
\begin{split}
\sigma_0 \left(\operatorname{Op}_h^{\text{w}}\mathcal A\operatorname{Op}_h^{\text{w}}\mathcal A^*\right) &= M_T(x,\xi)+3I_q \text{ and } \\
\sigma_0 \left(\operatorname{Op}_h^{\text{w}}\mathcal A^*\operatorname{Op}_h^{\text{w}}\mathcal A\right) &= \widehat{M_T}(x,\xi)+3I_q
\end{split}
\end{equation}
with the notation as in \eqref{eq:MG^2}.
Let $Z(x,\xi)$ now be either $M_T(x,\xi)+3I_q $ or $\widehat{M_T}(x,\xi)+3I_q.$ The lowest eigenvalue of $Z(x,\xi)$ is given by a smooth scalar function $(x,\xi) \mapsto \nu(x,\xi) = \vert \mu_{q+1}(x,\xi) \vert^2$, see Remark \ref{rem:isolate}. Thus, there are analytic unitary matrices $V$ separating the lowest eigenvalue from the rest of the matrix
\begin{equation}
\label{eq:diagonalization}
(V^* ZV)(x,\xi)= \operatorname{diag}(\nu(x,\xi),B(x,\xi)), 
\end{equation}
where by Remark \ref{rem:gap} we may assume that $\inf_{(x,\xi) \in T^*\RR }\vert \Sigma(B(x,\xi))-\nu(x,\xi) \vert>0$ and $B(x,\xi) \in \mathbb C^{(q-1) \times (q-1)}$. 

Thus, as for the Dirac-type operator above, \cite[Prop.\@ $3.1.1$ \& Corr.\@ $3.1.2$]{HS20} imply since the lowest band of $Z$, described by $\nu$, is gapped from the rest of the spectrum, there is a unitary operator $U$ and symbols $\widetilde{\nu},\widetilde{B}$ with asymptotic expansions in $\mathcal S$, such that 
\begin{equation}
\label{eq:nu}
U^* \left(\operatorname{Op}_h^{\text{w}}\mathcal A\operatorname{Op}_h^{\text{w}}\mathcal A^*\right) U = \left(\begin{matrix}\operatorname{Op}_h^{\text{w}}\widetilde{\nu} & \textbf{0} \\ \textbf{0} & \operatorname{Op}_h^{\text{w}}\widetilde{B} \end{matrix} \right) + \mathcal O_{\mathcal L(L^2(\mathbb R))}(h^{\infty}),
\end{equation}
where $\sigma_0(\widetilde{\nu}) = \nu$ and $\sigma_0(\widetilde{B})= B.$

The main result of this section, a semiclassical trace formula close to rational flux, is then stated in the following Theorem:

\begin{theo}[Semiclassical DOS and QHE close to a rational]
\label{theol}
For small $h>0$ sufficiently small, with respect to $p,q$, and magnetic flux $\phi=2\pi\frac{p}{q}+h,$ the \emph{DOS} of $H^{\phi}_{\tinyvarhexagon}$ admits the following expansion: \\
Let $I$ be an interval $I \subset (-\delta,\delta)$ for some $\delta>0$ sufficiently small\footnote{This interval encloses energies around the Dirac points in Figure \ref{fig:hex}.} and $ f \in C^{\alpha}_{\rm{c}} ( I), $ then
\begin{equation}
\begin{split}
\label{eq:TF}
&\widetilde \tr_{\Lambda} (f (  H^{\phi}_{\tinyvarhexagon})) = \tfrac{ qh}{\pi \vert \vec{b}_1 \wedge \vec{b}_2 \vert}  \sum_{ n \in \ZZ } 
f( z_{n} (h,p,q) )+ \mathcal O (\| f \|_{ C^{\alpha} } h^\infty),
\end{split}
\end{equation}
with Landau levels $z_{n} ( h,p,q  ) = \kappa ( n h , h,p,q )$
satisfying $ \kappa ( - \zeta, h,p,q ) = - \kappa ( \zeta , h,p,q )$, defined by a Bohr-Sommerfeld condition
\begin{equation}  
\begin{split} 
\label{eq:BSC}
F ( \kappa ( \zeta,h, p,q )^2, h, p,q ) &= |\zeta| + \mathcal O (h^{\infty}), \ F ( s,h,p,q  ) =  \sum_{j=0}^{\infty} F_j(s,p,q) h^j, \ F_j(0,p,q)=0, \\
\text{ where }&F_0(s,p,q):= \int_{\nu(x,\xi) \in [0,s]} \frac{dx \ d \xi}{4\pi q^2} \text{ and } \\
&F_1(s,p,q):= \frac{1}{2} - \frac{d}{d\zeta} \Big\vert_{\zeta=s} \int_{\nu(x,\xi) \in [0,\zeta]}  \sigma_1(\widetilde{\nu})(x,\xi) \ \frac{dx \ d \xi}{4\pi q^2}. 
 \end{split} 
\end{equation}
With the Fermi velocity $v_F$ defined in \eqref{eq:Fermivel}, $z_n$ satisfies
\begin{equation}
\begin{split}
\label{eq:LLrat}
z_0 &= \mathcal O(h^{\infty}) \text{ and } \\
z_n &= \sgn(n)v_F  \sqrt{\vert n \vert h} + \mathcal O(h) \ , \ n \neq 0.
\end{split}
\end{equation}
In addition, the spectrum of the magnetic Schr\"odinger operator around zero $ \Sigma(H^{\phi}_{\tinyvarhexagon}) \cap I$  is contained in disjoint closed \emph{Landau bands} $B_{\tinyvarhexagon,n}(h,p,q) \ni z_n(h,p,q)$ with spectral gaps
\begin{equation}
\label{eq:moregaps}
d\left(B_{\tinyvarhexagon,n}(h,p,q) ,B_{\tinyvarhexagon,n+1}(h,p,q) \right)\ge C_{n,p,q} h 
\end{equation} 
for some constant $C_{n,p,q} >0.$ The Hall conductivity satisfies for Fermi energies $\mu$
\begin{equation}
\begin{split}
\label{eq:Hallcond}
c_H(H^{\phi}_{\tinyvarhexagon},\mu) &= \begin{cases} & \frac{(2n+1)q}{2\pi}, \ \mu\text{ betw. }\mathcal B_{\tinyvarhexagon,\lambda,n}\text{ and } \mathcal B_{\tinyvarhexagon,\lambda,n+1}\text{ with }0 \le n \le N_{\scriptscriptstyle\tinyvarhexagon}(h,\lambda_0) \\
&\frac{(2n-1)q}{2\pi}, \ \mu\text{ betw. }\mathcal B_{\tinyvarhexagon,\lambda,n-1}\text{ and } \mathcal B_{\tinyvarhexagon,\lambda,n}\text{ with }0 \ge n \ge -N_{\scriptscriptstyle\tinyvarhexagon}(h,\lambda_0).
\end{cases}
\end{split}
\end{equation}
An illustration of the Hall conductivities is given in Figure \ref{fig:HC2}.

\end{theo}
\begin{rem}[Dynamical delocalization]\label{rem:DLrational}
In particular, using the results from Subsection \ref{sec:DLoc}, we conclude from \eqref{eq:moregaps} that for sufficiently weak disorder, such that the (disorder-broadened) Landau bands remain non-overlapping, there exists at least one mobility edge inside each Landau band at which delocalization occurs.
\end{rem}

\begin{figure}
\label{fig:bf}
\captionsetup[subfigure]{oneside,margin={0.5cm,0cm}}
    \begin{subfigure}[t]{0.45\textwidth}
         \begin{flushleft}
        \includegraphics[height=1.6in,width=2.5in]{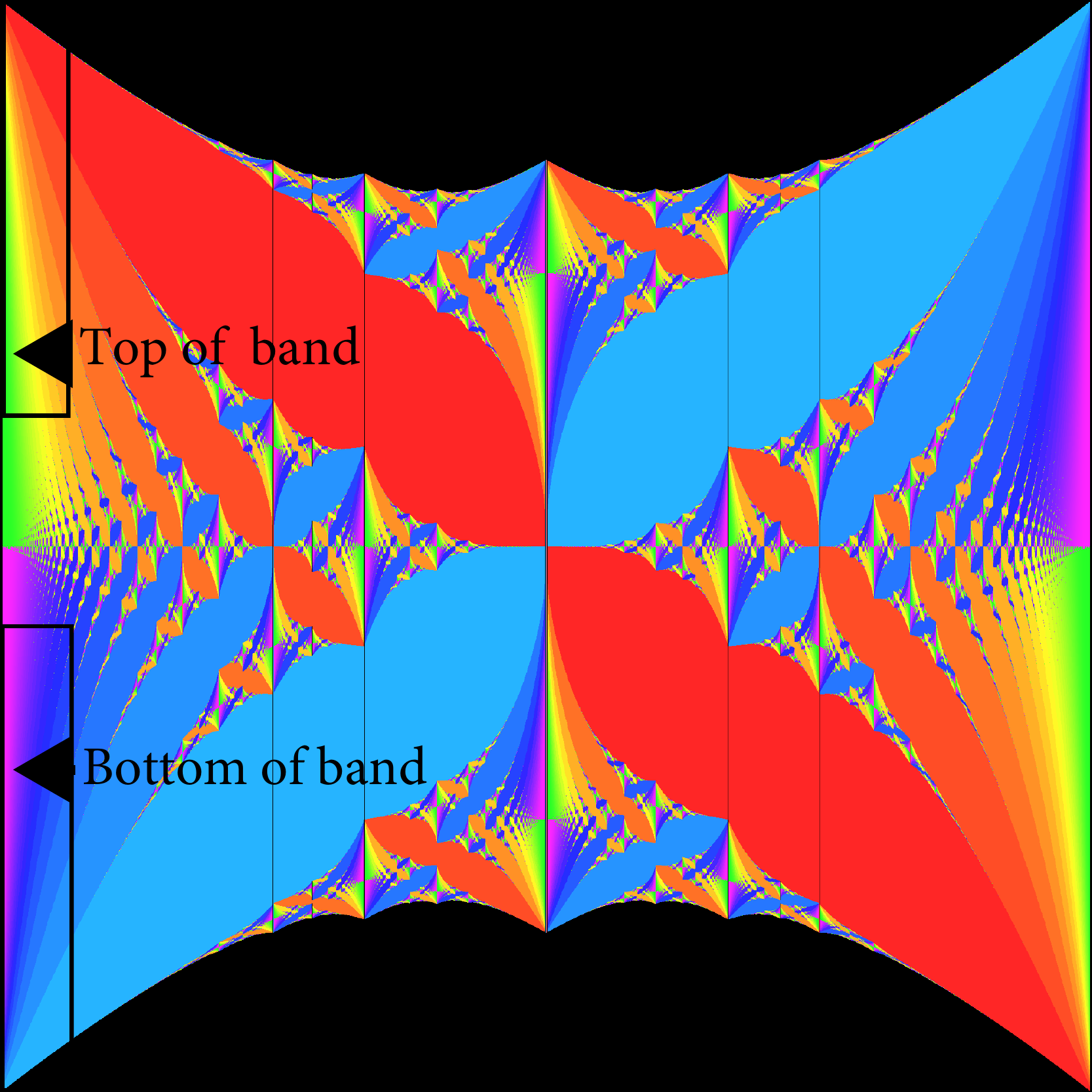}
        \caption{The square lattice $\Lambda_{\blacksquare}.$ The Hall conductivity on the lower and upper spectral edge that is computed in this paper, in the regime of small magnetic flux, is located on the strip below/above the respective arrow. The energy on the vertical axis covers the full range of the operator.}
        \end{flushleft}
        \label{fig:1}
    \end{subfigure}%
    ~ 
    \begin{subfigure}[t]{0.45\textwidth}
    \begin{flushright}
        \includegraphics[height=1.6in,width=2.5in]{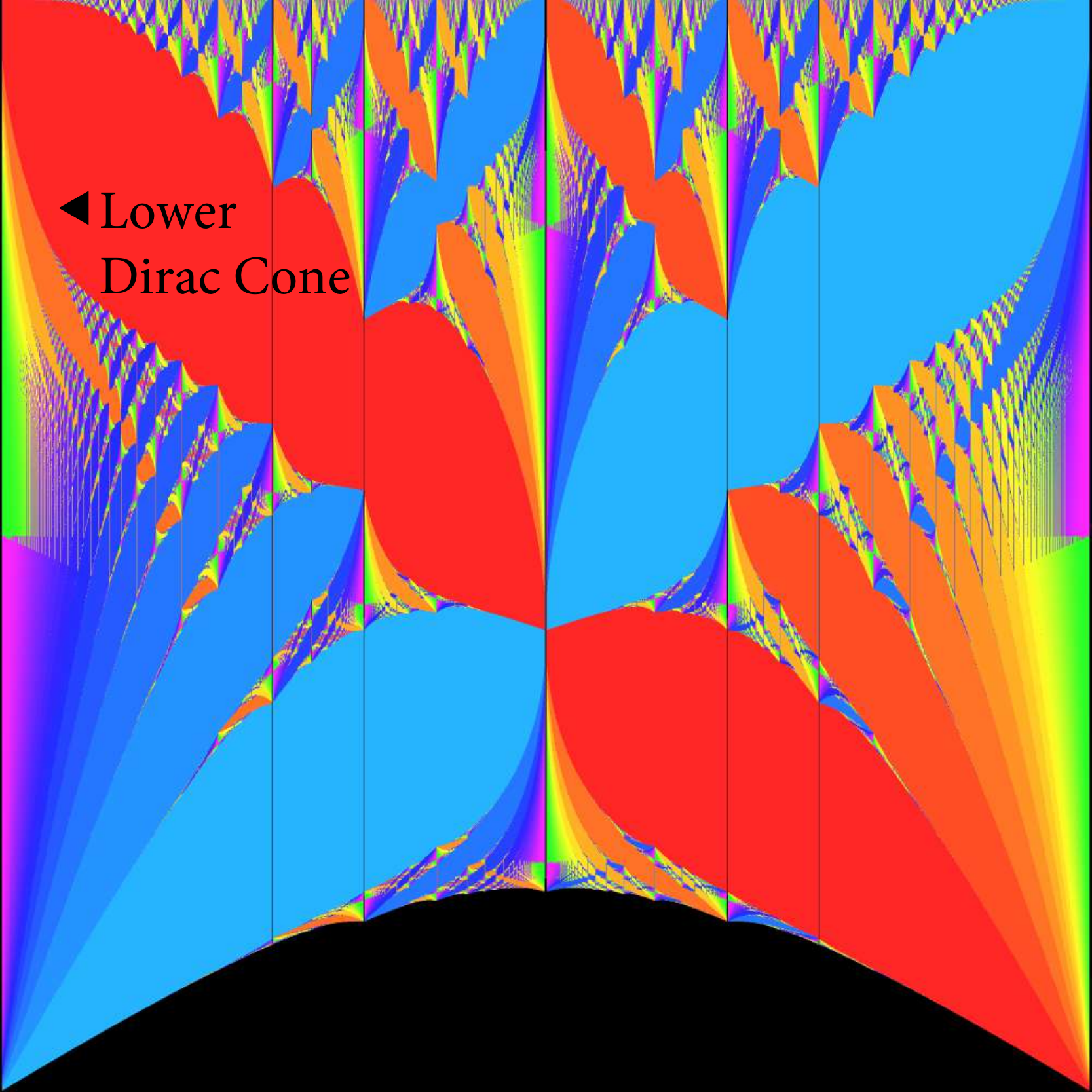}
        \caption{The hexagonal lattice $\Lambda_{\hexagon}$ (lower band, only). The Hall conductivity on the lower Dirac cone that is computed in this paper is located on the strip to the left and above the arrow. The energy scale on the vertical axis covers the interval $[-1,0]$.}
        \end{flushright}
          \end{subfigure}
    \caption{Hall conductivity (coloured) as a function of magnetic flux $h \in [0,2\pi]$ (horizontal axis) and energy (vertical axis). Dark regions do not carry spectrum. Different colours represent different conductivities. \label{fig:HC2}}
\end{figure}

\section{Proofs}
\label{sec:proofs}
We now state the proof of Theorems \ref{theo1} and \ref{theol} with several references to details that are already discussed in \cite{BZ,HS0}.
\begin{proof}[Proof of Thm.\@ \ref{theo1} \& Thm.\@ \ref{theol}] 
\label{fullproof}

\smallsection{Step $1$: Quasimodes and Landau levels}
Quasimodes and Landau levels are constructed as eigenfunctions and eigenvalues to \emph{localized operators}, i.e.\@ operators that coincide microlocally, up to a constant shift of the spectrum, with $\Psi$DOs \eqref{eq:effha} in a neighbourhood of a single potential well. 
For the square lattice, such a localized operator with discrete spectrum at the bottom of the potential well, see Fig. \ref{fig:square}, is defined by the Weyl symbol 
 \begin{equation}
\begin{gathered} 
\label{eq:Q0}
Q^{0}_{{\scriptscriptstyle{ \blacksquare}}}( x, \xi ) := Q_{{\scriptscriptstyle{ \blacksquare}}} ( x, \xi) + 2 -\chi_{{\scriptscriptstyle{ \blacksquare}}}(x,\xi), \text{ where }\\ 
\chi_{{\scriptscriptstyle{ \blacksquare}}} \in C^\infty_{\rm{c}}  ( \RR^2 ; [ 0, 1 ] ) ,  \ \ 
\chi_{{\scriptscriptstyle{ \blacksquare}}} (x,\xi) = \left\{ \begin{array}{ll} 1, &\left\lVert (x,\xi)-(\pi,\pi) \right\rVert_\infty < \frac1{10}  , \\
0, & \left\lVert (x,\xi)-(\pi,\pi)\right\rVert_{\infty} >  \frac1{5}. \end{array} \right.  
\end{gathered}
\end{equation}
Thus, $\operatorname{Op}_h^{\text{w}}Q^{0}_{{\scriptscriptstyle{ \blacksquare}}}-z$ is elliptic \cite[Sec.\@ 4.7]{ev-zw} for $z$ in a small neighbourhood of zero and $(x,\xi) \notin \text{nbhd}(\pi,\pi)$ where the neighbourhood depends on $z.$

On the hexagonal lattice such a localized operator with discrete spectrum close to zero energy, the energy level of the Dirac points, see Fig. \ref{fig:hex}, is defined by the symbol
\begin{equation}
\begin{split} 
\label{eq:M0}
&M_{\tinyvarhexagon}^0( x, \xi ) := M_{\tinyvarhexagon} ( x, \xi) + 
\left(\begin{matrix} (\chi_{\tinyvarhexagon} ( x, \xi)-1) I_q& \textbf{0}  \\
 \textbf{0}  & (1 - \chi_{\tinyvarhexagon}   ( x, \xi))I_q \end{matrix}\right),\\
&\chi_{ \tinyvarhexagon} \in C^\infty_{\rm{c}}  ( \RR^2 ; [ 0, 1 ] ), \chi_{ \tinyvarhexagon}(z)=\chi_{ \tinyvarhexagon} (-z),
\end{split}
\end{equation}
where $\chi_{\tinyvarhexagon} (x,\xi ) =1$ on all $\cup_{j \in \left\{1,..,2q^2 \right\}} V_{j,\delta}$ for some $\delta>0$ sufficiently small and vanishes outside of $\mathbb T_*^2.$

Next, we argue that the spectrum of both $\operatorname{Op}_h^{\text{w}}Q^{0}_{\scriptscriptstyle{ \blacksquare}}$ and $\operatorname{Op}_h^{\text{w}}M_{\tinyvarhexagon}^{0}$ is indeed contained in discrete intervals around zero. To do so, we define another pair of symbols
\begin{equation}
\begin{split}
Q^{1}_{{\scriptscriptstyle{ \blacksquare}}}( x, \xi )&:= Q_{{\scriptscriptstyle{ \blacksquare}}} ( x, \xi) + 2 \text{ and } 
M_{\tinyvarhexagon}^1 ( x, \xi ):= M_{\tinyvarhexagon} ( x, \xi) + \operatorname{diag}(-I_q,I_q).
\end{split}
\end{equation}
The two associated operators with upper index $1$ are invertible close to zero and we have
\begin{equation}
\begin{split}
\operatorname{Op}_h^{\text{w}}Q^{0}_{{\scriptscriptstyle{ \blacksquare}}}-z &=\left(\operatorname{Op}_h^{\text{w}}Q^{1}_{{\scriptscriptstyle{ \blacksquare}}} -z \right) \left(\operatorname{id}+K_{{\scriptscriptstyle{ \blacksquare}}}(z)\right)\text{ and } \\
 \operatorname{Op}_h^{\text{w}}M_{\tinyvarhexagon}^{0}-z &=\left(\operatorname{Op}_h^{\text{w}}M_{\tinyvarhexagon}^{1} -z \right) \left(\operatorname{id}+K_{\tinyvarhexagon}(z)\right)
\end{split}
\end{equation} 
for some compact operators
\begin{equation}
\begin{split}
K_{{\scriptscriptstyle{ \blacksquare}}}(z) &=\left( \operatorname{Op}_h^{\text{w}}Q^{1}_{{\scriptscriptstyle{ \blacksquare}}}-z \right)^{-1} \chi_0^{\rm{w}}  \text{ for }z \notin \Sigma( \operatorname{Op}_h^{\text{w}}Q^{1}_{{\scriptscriptstyle{ \blacksquare}}})\text{ and }\\
K_{\tinyvarhexagon}(z) &=\left( \operatorname{Op}_h^{\text{w}}M_{\tinyvarhexagon}^{1}-z \right)^{-1} \operatorname{diag}(\chi_0^{\rm{w}},-\chi_0^{\rm{w}}) \text{ for } z \notin \Sigma(\operatorname{Op}_h^{\text{w}}M_{\tinyvarhexagon}^{1}).
\end{split}
\end{equation} 
By analytic Fredholm theory \cite[Theorem D.4]{ev-zw} this implies the discreteness of the spectrum of $Q^{0}_{{\scriptscriptstyle{ \blacksquare}}}$ and $M_{\tinyvarhexagon}^{0}$ close to zero.
Thus, there exists a family of eigenvalues and orthonormal eigenfunctions such that 
\begin{equation}
\begin{split}
\label{eq:zeros}
\left(\operatorname{Op}_h^{\text{w}} Q^{0}_{{\scriptscriptstyle{ \blacksquare}}}-\kappa_{{\scriptscriptstyle{ \blacksquare}}}(nh,h) \right) u_{n, {\scriptscriptstyle{ \blacksquare}}} &= 0\text{ and } 
\left(\operatorname{Op}_h^{\text{w}} M^{0}_{\tinyvarhexagon}-\kappa_{\tinyvarhexagon}(nh,h)\right) u_{n, \tinyvarhexagon} =0.
\end{split}
\end{equation}
Localized operators with upper index $0$ have the property that their spectra for energies close to zero stay close to the spectra of operators $\operatorname{Op}_h^{\text{w}}Q_{{\scriptscriptstyle{ \blacksquare}}}$ and $\operatorname{Op}_h^{\text{w}}M_{\tinyvarhexagon} $, respectively. In fact, an immediate adaptation of the proof of \cite[Lemma $5.2$]{BZ} shows that after possibly shrinking the energy window around zero to some $\varepsilon_1$ with $ 0 < \varepsilon_{1} < \varepsilon $ and $ z \in [0, \varepsilon_{1} ] - i [ -1,1] $ such that $d \left( z ,\Sigma\left( \operatorname{Op}_h^{\text{w}}Q^{0}_{{\scriptscriptstyle{ \blacksquare}}}\right)\right) > h^{n},$ for some arbitrary but fixed $n \in \mathbb N$, there is $h_0$ such that for $h \in (0,h_0) $, 
\begin{equation}
\label{eq:comparison}
 ( \operatorname{Op}_h^{\text{w}}Q_{{\scriptscriptstyle{ \blacksquare}}}- z )^{-1} = \mathcal O_{ L^2 \to L^2 } (  d ( z , \Sigma(  \operatorname{Op}_h^{\text{w}}Q^{0}_{\scriptscriptstyle{ \blacksquare}}))^{-1})
 \end{equation}
and the analogous result is true for $M^{\rm{w}}_{\tinyvarhexagon}$ as well.

Since $\operatorname{Op}_h^{\text{w}}M_{\tinyvarhexagon}$ and $\operatorname{Op}_h^{\text{w}}M_{\tinyvarhexagon}^0$ in $\cup_{j \in \left\{1,..,2q^2 \right\}} V_{j,\delta}$ coincide microlocally we infer from \eqref{eq:zeros} that
\begin{equation}
\begin{split}
\left(\operatorname{Op}_h^{\text{w}} M_{\tinyvarhexagon}-\kappa_{\tinyvarhexagon}(nh,h)\right) u_{n, \tinyvarhexagon} =\mathcal O(h^{\infty}).
\end{split}
\end{equation}
Thus, one has to find all such microlocal solutions with $\WF_h(u_{n, \tinyvarhexagon} ) \subset \cup_{j \in \left\{1,..,2q^2 \right\}} V_{j,\delta}.$
Microlocal solutions $\left(\operatorname{Op}_h^{\text{w}} M_{\tinyvarhexagon}-z\right)u = \mathcal O(h^{\infty})$ for $z \ge c \sqrt{h}$ are in one-to-one correspondence with microlocal solutions $v \in\WF_h(u_{n, \tinyvarhexagon} ) \subset \cup_{j \in \left\{1,..,2q^2 \right\}}V_{j,\delta} $ such that by \eqref{eq:square}
\begin{equation}
\begin{split}
\label{eq:quassquare}
&\left(\operatorname{Op}_h^{\text{w}}\mathcal A\operatorname{Op}_h^{\text{w}}\mathcal A^*-\lambda  \right)v = \mathcal O(h^{\infty}) \\
&z= \pm \lambda, \ u:=(u_1,u_2) := \left(v,z^{-1} \operatorname{Op}_h^{\text{w}}\mathcal A^*v\right) \in \mathbb C^{2q}.
\end{split}
\end{equation}

Since $0$ is in the spectrum of $H^h_{\tinyvarhexagon}$ for all $h \in [0,2\pi]$ \cite[Lemma $5.1$]{BHJ17}, we have that $0 \in \Sigma(\operatorname{Op}_h^{\text{w}}M_{\tinyvarhexagon})$ for all $h$ by \eqref{eq:spectra}. Invoking now \eqref{eq:comparison} for the hexagonal lattice, implies that there exists an eigenvalue $\mathcal O(h^{\infty})$ to the localized operator $\operatorname{Op}_h^{\text{w}}M_{\tinyvarhexagon}^{0}$. 

We can now apply the following Bohr-Sommerfeld condition \cite{hr,HS20,CdV}:

Let $H:T^*\RR \rightarrow \RR$ be a classical symbol with expansion $H \sim \sum_{i=0}^{\infty} H_i h^i$ 
Moreover, we assume the principal symbol $H_0$ to satisfy
\begin{enumerate}
\item \label{eq:conditions1}$ H_0(z) = 0\text{ and } (D^2H_0)(z)>0,$
\item The set $\{\nu \in \mathbb R^2: H_0(\nu) < \delta \}$  is compact and connected for some $\delta>0$ sufficiently small.
\item \label{eq:conditions3}$H_0$ is strictly positive and does not possess any other critical points, apart from $z$ in a sufficiently small nbhd of $z.$
\end{enumerate}
Then, the spectrum of $\operatorname{Op}^{\text{w}}_h(H)$ close to zero is given by the Bohr-Sommerfeld condition
\[ F(E,h) = \sum_{j=0}^{\infty} F_j (E) h^j = nh\]
where the leading-order term is the \emph{Bohr-Sommerfeld term}
\[ F_0(E) = \frac{1}{2\pi} \int_{\left\{ H_0 \le E \right\}} dx \ d\xi \]
and the subprincipal term $F_1$ includes the \emph{Maslov correction} and the contribution from the subprincipal symbol $H_1$
\begin{equation}
\label{eq:F1}
F_1(E) = \frac{1}{2}- \frac{1}{2\pi} \frac{d}{ds} \Big \vert_{s=E} \int_{\left\{ H_0 \le s \right\}} H_1(x,\xi) \ dx \ d\xi.
\end{equation}
Expressions for higher-order terms $F_j$ with $j \ge 2$ can be found in \cite{CdV}.

This immediately yields the Bohr-Sommerfeld condition for the square lattice \eqref{eq:g2F}, by applying it to the microlocally equivalent symbol $Q^{0}_{{\scriptscriptstyle{ \blacksquare}}}$ in \eqref{eq:Q0}, since the subprincipal is zero and therefore $F_1(E)=\frac{1}{2}.$

In case of the hexagonal lattice, we use that by \eqref{eq:quassquare} and \eqref{eq:nu} it suffices to study the quasimodes to the symbol $\widetilde{\nu}.$ Clearly, $\widetilde{\nu}$ satisfies both assumptions (1) and (3) of the Bohr-Sommerfeld condition. 

By using cut-off functions $\chi_{j,\tinyvarhexagon}$ that are localized to neighbourhoods $V_{j,\delta}$ of a single well, the localized symbol
\[\widetilde{\nu}_j(x,\xi) := \widetilde{\nu}(x,\xi) + (1-\chi_{j,\tinyvarhexagon})(x,\xi) \]
satisfies then all three conditions of the Bohr-Sommerfeld rule which yields \eqref{eq:BSC}.

When $q=1$ and $\mathcal A$ is scalar, a direct computation of \eqref{eq:F1} shows that $F_1 = 0$ \cite{BZ2}. This yields the Bohr-Sommerfeld condition stated in Theorem \ref{theo1}.

Finally for the analysis close to rationals, the asymptotics of Landau levels \eqref{eq:LLrat} and the presence of gaps \eqref{eq:moregaps} follow immediately from both \eqref{eq:Diracform} and \eqref{eq:Fermivel}, and the explicit spectral analysis of the $2D$-magnetic Dirac operator, cf. \cite{HS20}[Prop $3.6.1$ and $(3.6.22)$].

\smallsection{Step $2$: The Grushin problem}
To prove the trace formulae, we fix one Landau level and take $ z_1 $ and $ \epsilon_0 $ with
\begin{equation}  
\label{eq:specrah}     \{ \kappa ( n h, h ) \}_{ n }  \cap [ z_1 - 
2 \epsilon_0 h , z_1 +  2 \epsilon_0 h ]  = \{ \kappa ( n_1 h, h ) \} , \ \ 
n_1 = n_1 ( z_1 , h ).
\end{equation}  
Since symbols $Q_{\scriptscriptstyle{ \blacksquare}}$ and $M_{{\tinyvarhexagon}}$ are $2\pi$-periodic, they possess infinitely many potential wells. Therefore, we introduce a translation operator $r_\gamma u ( x ) := e^{ \frac i h \gamma_2 x } u ( x - \gamma_1 )$ to define translations of the quasimodes $w_{\gamma}:=r_{\gamma} u$ for $\gamma \in \mathbb Z_{*}^2.$ 
We then define operators $ R_+ : L^2 ( \RR , \CC^{m} ) \to \ell^2 ( \mathbb Z^2_*; \CC^{n} ) $ and
$ R_- : \ell^2 ( \mathbb Z^2_*; \CC^{n} ) \to L^2 ( \RR , \CC^{m} ) $ by
\begin{equation}
\begin{split}
\left( R_+ u_{+} \right) ( \gamma ) & := \int_{\RR} \overline{u_{+}(x)}\ {}^{t} w_{\gamma}(x) \ dx \in \CC^{n} ,  
\ \ 
R_- u_- ( x ) 
:= \sum_{ \gamma \in \mathbb Z^2_* }
  w_\gamma ( x )  u_- ( \gamma ) , 
\end{split}
\end{equation}
where 
\begin{itemize}
\item $n=m=1$ for the square lattice and
\item $n=2q^2$, $m=2q$ on the hexagonal lattice close to the flux $2 \pi p/q$, in which case 
\[  u_- ( \gamma ) = \left(\begin{matrix}   u_-^{1} ( \gamma) \hdots  u_-^{2q^2} (\gamma ) \end{matrix}\right)^t 
\in \CC^{2q^2} \text{ and } w_\gamma ( x ) = \left( w_\gamma^1 \hdots w_\gamma^{2q^2} \right) \in 
\CC^{2q \times 2q^2}.\]
\end{itemize}
This way, the following Grushin problem \cite[Prop.\@ $5.4$]{BZ} is well-posed for $ z \in (z_1 - \epsilon_0 h , z_1 + \epsilon_0 h ) 
+ i ( - 1, 1 ) $, where $\mathcal P(z):=\operatorname{Op}_h^{\text{w}}Q_{{\scriptscriptstyle{ \blacksquare}}}- z$ for the square and $\mathcal P(z):=\operatorname{Op}_h^{\text{w}}M_{\tinyvarhexagon}-z$ for the hexagonal lattice,
\begin{equation}
\label{eq:E}\left(\begin{matrix} \mathcal P(z)  & R_- \\
R_+ & 0 \end{matrix}\right)^{-1}=: \left(\begin{matrix} E ( z, h ) & E_{+} (z, h ) \\
E_- ( z, h ) & E_{-+} ( z , h ) \end{matrix}\right). \end{equation}
Schur's complement formula implies that
\[\mathcal P(z)^{-1} = E(z,h) - E_{+}(z,h)E_{\pm}(z,h)E_{-}(z,h)\]
where $E_{+},E_{\pm}$, and $E_{-}$ can be approximated by
\begin{equation}
\label{eq:approx}
  E_+^0  := R_- , \ \  E_-^0  := R_+  , \ \ 
E_{\pm}^0 = ( z - \kappa ( h n_1, n_1 ) ) \delta_{\gamma,0} . 
\end{equation}
Here, $E_{\pm}(\gamma) = E_{\pm}^0(\gamma)+\mathcal O(h^{\infty} \langle \gamma \rangle^{-\infty})$ for $\vert \Im(z) \vert>h^m$, for some fixed $m$, and 
\begin{equation}
\begin{split}
\label{eq:Epm}
 &E_+ ( z, h ) v_+ ( x ) = \sum_{ \gamma \in \mathbb Z^2_*}    r_\gamma  W_0  ( x ) v_+ ( \gamma ) , \ \ W_0 = w_0  + e_0  , \ \ 
e_0 = \mathcal O ( h^\infty)_{\mathscr S } , \\
&( E_- ( z , h ) v )( \gamma ) =  \langle v , r_\gamma W_- \rangle, \ \ W_- = w_0 + f_0 , \ \ f_0 \in \mathcal O ( h^\infty )_{\mathscr S }
\end{split}
\end{equation}
where the estimates follow as in \cite[Proof of Prop.\@ 5.4]{BZ}.
Moreover, we define the function $G(z,h) :=  \int_{\mathbb T_{*}^2} \sigma(E(z,h)))(x,\xi) \frac{dx \ d\xi}{\vert \mathbb{T}_*^2 \vert}$
which is holomorphic in $z \in (z_1-\varepsilon_0 h , z_1+\varepsilon_0 h )+i(-1,1)$ \cite[(6.1)]{BZ}.

To study 
\[ J(z,h)=\int_{\mathbb T_*^2} \operatorname{tr}_{\CC^m} \sigma \left( E_{+} E_{\pm} E_{-} \right)(x,\xi) \frac{dx \ d\xi}{\vert \mathbb T_*^2 \vert}\]
we define, for fixed $M$, the approximation $J_0$ for 
\[ z \in (z_1 - \epsilon_0 h , z_1 + \epsilon_0 h ) 
+ i ( - 1, 1 ), n = n_1 ( z_1 , h ),\text{ and }| \Im z | > h^{M }\] 
by using approximations \eqref{eq:approx}
\begin{equation}
\label{eq:J0}
J_0(z,h)=\int_{\mathbb T_*^2} (z-\kappa(n_1h,h))^{-1} \operatorname{tr}_{\CC^m} \sigma \left( E_{+}^0 E_{-}^0 \right)(x,\xi) \frac{dx \ d\xi}{\vert \mathbb T_*^2 \vert}.
\end{equation}
Estimates \eqref{eq:Epm} imply then that $J(z,h)=J_0(z,h) + \mathcal O(h^{\infty}).$

To find a more explicit expression for $J_0$ we study the Schwartz kernel $K$ of the operator $E_+^0 E_-^0$ given by
\[  K ( x, y ) = \sum_{ \alpha \in \mathbb Z_*^2 } E_+^0 ( x, \alpha ) E_-^0 ( \alpha, y ) = 
\sum_{\alpha } w_\alpha ( x ) w_\alpha ( y )^* ,
\]
from which the symbol of the pseudodifferential operator, appearing in \eqref{eq:J0}, can be derived from the Schwartz kernel
\[ \begin{split} \sigma ( E^0_+ ( z, h )   E^0_{-} ( z, h ) ) (x , \xi ) & = 
\sum_{\alpha  \in \mathbb Z_*^2 } \int_\RR w_\alpha ( x - \tfrac w 2) w_\alpha^* ( x - \tfrac w 2) 
e^{ \frac i h w \xi } dw \\
& = \sum_{\alpha  \in \mathbb Z_*^2 } \int_\RR e^{  \frac i h w ( \xi - \alpha_2) }  w_0 ( x - \tfrac w 2 - \alpha_1 ) { w_0 }( x + \tfrac w2 - \alpha_1 )^* dw.
\end{split} 
 \]
Hence, we obtain for the integral over the Weyl symbol
\begin{equation}
\begin{split} 
&  \int_{\mathbb{T}^2_*}  \sigma ( E^0_+ ( z, h ) 
   E^0_{-} ( z, h ) ) (x,\xi)\frac{dx d\xi }{ 4 \pi^2 } \\
 &=  \sum_\alpha 
   \int_{\mathbb{T}_*^2}  \int_\RR 
e^{  \frac i h w ( \xi - \alpha_2) }  w_0 ( x - \tfrac w 2 - \alpha_1 ) { w_0 }( x + \tfrac w2 - \alpha_1 )^*
d w \frac{dx d\xi }{ 4 \pi^2 } \\
&=  \int_{\RR^2} \int_\RR
e^{  \frac i h w  \xi }  w_0 ( x - \tfrac w 2  ) { w_0 }( x + \tfrac w 2 )^* 
d w \frac{dx d\xi }{ 4 \pi^2 }  \\
&= \ \frac{h}{ 2 \pi }  \int_\RR w_0(x) w_0(x)^* dx. 
\end{split} \end{equation}
This implies for $J_0$ as in \eqref{eq:J0}
\begin{equation}
\begin{split}
J_0(z,h)
&=\int_{\mathbb T_*^2} (z-\kappa(n_1h,h))^{-1} \operatorname{tr}_{\CC^m} \sigma \left( E_{+}^0 E_{-}^0 \right)(x,\xi) \frac{dx \ d\xi}{\vert \mathbb T_*^2 \vert} \\
&=\frac{h (z-\kappa(n_1h,h,p,q))^{-1}}{2\pi} \sum_{i=1}^{m} \sum_{j=1}^{n} \int_{\mathbb R } \left\vert \left\langle \widehat{e}_i, w_j(x) \right\rangle \right \vert^2   dx  \\
&=\frac{h (z-\kappa(n_1h,h,p,q))^{-1}}{2\pi} \sum_{j=1}^{n} \int_{\mathbb R } \left\vert w_j(x)  \right \vert^2   dx  \\
&= \frac{hn}{2\pi} (z-\kappa(n_1h,h,p,q))^{-1}.
\end{split}
\end{equation}
For the hexagonal lattice with magnetic flux $h$, the reflection symmetry of the Dirac points located at quasimomenta $\pm \left(\left(\frac{2\pi}{3}, -\frac{2\pi}{3} \right)\right)$ implies that the eigenfunctions 
$u^{\pm}_{n_1} = (u^{\pm}_{n_1,1},u^{\pm}_{n_1,2}) = (u^{\mp}_{n_1,2},u^{\mp}_{n_1,1}) $ satisfy 
\[\int_\RR \Vert w_0(x)^* \vec{e}_i \Vert^2 \ dx = \int_\RR   \vert u^{+}_{n_1,i}(x)\vert^2 +\vert u^{-}_{n_1,i}(x)\vert^2 \ dx=1 + \mathcal O(h^{\infty}).\]
Taking the regularized trace and exhibiting leading-order contributions shows that for $\vert \Im(z) \vert>h^M$, with arbitrary $M$, and $\vert z-z_1 \vert \le \varepsilon h$ there are analytic functions
\begin{equation}
\begin{split}
\label{eq:g}
g_{\scriptscriptstyle{ \blacksquare},n_1}(z,h) &:= G(z,h)+\tfrac{h}{2\pi}\sum_{n \neq n_1}(z-z_n(h))^{-1},\\
g_{\tinyvarhexagon,n_1}(i,z,h) &:= \langle \vec{e}_i,G(z,h)\vec{e}_i \rangle_{\CC^2} +\tfrac{h}{2\pi}\sum_{n \neq n_1}(z-z_n(h))^{-1},   \\ 
g_{\tinyvarhexagon,n_1}(z,h) &:=  \tfrac{ g_{\tinyvarhexagon,n_1}(1,z,h)+g_{\tinyvarhexagon,n_1}(2,z,h)}{2},\\
g_{\tinyvarhexagon,n_1}(z,h,p,q) &:= \operatorname{tr}_{\mathbb C^{2q}}G(z,h,p,q)+\tfrac{hn}{2\pi}\sum_{n \neq n_1}(z-z_n(h,p,q))^{-1},  
\end{split}
\end{equation}
such that we obtain \cite[Prop.\@ 6.1]{BZ} 
\begin{equation}
\begin{split}
\label{eq:leadorder}
\widetilde{\operatorname{tr}}\left((Q_{\scriptscriptstyle{ \blacksquare}}^{\rm{w}}(x,hp_x)-z)^{-1}\right)  &=  \tfrac{h}{2\pi} (z-z_{n_1,\scriptscriptstyle{ \blacksquare}}(h))^{-1}+g_{{\scriptscriptstyle{ \blacksquare}},n_1}(z,h)+\mathcal O(h^{\infty}), \\
\widetilde{\operatorname{tr}}\left(\langle \vec{e}_i,(Q_{\tinyvarhexagon}^{\rm{w}}(x,hp_x)-z)^{-1}\vec{e}_i \rangle_{\CC^2} \right) &= \tfrac{h}{2\pi} (z-z_{n_1,\tinyvarhexagon}(h))^{-1}+g_{\tinyvarhexagon,n_1}(i,z,h) +\mathcal O(h^{\infty}), \text{ and } \\ 
\widetilde{\operatorname{tr}}\left((M_{\tinyvarhexagon}^{\rm{w}}(x,hp_x)-z)^{-1} \right) &= \tfrac{hq^2}{\pi} (z-z_{n_1,\tinyvarhexagon}(h))^{-1}+g_{\tinyvarhexagon,n_1}(z,h,p,q) +\mathcal O(h^{\infty}).
\end{split}
\end{equation}
We also observe for later that
\begin{equation}
\begin{split}
\label{eq:squared}
&\left( \widetilde{\operatorname{tr}}(Q_{\scriptscriptstyle{ \blacksquare}}^{\rm{w}}(x,hp_x)-z)^{-1} \right)^2  = -\tfrac{h^2}{4\pi^2} D_z(z-z_{n_1,\scriptscriptstyle{ \blacksquare}}(h))^{-1} \\
&\qquad \qquad \qquad+\tfrac{h}{2\pi} (z-z_{n_1,\scriptscriptstyle{ \blacksquare}}(h))^{-1} g_{\scriptscriptstyle{ \blacksquare},n_1}(z,h)+ g_{\scriptscriptstyle{ \blacksquare},n_1}(z,h)^2+\mathcal O(h^{\infty}) \text{ and }\\
&\left(\widetilde{\operatorname{tr}}\langle \vec{e}_i,(Q_{\tinyvarhexagon}^{\rm{w}}(x,hp_x)-z)^{-1}\vec{e}_i \rangle_{\CC^2} \right)^2 = -\tfrac{h^2}{4\pi^2} D_z(z-z_{n_1,\tinyvarhexagon}(h))^{-1}\\
&\qquad \qquad \qquad+\tfrac{h}{2\pi}  (z-z_{n_1,\tinyvarhexagon}(h))^{-1} g_{\tinyvarhexagon,n_1}(i,z,h) + g_{\tinyvarhexagon,n_1}(i,z,h)^2+\mathcal O(h^{\infty}). 
\end{split}
\end{equation}

\smallsection{Step $3$: Trace formulae}

We can now assume that $\Re(z) \in (z_1-\varepsilon h,z_1+\varepsilon h)$ is close to a Landau level and apply \eqref{eq:leadorder}, as analyticity of the resolvent $(Q^{\rm{w}}(x,hp_x)-z)^{-1}$ away from the Landau bands implies that there is no contribution from $z$ outside these intervals (integration by parts in Helffer-Sj\"ostrand formula). 

\smallsection{Trace formulae in Thm.\@ \ref{theo1}}
From \eqref{eq:diff}, we have since $f \in C^5(I)$ that $D_{\overline{z}}\widetilde{f}(z) = \mathcal O\left(\Vert f \Vert_{C^{5}}\vert \Im(z) \vert^4 \right)$.
By  Proposition \ref{fP2Q},  we obtain, by writing the adjusted prefactors for the hexagonal lattice in parenthesis $\left[ \right]$ and for the square lattice without parenthesis,
\begin{equation}
\begin{split} 
&\widetilde{\operatorname{tr}}_{\Lambda}(f(H_{\lambda,\omega}^h)) 
  = \tfrac{ \left[2 \right]h} {2\pi^2 \vert \vec{b}_1 \wedge \vec{b}_2 \vert } 
 \int_{\CC }\sum_{k=0}^2 \frac{\lambda^k \mathbb E(V)^k D_{\overline{z}}\widetilde{f^{(k)}}(z)}{k!} 
 \sum_{ n } ( z - z_n(h))^{-1 }  \ dm ( z ) \\
&\ \ \ \ \ \ \ \ \ \ \ \ \ \  \ \ \ \ \ - \tfrac{\left[2 \right] h^2 \operatorname{Var}(V) \lambda^2}{8\pi^3| \vec{b}_1 \wedge \vec{b}_2 |} \sum_{ n  }  \int_{\mathbb{C}}D_{\overline{z}}\widetilde{f''}(z)    ( z -z_n(h) )^{-1 }    \ d m (z)\\
&\ \ \ \ \ \ \ \ \ \ \ \ \ \  \ \ \ \ \ - \tfrac{ \left[2 \right]h \operatorname{Var}(V) \lambda^2}{2\pi^2| \vec{b}_1 \wedge \vec{b}_2 |}\sum_{ n  }   \int_{\mathbb{C}}D_{\overline{z}}(\widetilde{f'}(z)g_n(z,h))    ( z -z_n(h) )^{-1 }    \ d m (z)\\
&\ \ \ \ \ \ \ \ \ \ \ \ \ \  \ \ \ \ \ + 
\tfrac 1 \pi \int_{ |\Im z | < h^M } D_{\overline{z}}\widetilde{f}(z) 
\mathcal O \left( |\Im z |^{-1} \right)  \  dm ( z )  + \mathcal O\left(\Vert f \Vert_{C^{5}}(\lambda^3+h^{\infty})\right)
\\
&\qquad \qquad  =  
 \tfrac{ \left[2 \right] h}{2\pi| \vec{b}_1 \wedge \vec{b}_2 |}  \sum_{ n } \sum_{k=0}^2 \frac{\lambda^k \mathbb E(V)^k}{k!}
f^{(k)} ( z_n ( h  ) ) + \mathcal O\left(\Vert f \Vert_{C^{5}}(\lambda^3+h^{3M} +h^{\infty})\right)\\
&\ \ \ \ \ \ \ \ \ \ \ \ \ \  \ \ \ \  - \tfrac{\left[2 \right] h  \operatorname{Var}(V) \lambda^2}{2 \pi| \vec{b}_1 \wedge \vec{b}_2 |}  \sum_{ n } 
\left(\tfrac{f'' ( z_n ( h  ) )}{4\pi} +  f' ( z_n ( h  ) ) g_n(z_n(h),h) \right)\\ 
&\qquad \qquad  =  
 \tfrac{\left[2 \right] h}{2\pi | \vec{b}_1 \wedge \vec{b}_2 |}  \sum_{ n  } 
f ( z_n ( h  )+\lambda \mathbb E(V) ) + \mathcal O\left(\Vert f \Vert_{C^{5}}(\lambda^3+h^{3M} +h^{\infty})\right)\\
&\ \ \ \ \ \ \ \ \ \ \ \ \ \  \ \ \ \  - \tfrac{\left[2 \right] h  \operatorname{Var}(V) \lambda^2}{2\pi| \vec{b}_1 \wedge \vec{b}_2 |}  \sum_{ n  } 
\left(\tfrac{f'' ( z_n ( h  ) )}{4\pi} + f' ( z_n ( h  ) )g_n(z_n(h),h) \right).
\end{split}
 \end{equation}
 By taking $M$ arbitrarily large the trace formulae \eqref{eq:tracef} and \eqref{eq:tracef2} of Theorem \ref{theo1} follow.
 
\smallsection{Trace formula in Thm.\@ \ref{theol}}
Since $f$ is now only assumed to be H\"older continuous, we require an additional approximation argument:

Let $\psi \in C_c^{\infty}((0,1))$ be a positive function with $\int_{\mathbb R} \psi(s) \ ds =1$ and define $\psi_{h}(s):=h^{-1}\psi(h^{-1}s)$ with $f_h:=f*\psi_{h^{M_0}}$.  
Moreover, we find $\left\lVert f-f*\psi_{h^{M_0}} \right\rVert_{L^{\infty}}\le \left\lVert f \right\rVert_{C^{\alpha}} h^{\alpha M_0}$ and since the interval $I$ can contain only $\mathcal O(h^{-1})$ many Landau levels, we have
\begin{equation}
\label{eq:est}
h  \sum_{ \vert n \vert \le C/h} \left\lvert f(z_n(h))-f_{h}(z_n(h))  \right\rvert =\mathcal O\left( \left\Vert f \right\Vert_{C^{\alpha}} h^{\alpha M_0}\right).
\end{equation}
We observe that by \eqref{eq:diff} we have 
\begin{equation}
\label{eq:est1}
\Vert D_{\overline{z}}\widetilde{f_h}(z)  \Vert_{L^{\infty}} \le \Vert f_h \Vert_{C^2} \Vert \vert \Im (z) \vert = \mathcal O( \left\lVert f \right\rVert_{L^{\infty}} h^{-2M_0}\vert \Im (z) \vert ). 
\end{equation}
We then use \eqref{eq:est1} and \eqref{eq:traceseq} for the hexagonal lattice to conclude that
\begin{equation}
\begin{split} 
&\widetilde{\operatorname{tr}}_{\Lambda}(f_h(H_{\lambda,\omega}^h)) 
  = \tfrac{q h} {\pi^2 \vert \vec{b}_1 \wedge \vec{b}_2 \vert } 
 \int_{\CC } D_{\overline{z}}\widetilde{f_h}(z) 
 \sum_{ n } ( z - z_n(h))^{-1 }  \ dm ( z ) \\
 &\ \ \ \ \ \ \ \ \ \ \ \ \ \  \ \ \ \ \ + 
\tfrac 1 \pi \int_{ |\Im z | < h^M } D_{\overline{z}}\widetilde{f_h}(z) 
\mathcal O \left( |\Im z |^{-1} \right)  \  dm ( z )  + \mathcal O\left(\Vert f_h \Vert_{L^{\infty}}h^{\infty}\right)
\\
&\qquad \qquad  =  
 \tfrac{ q h}{\pi| \vec{b}_1 \wedge \vec{b}_2 |}  \sum_{ n }  
f_h ( z_n ( h  ) ) +\mathcal O \left( \left\lVert f \right\rVert_{L^{\infty}} h^{M-2M_0} \right)+  \mathcal O\left(\Vert f \Vert_{L^{\infty}}h^{\infty}\right).
\end{split}
 \end{equation}
 Thus, we have from \eqref{eq:est} that
 \begin{equation}
\begin{split} 
&\widetilde{\operatorname{tr}}_{\Lambda}(f(H_{\lambda,\omega}^h))  =  
 \tfrac{ q h}{\pi| \vec{b}_1 \wedge \vec{b}_2 |}  \sum_{ n }  
f ( z_n ( h  ) ) +\mathcal O \left( \left\lVert f \right\rVert_{L^{\infty}} h^{M-2M_0}+  \left\lVert f \right\rVert_{C^{\alpha}} h^{\alpha M_0} \right)
\end{split}
 \end{equation}
 which by choosing $M=3M_0$ and $M_0$ arbitrarily large implies \eqref{eq:TF}.
 
\smallsection{Step $4$: QHE and mobility edges for the hexagonal lattice}

From \eqref{eq:Rieffel} we conclude that for any Fermi projection $P= \indic_{J}(H^a_{\tinyvarhexagon})$ such that $J \subset I$ with $\partial J$ located inside a spectral gap of $H^a_{\tinyvarhexagon}$ there are $\gamma_1,\gamma_2 \in \mathbb Z$ such that
\begin{equation}
\widetilde{\operatorname{tr}}_{\Lambda_{\tinyvarhexagon}}( P ) =\vert \vec{b}_1 \wedge \vec{b}_2 \vert^{-1} \left(  \gamma_1+\gamma_2 \left(\tfrac{p}{q}+\tfrac{h}{2\pi} \right)\right).
\end{equation}
The trace formula \eqref{eq:TF} on the other hand yields that 
\begin{equation}
\widetilde{\operatorname{tr}}_{\Lambda_{\tinyvarhexagon}}( P ) = \tfrac{hq }{\vert \vec{b}_1 \wedge \vec{b}_2 \vert\pi} \sum_{n \in \ZZ} \indic_{J}(z_n(h,p,q)) + \mathcal O(h^{\infty}).
\end{equation}
Comparing coefficients \eqref{eq:express1} implies that the Hall conductivity, when gauged to be zero at zero energy, is given by \eqref{eq:Hallcond} for sufficiently small $h$.
\end{proof}

\end{document}